\documentclass[11pt]{article}

\usepackage{cite}
\usepackage{graphicx}
\usepackage{amsmath}
\usepackage{amsthm}
\usepackage{amsfonts}
\usepackage{amssymb}
\usepackage{fullpage}
\usepackage{color}
\usepackage{latexsym}
\usepackage{algorithm}
\usepackage[noend]{algorithmic}
\usepackage{subfigure}
\usepackage{soul}
\usepackage{fixltx2e}

\newtheorem{theorem}{Theorem}[section]

\newtheorem{lemma}[theorem]{Lemma}
\newtheorem{claim}[theorem]{Claim}

\newcommand{\old}[1]{{{}}}

\def\segment#1{{\overline{#1}}}
\def\wedge#1{{\textsc{w}_{#1}}}
\def\orientation#1{{\theta(#1)}}

\def\leftray#1{{\mathrel{\vbox{\offinterlineskip\ialign{%
    \hfil##\hfil\cr
    $\leftarrow\,$\cr
    $\textsc{w}_#1$\cr}}}}}

\def\rightray#1{\mathrel{\vbox{\offinterlineskip\ialign{%
    \hfil##\hfil\cr
    $\rightarrow\,$\cr
    $\textsc{w}_#1$\cr}}}}

\MakeRobust{\leftray}
\MakeRobust{\rightray}

\def\thirdray#1{\widetilde{\textsc{w}}_{#1}}

\def\connected#1#2{\{{#1}\} \leftrightarrow \{{#2}\}}
\def\notconnected#1#2{\{{#1}\} \not\leftrightarrow \{{#2}\}}
\def\bisector#1{bis(\wedge{#1})}

\def\ra{R_1}
\def\rb{R_2}
\def\rc{R_3}
\def\rd{R_4}
\def\re{R_5}
\def\rf{R_6}

\def\UDG{{\textsc{udg\/(\textit{\small P})}}}

\begin{document}

\title{Bounded-Angle Spanning Tree: Modeling Networks with Angular Constraints\thanks{Work by R. Aschner was partially supported by the Lynn and William Frankel Center for Computer Sciences. Work by M. Katz was partially supported by grant 1045/10 from the Israel Science Foundation. Work by M. Katz and R. Aschner was partially supported by grant 2010074 from the United States -- Israel Binational Science Foundation.}
}

\author{Rom Aschner \ \ \  Matthew J. Katz
\\
\\
{\small Department of Computer Science, Ben-Gurion University, Israel} \\
{\small {\tt $\{$romas,matya$\}$@cs.bgu.ac.il}}
}


\maketitle


\begin{abstract} 
We introduce a new structure for a set of points in the plane and an angle $\alpha$, which is similar in flavor to a bounded-degree MST. We name this structure $\alpha$-MST. 
Let $P$ be a set of points in the plane and let $0 < \alpha \le 2\pi$ be an angle. An $\alpha$-ST of $P$ is a spanning tree of the complete Euclidean graph induced by $P$, with the additional property that for each point $p \in P$, the smallest angle around $p$ containing all the edges adjacent to $p$ is at most $\alpha$. An $\alpha$-MST of $P$ is then an $\alpha$-ST of $P$ of minimum weight.
For $\alpha < \pi/3$, an $\alpha$-ST does not always exist, and, for $\alpha \ge \pi/3$, it always exists~\cite{AGP13,AHHHPSSV13,CKLR11}. In this paper, we study the problem of computing an $\alpha$-MST for several common values of $\alpha$.

Motivated by wireless networks, 
we formulate the problem in terms of directional antennas. With each point $p \in P$, we associate a wedge $\wedge{p}$ of angle $\alpha$ and apex $p$.
The goal is to assign an orientation and a radius $r_p$ to each wedge $\wedge{p}$, such that the resulting graph is connected and its MST is an $\alpha$-MST. (We draw an edge between $p$ and $q$ if $p \in \wedge{q}$, $q \in \wedge{p}$, and $|pq| \le r_p, r_q$.)
Unsurprisingly, the problem of computing an $\alpha$-MST is NP-hard, at least for $\alpha=\pi$ and $\alpha=2\pi/3$. We present constant-factor approximation algorithms for $\alpha = \pi/2, 2\pi/3, \pi$.

One of our major results is a surprising theorem for $\alpha = 2\pi/3$, which, besides being interesting from a geometric point of view, has important applications. For example, the theorem guarantees that 
given {\em any} set $P$ of $3n$ points in the plane and {\em any} partitioning of the points into $n$ triplets, one can orient the wedges of each triplet {\em independently}, such that the graph induced by $P$ is connected. We apply the theorem to the {\em antenna conversion} problem.
\end{abstract}

\section {Introduction}

Let $P$ be a set of points in the plane and let $0 < \alpha \le 2\pi$ be an angle. An $\alpha$-ST of $P$ is a spanning tree of the complete Euclidean graph induced by $P$, with the additional property that for each point $p \in P$, the smallest angle around $p$ containing all the edges adjacent to $p$ is at most $\alpha$. An $\alpha$-MST of $P$ is then an $\alpha$-ST of $P$ of minimum weight.
 
In this paper, we study the problem of computing an $\alpha$-MST for several common values of $\alpha$. For $\alpha < \pi/3$, an $\alpha$-ST does not always exist (consider, e.g., an equilateral triangle). Moreover, it is well known that there always exists a Euclidean MST of degree at most 5.
Therefore, it is interesting to focus on the range $\pi/3 \leq \alpha < 8\pi/5$.

Carmi et al.~\cite{CKLR11} showed that, for $\alpha = \pi/3$, an $\alpha$-ST always exists. A somewhat simpler construction was subsequently proposed by Ackerman et al.~\cite{AGP13}. Aichholzer et al.~\cite{AHHHPSSV13} have also obtained this result (together with additional related results),   independently. However, in all these papers, the goal is to construct an $\alpha$-ST (for $\alpha = \pi/3$) and not an $\alpha$-MST. 

The problem of computing an $\alpha$-MST is similar in flavor to the problem of computing a Euclidean minimum weight degree-$k$ spanning tree, which has been studied extensively (see, e.g.,~\cite{A98, C04, JR09, KRY96, M99}). A minimum weight degree-$k$ spanning tree is a minimum weight spanning tree, such that the degree of each point is at most $k$, where the interesting values of $k$ are 2,3, and 4. Notice that for $k=2$ we get the Euclidean traveling salesman path problem. 

The problem of computing an $\alpha$-ST is closely related to problems in which one needs to compute a Hamiltonian path or cycle, with some restrictions on the angles. Fekete and Woeginger~\cite{FW97} showed that every set of points has a Hamiltonian {\em path}, such that all its angles are bounded by $\pi/2$. An alternative construction was given later in~\cite{CKLR11}. Fekete and Woeginger also conjectured that for every set of $2k \geq 8$ points there exists a Hamiltonian {\em cycle}, such that all its angles are bounded by $\pi/2$. Recently, Dumitrescu et al.~\cite{DPT12} showed how to construct a Hamiltonian cycle whose angles are bounded by $2\pi/3$.
As for lower bound, in~\cite{CKLR11} and, independently, in~\cite{DPT12} it is shown that, for any $\varepsilon>0$, there exists a set of points, for which any Hamiltonian path has an angle greater than $\pi/2 - \varepsilon$.
The problem of finding Hamiltonian paths with large angles was also considered in~\cite{FW97}, where it is conjectured that every point set admits a Hamiltonian path, whose angles are at least $\pi/6$; B{\'a}r{\'a}ny et al.~\cite{BPV09} showed how to construct a path, whose angles are at least $\pi/9$. 

Unsurprisingly, the problem of computing an $\alpha$-MST is NP-hard, at least for $\alpha=\pi$ and $\alpha=2\pi/3$. 
For $\alpha=\pi$, one can show this by a reduction from the problem of finding a Hamiltonian path in grid graphs of degree at most $3$, which is known to be NP-hard~\cite{IPS82}. The reduction is similar to the one described for the problem of computing a minimum weight degree-3 spanning tree~\cite{PV84}, with a few simple adaptations. 
For $\alpha=2\pi/3$, one can show this by a straight-forward reduction from Hamiltonian path in hexagonal grid graphs. Arkin et al.~\cite{AFIMMRPRX09} showed that the problem of finding a Hamiltonian cycle in hexagonal grid graphs is NP-hard. However, with not too much effort, one can prove that finding a Hamiltonian path in hexagonal grid graphs is NP-hard as well. 

Motivated by wireless networks, we formulate the problem of computing an $\alpha$-MST in terms of directional antennas. 
In the last few years, directional antennas have received considerable attention (see, e.g.,~\cite{KKM,BCDFKM11,CKKKW08}), as they have some noticeable advantages over omni-directional antennas. In particular, they require less energy to reach a receiver at a given distance, and
when broadcasting to this receiver the affected region is much smaller, reducing the probability of causing interference at friendly receivers or being subject to eves dropping by hostile receivers.
With each point $p \in P$, we associate a wedge $\wedge{p}$ of angle $\alpha$ and apex $p$. The goal now is to assign an orientation and a radius $r_p$ to each wedge $\wedge{p}$, such that the resulting graph is connected and its MST is an $\alpha$-MST. (We draw an edge between $p$ and $q$ if $p \in \wedge{q}$, $q \in \wedge{p}$, and $|pq| \le r_p, r_q$.) 

An interesting related problem is the {\em antenna conversion} problem. The {\em unit disk graph} of $P$, denoted $\UDG$, is the graph in which there is an edge between $p$ and $q$ if $|pq| \le 1$. This is the communication graph induced by $P$, where each point in $P$ represents a transceiver equipped with an omni-directional antenna of radius 1. We assume that $\UDG$ is connected. Suppose that one wishes to replace the omni-directional antennas with directional antennas of angle $\alpha$. The goal now is to assign an orientation to each of the wedges $\wedge{p}$ and to fix a common range $\delta=\delta(\alpha)$, such that the resulting (symmetric) communication graph is a $c$-hop-spanner of $\UDG$, where $c = c(\alpha)$. Moreover, $\delta$ and $c$ should be small constants.
Aschner et al.~\cite{AKM13} considered this problem for $\alpha=\pi/2$. Here we solve it for $\alpha=2\pi/3$, using significantly smaller constants.

\paragraph{Our results.}
In Section~\ref{sec:gadget} we focus on the case $\alpha = 2\pi/3$. We begin by describing a simple gadget: Given any set $S$ of three points in the plane, we show how to orient the wedges associated with the points of $S$, such that $G_S$, the graph induced by $S$, is connected, and, moreover, the union of the wedges of $S$ covers the plane. We then prove a surprising theorem, which, besides being interesting from a geometric point of view, has far-reaching applications, such as the one mentioned in the abstract. Informally, the theorem states that any two such gadgets are connected. That is, let $S_1$ and $S_2$ be two triplets of points in the plane, and assume that the wedges (associated with the points) of $S_1$ and, independently, of $S_2$ are oriented according to the gadget construction instructions, then the graph induced by $S_1 \cup S_2$ is connected.
Proving this theorem turned out to be a very challenging task, due to the huge number of possible configurations that must be considered, and only after arriving at the current three-stage proof structure (see Section~\ref{sec:main_theorem}), were we able to complete the proof.

In Section~\ref{sec:tsp_apx}, we present constant-factor approximation algorithms for computing an $\alpha$-MST. In particular, we compute a $2$-approximation for a $\pi$-MST, a $6$-approximation for a $2\pi/3$-MST, and a $16$-approximation for a $\pi/2$-MST. These approximations are actually with respect to a Euclidean MST, which is a lower bound for an $\alpha$-MST, for any $\alpha$.
In Section~\ref{sec:boundedrange}, we present a solution to the antenna conversion problem for $\alpha=2\pi/3$, based on the theorem above. Specifically, we construct, in $O(n\log n)$ time, a 6-hop-spanner of $\UDG$, in which each edge is of length at most 7.
Finally, NP-hardness proofs for the problem of computing an $\alpha$-MST, for $\alpha=\pi$ and $\alpha=2\pi/3$, can be found in Section~\ref{sec:np_hardness}.

\section{$\boldsymbol{\alpha = \frac{2\pi}{3}}$}\label{sec:gadget}

{\bf Notation.}
Let $p$ be a point and let $\alpha$ be an angle. We denote the wedge of angle $\alpha$ and apex $p$ by $\wedge{p}$. 
The left ray bounding $\wedge{p}$ (when looking from $p$ into $\wedge{p}$) is denoted by $\leftray{p}$ and the right ray by $\rightray{p}$. The bisector of $\wedge{p}$ is denoted by $\bisector{p}$.  
The orientations of $\leftray{p}$, $\rightray{p}$, and $\bisector{p}$ are denoted by $\orientation{\leftray{p}}$, $\orientation{\rightray{p}}$, and $\orientation{\bisector{p}}$, respectively.
The {\em orientation} of $\wedge{p}$ is the orientation of its bisector and is denoted by
$\orientation{\wedge{p}}$. We denote the ray emanating from $p$ of orientation $\orientation{\bisector{p}} + 180$ by $\thirdray{p}$; its orientation is denoted by $\orientation{\thirdray{p}}$. 

Let $S$ be a set of points, where each point $p \in S$ is associated with a wedge $\wedge{p}$ of some orientation.
The graph induced by $S$, denoted $G_S$, is the graph in which there is an edge between $p,q \in S$ if and only if $p \in \wedge{q}$ and $q \in \wedge{p}$. If there is an edge between $p$ and $q$, we say that $p$ and $q$ are {\em connected} and denote this by $\connected{p}{q}$. Similarly, if $S_1$ and $S_2$ are two such sets of points, and there exist a point $p$ in $S_1$ and a point $q$ in $S_2$ such that $p$ and $q$ are connected, then we say that $S_1$ and $S_2$ are {\em connected} and denote this by $\connected{S_1}{S_2}$. The notation $\notconnected{p}{q}$ means that $p$ and $q$ are not connected, and, similarly, $\notconnected{S_1}{S_2}$ means that there does not exist a point in $S_1$ and a point in $S_2$ such that these points are connected.

\subsection{The basic gadget}\label{sec:orient}

\begin{claim}\label{lem:three_pts}
Let $S=\{a,b,c\}$ be a set of three points in the plane, and set $\alpha=2\pi/3$. Then, one can orient the wedges
of $S$, such that $G_S$, the induced graph of $S$, contains a $2\pi/3$-ST of $S$, and the wedges of $S$ cover
the plane.
\end{claim}

\begin{proof}
Consider $\triangle abc$, and
assume w.l.o.g. that $\angle b \le \angle c \le \angle a$. Then, $\angle b \le 60$ and $\angle c < 90$. Draw $\triangle abc$, such that $\segment{bc}$ is horizontal (with $b$ to the left of $c$) and $a$ is not below the line containing $\segment{bc}$.
Orient the wedges of $S$ as follows (see Figure~\ref{fig:three_pts_fig}(a)): 
$\orientation{\wedge{a}}=240$,
$\orientation{\wedge{b}}=0$,
$\orientation{\wedge{c}}=120$. 

It is easy to see that the non-directed edges $(a,b)$ and $(b,c)$ are in the induced graph $G_S$.
Thus, $G_S$ contains a $2\pi/3$-ST. As for the second requirement, notice that 
$\wedge{a}$ contains the wedge $\wedge{a}'$ of orientation $\orientation{\wedge{a}}$ and apex $b$, and
$\wedge{c}$ contains the wedge $\wedge{c}'$ of orientation $\orientation{\wedge{c}}$ and apex $b$. But, clearly,
$\wedge{a}' \cup \wedge{b} \cup \wedge{c}' = \mathbb{R}$.  
\end{proof}


\begin{figure}[htb]
\centering
  \subfigure[]{
   \centering
       \includegraphics[width=0.35\textwidth,page=1]{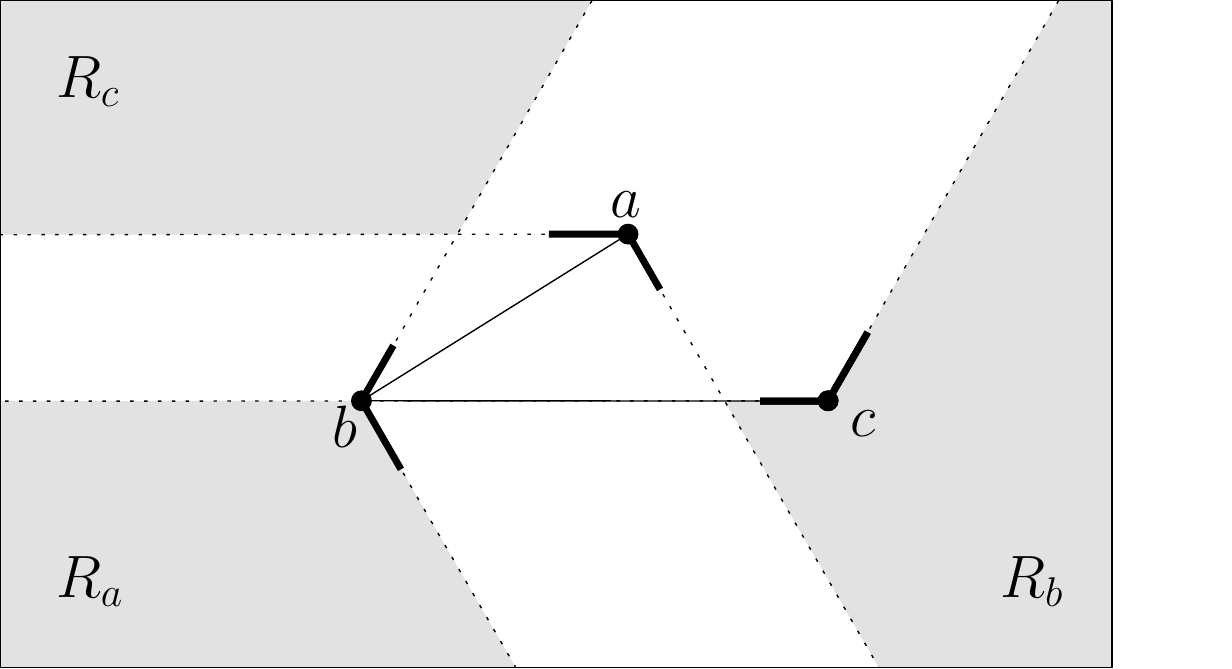}
  }
  \subfigure[]{
   \centering
       \includegraphics[width=0.35\textwidth,page=2]{fig/three_pts_regions}
  }
  \subfigure[]{
   \centering
       \includegraphics[height=0.19\textwidth,page=1]{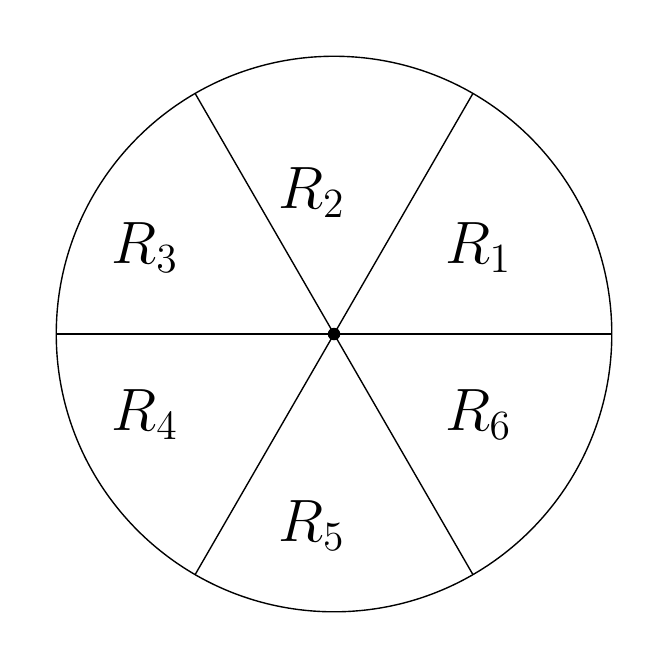}
  }
 	\caption{(a) The basic gadget of Claim~\ref{lem:three_pts}. $\orientation{\wedge{a}}=240$, $\orientation{\wedge{b}}=0$, and $\orientation{\wedge{c}}=120$. A point $p$ is in region $R_a$ if and only if $p \in \wedge{a}$ and $p \not \in \wedge{b}, \wedge{c}$, i.e., $R_a = \wedge{a} \setminus (\wedge{b} \cup \wedge{c})$. Regions $R_b$ and $R_c$ are defined analogously. (b) $\orientation{\leftray{a}}=\orientation{\rightray{b}}=\orientation{\thirdray{c}}=300$, $\orientation{\rightray{a}}=\orientation{\thirdray{b}}=\orientation{\leftray{c}}=180$, and $\orientation{\thirdray{a}}=\orientation{\leftray{b}}=\orientation{\rightray{c}}=60$. (c) The six ranges $R_1, \ldots, R_6$.}	\label{fig:three_pts_fig}	
\end{figure}

The gadget of Claim~\ref{lem:three_pts} has some noticeable properties:\\
{\bf Property 1.} For any $x \in S$, the orientations of the wedges of $S$ are $\orientation{\wedge{x}}$ and $\orientation{\wedge{x}} \pm 120$.\\ 
{\bf Property 2.} For any $x \in S$, the orientations of the rays bounding the wedges of $S$ are $\orientation{\wedge{x}} \pm 60$ and $\orientation{\wedge{x}}+180$. Moreover, each of these three orientations appears exactly twice, once as the orientation of a left ray bounding some wedge and once as the orientation of a right ray bounding some other wedge (see Figure~\ref{fig:three_pts_fig}(b)).\\
{\bf Property 3.} Consider any two wedges $\wedge{x}$ and $\wedge{y}$ and the four rays defining them. Then, by Property~2, exactly two of these rays, $\rho_1$ from $\wedge{x}$ and $\rho_2$ from $\wedge{y}$, have the same orientation. Let $l$ be a line intersecting both $\rho_1$ and $\rho_2$ and perpendicular to $\rho_1$ (and to $\rho_2$). Then, $\wedge{x} \cup \wedge{y}$ covers the halfplane defined by $l$ that does not include the points $x$ and $y$.

Finally, let $R_i$ denote the range $((i-1)60,i60)$, for $i = 1,\ldots,6$ (see Figure~\ref{fig:three_pts_fig}(c)).


\subsection{The induced graph of $\boldsymbol{S_1 \cup S_2}$ is connected}
\label{sec:main_theorem}

In this section, we prove the following surprising theorem (Theorem~\ref{thm:no_cliques}), which, as mentioned, has far-reaching applications.
Let $S_1=\{a,b,c\}$ and $S_2$ be two triplets of points in the plane, and assume that the wedges (associated with the points) of $S_1$ and, independently, of $S_2$ are oriented according to the proof of Claim~\ref{lem:three_pts}. Then, the induced graph of $S_1 \cup S_2$ is connected. 

In order to cope with the huge number of cases, we prove Theorem~\ref{thm:no_cliques} in three stages.
In the first stage (Lemma~\ref{lem:two_cliques}), we prove the statement assuming that both induced graphs of $S_1$ and of $S_2$ are cliques. In the second stage (Lemma~\ref{lem:one_clique}), we prove the statement assuming only one of the induced graphs is a clique, using, of course, Lemma~\ref{lem:two_cliques}. Finally, in the third stage (Theorem~\ref{thm:no_cliques}), we prove the statement without any additional assumptions, using Lemma~\ref{lem:one_clique}.

Throughout this section, we assume (as in the proof of Claim~\ref{lem:three_pts}) that, in $\Delta abc$, $\angle b \le \angle c \le \angle a$, $\overline{bc}$ is horizontal, with $b$ to the left of $c$, and $a$ is not below the line $l$ containing $\overline{bc}$ (see Figure~\ref{fig:three_pts_fig}(a)).     

\begin{lemma}[Two cliques]\label{lem:two_cliques}
Let $S_1=\{a,b,c\}$ and $S_2$ be two triplets of points in the plane and let $\alpha=2\pi/3$.
Assume that the wedges (associated with the points) of $S_1$ and, independently, of $S_2$ are oriented according to the proof of Claim~\ref{lem:three_pts}, and that both induced graphs, $G_{S_1}$ and $G_{S_2}$, are cliques.
Then, the induced graph $G_{S_1\cup S_2}$ is connected.
\end{lemma}

\old{
\begin{proof}
The wedges of $S_2$ cover the plane, in particular they cover all points of $S_1$. Therefore, we distinguish between three (not necessarily disjoint) cases: 
(i) there exists a point $x \in S_2$ such that $\wedge{x}$ covers all points of $S_1$, 
(ii) there exists a point $x \in S_2$ such that $\wedge{x}$ covers exactly two points of $S_1$, and 
(iii) the wedge of each point in $S_2$ covers exactly one point of $S_1$.

{\bf Case (i):} There exists a point $x \in S_2$ such that $\wedge{x}$ covers all points of $S_1$. Since
the wedges of $S_1$ cover the plane, at least one of them must cover $x$, and therefore
$\connected{x}{S_1}$.

{\bf Case (ii):} There exists a point $x \in S_2$ such that $\wedge{x}$ covers exactly two points of $S_1$. We divide this
case into three sub-cases, according to which two points of $S_1$ are covered by $\wedge{x}$.

\begin{figure}[htb]
 \centering
 \subfigure[$y\in R_a$]{
   \centering
       \includegraphics[width=0.45\textwidth]{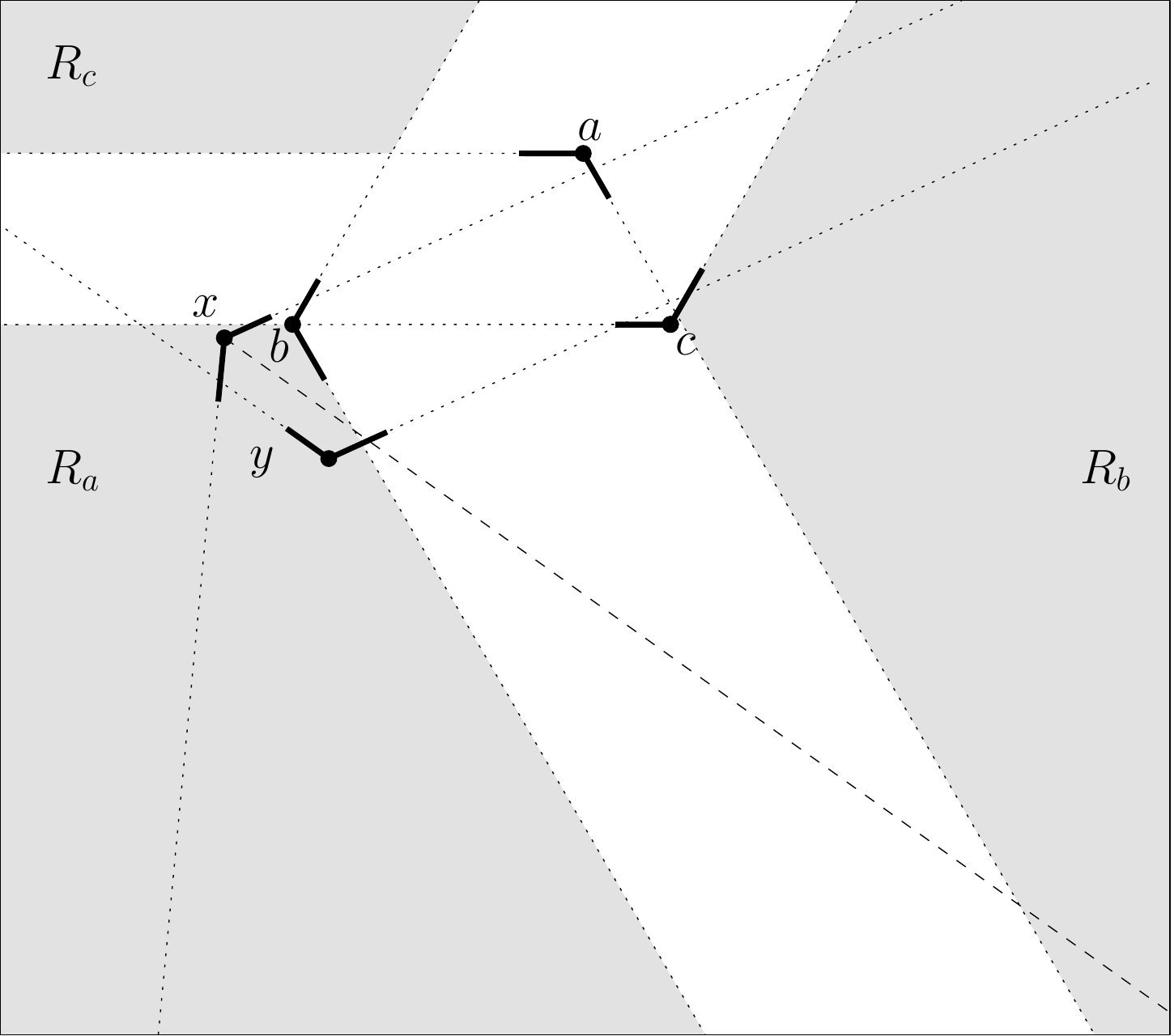}
 }
 \subfigure[$y \in \wedge{b}$]{
  \centering
       \includegraphics[width=0.45\textwidth]{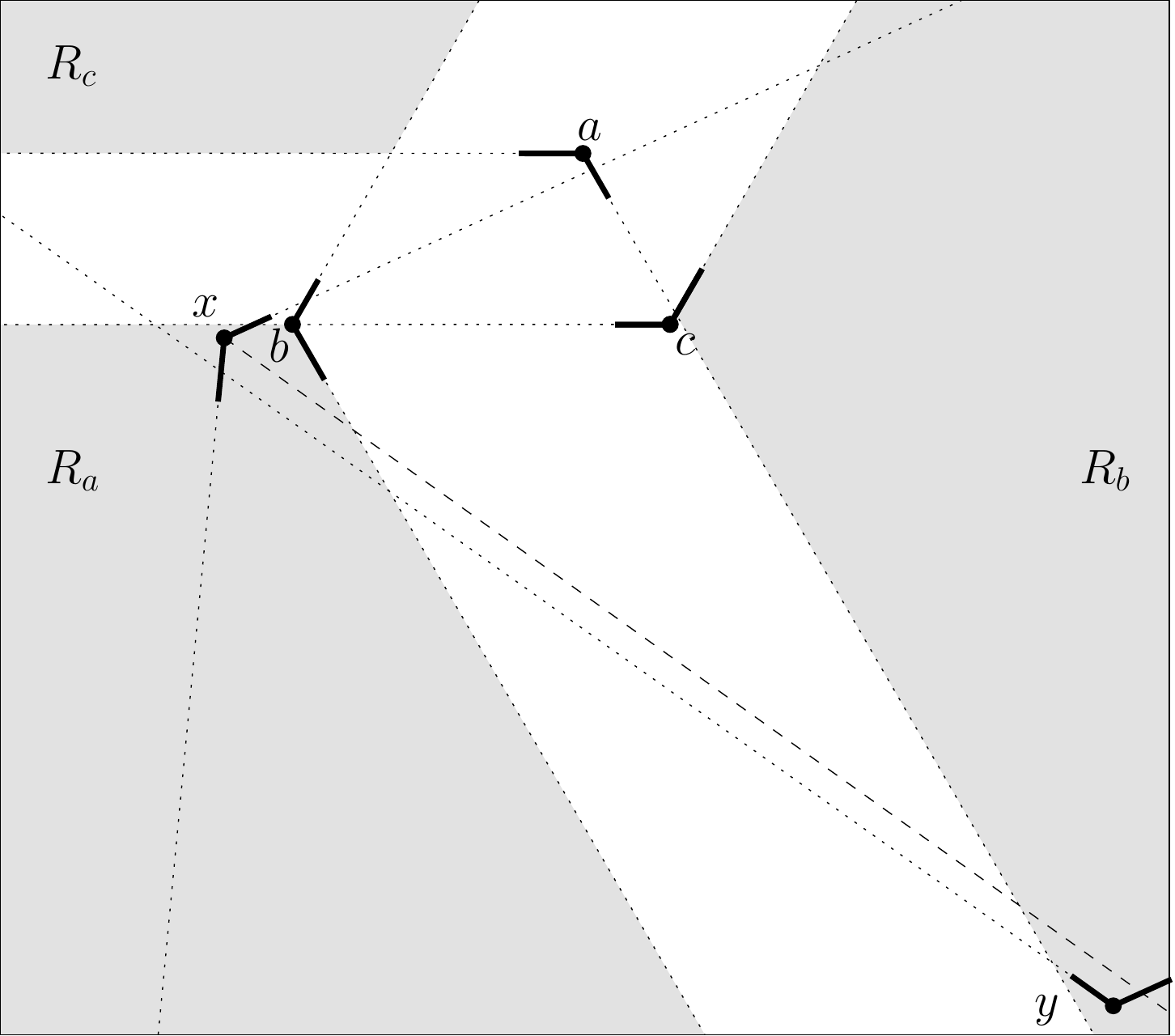}
 }
 \caption{Proof of Lemma~\ref{lem:two_cliques}, Case (ii)(1).}	\label{fig:case2b}	
\end{figure} 

(1) $\wedge{x}$ covers $b$ and $c$ and does not cover $a$. Assume $\notconnected{x}{b,c}$ (since otherwise
we are done), then $x \in R_a$ and one of the rays of $\wedge{x}$ intersects $\segment{ab}$ and
$\segment{ac}$. Notice that this ray must be $\leftray{x}$ and that $\leftray{x}$ also intersects $\leftray{b}$ (see~Figure~\ref{fig:case2b}).
Since $x$ lies below $l$, $\leftray{x}$ intersects $\leftray{b}$, and $\orientation{\leftray{b}}=60$, we have that $\orientation{\leftray{x}} \in \ra$. 
It follows that $\orientation{\rightray{x}} \in \re$, $\orientation{\wedge{x}} \in \rf$, and $\orientation{\thirdray{x}} \in \rc$.
Therefore, $\bisector{x}$ (whose orientation is $\orientation{\wedge{x}}$) does not intersect $l$.
Let $y$ be the point of $S_2$ such that $\orientation{\leftray{y}}= \orientation{\thirdray{x}} \in \rc$ and 
$\orientation{\rightray{y}} = \orientation{\leftray{x}} \in \ra$. Since $\connected{x}{y}$, we have that
$y \in \wedge{x}$ and $y$ lies to the right of $\bisector{x}$. Notice that $\wedge{y}$ contains the (imaginary) wedge of orientation $\orientation{\wedge{y}}$ and apex $x$. If $y \in R_a$ (see Figure~\ref{fig:case2b}(a)), then $\connected{y}{a}$, since $\wedge{y}$ covers $a$. 
Otherwise, $y \in \wedge{b}$ and in particular $y$ lies to the right of $b$ (see Figure~\ref{fig:case2b}(b)). In this case we show that $\connected{y}{b}$. Indeed, $\rightray{y}$ intersects $l$ to the right of $b$, since $\orientation{\rightray{y}}\in\ra$, and, since $l(\leftray{y})$ is parallel to $l(\bisector{x})$ and below it, we have that $\leftray{y}$ intersects $l$ to the left of $b$. We conclude that $b \in \wedge{y}$ and $\connected{y}{b}$.

\begin{figure}[htb]
 \centering
       \includegraphics[width=0.35\textwidth]{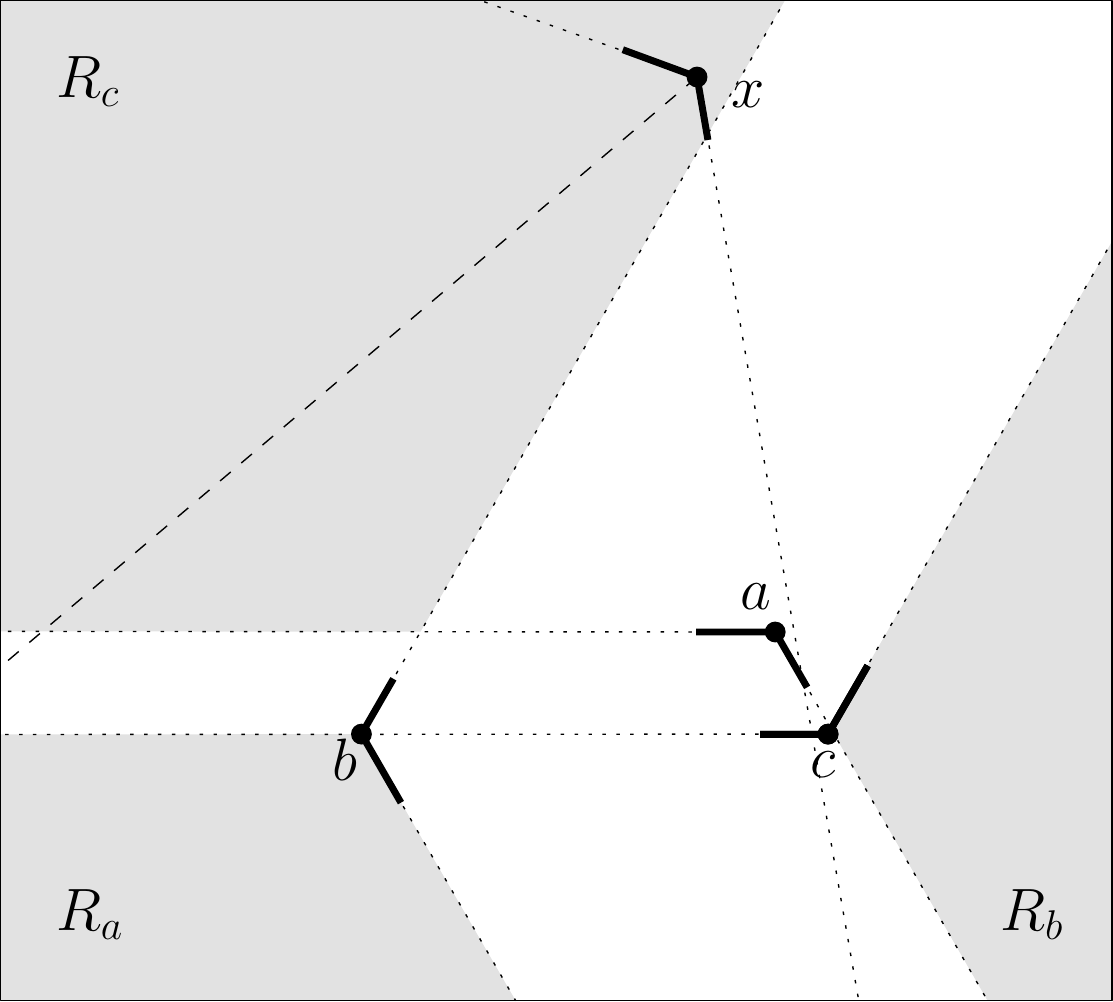}
 \caption{Proof of Lemma~\ref{lem:two_cliques}, Case (ii)(2).}	\label{fig:case2c}	
\end{figure} 
  
(2) $\wedge{x}$ covers $a$ and $b$ and does not cover $c$. This case is similar to Case~(ii)(1). Assume $\notconnected{x}{a,b}$ (since otherwise
we are done), then $x \in R_c$
and one of the rays of $\wedge{x}$ intersects $\segment{ac}$ and $\segment{bc}$. Notice that this ray must be $\leftray{x}$ and that $\leftray{x}$ also intersects $\leftray{a}$ (see Figure~\ref{fig:case2c}).
Since $x$ lies above $l$, $\leftray{x}$ intersects $\leftray{a}$, and $\orientation{\leftray{a}}=300$, we have that $\orientation{\leftray{x}} \in \re$. 
It follows that $\orientation{\rightray{x}} \in \rc$, $\orientation{\wedge{x}} \in \rd$, and $\orientation{\thirdray{x}} \in \ra$. Therefore, $\bisector{x}$ does not intersect $\leftray{b}$.
Let $y$ be the point of $S_2$ such that $\orientation{\leftray{y}}= \orientation{\thirdray{x}} \in \ra$ and
$\orientation{\rightray{y}} = \orientation{\leftray{x}} \in \re$. 
Since $\connected{x}{y}$, we have that $y \in \wedge{x}$ and $y$ lies to the right of $\bisector{x}$. Notice that $\wedge{y}$ contains the (imaginary) wedge of orientation $\orientation{\wedge{y}}$ and apex $x$. If $y \in \wedge{c}$, then $\connected{y}{c}$, since $\wedge{y}$ covers $c$. 
Otherwise, $y \in R_a$ and in particular $y$ lies to the left of $a$. In this case we show that $\connected{y}{a}$. Indeed, $\rightray{y}$ does not intersect $l$, since $y$ lies below $l$ and $\orientation{\rightray{y}} \in \re$, and, since $l(\leftray{y})$ is parallel to $l(\bisector{x})$ and above it, we have that $\leftray{y}$ intersects $\rightray{a}$. We conclude that $a \in \wedge{y}$ and $\connected{y}{a}$.

\begin{figure}[htb]
 \centering 
  \subfigure[$\leftray{x}$ intersects $\segment{ab}$ and $\segment{bc}$]{
    \centering
        \includegraphics[width=0.45\textwidth]{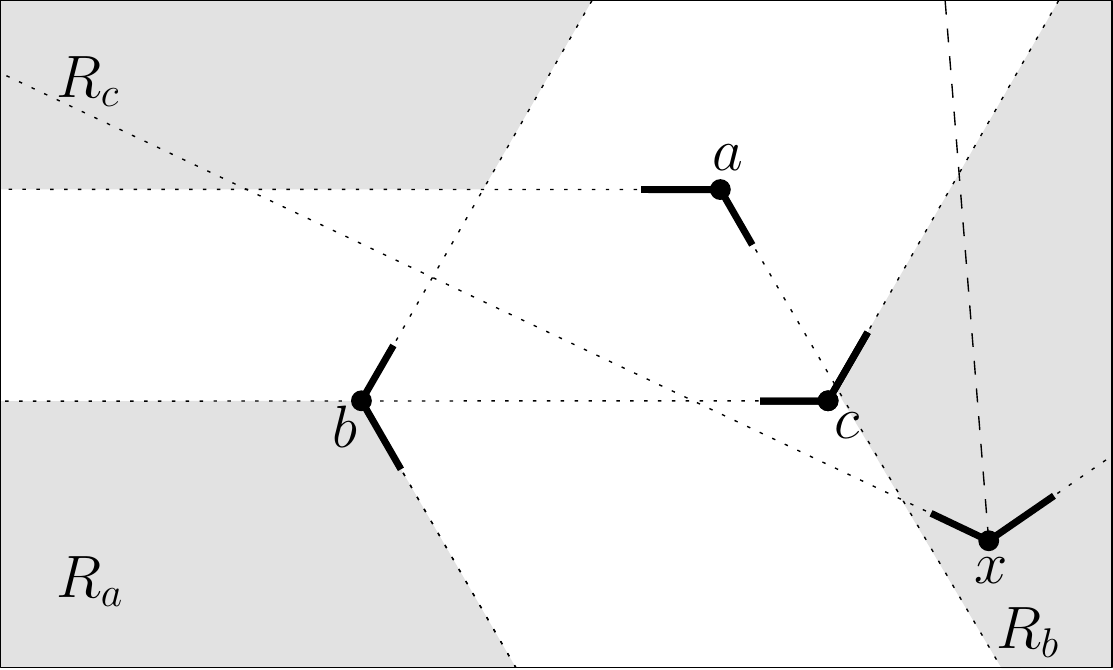}
	\label{fig:case2d}
 }
  \subfigure[$\rightray{x}$ intersects $\segment{ab}$ and $\segment{bc}$]{
    \centering
        \includegraphics[width=0.45\textwidth]{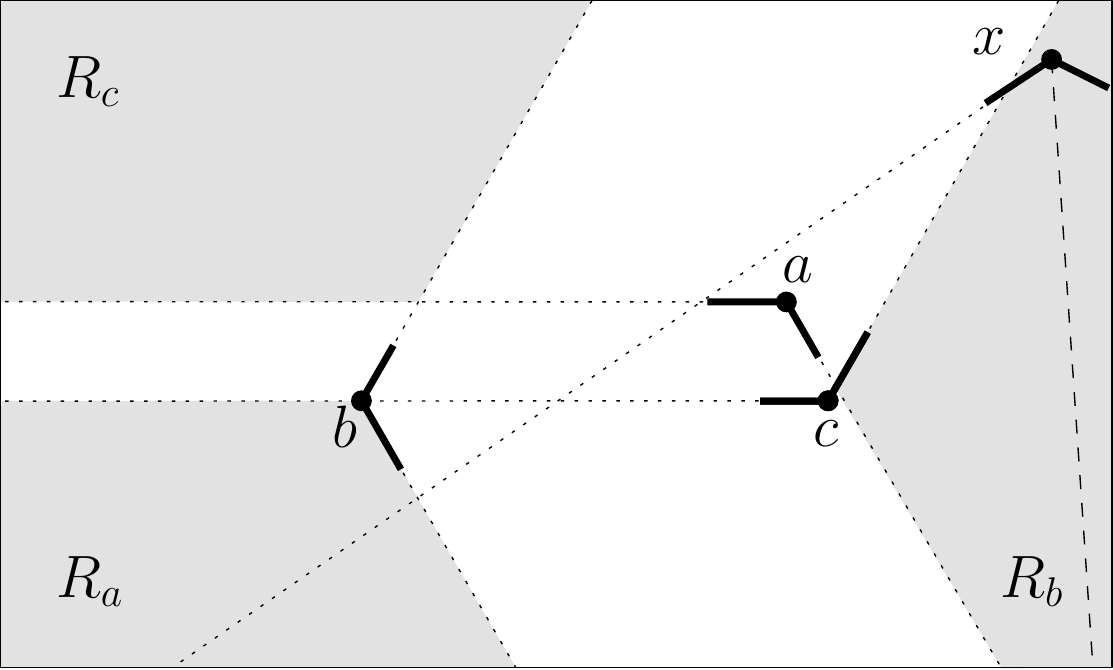}
	\label{fig:case2e}
 }
 \caption{Proof of Lemma~\ref{lem:two_cliques}, Case (ii)(3).}	\label{fig:caseb}
\end{figure}
 
(3) $\wedge{x}$ covers $a$ and $c$ and does not cover $b$. Assume $\notconnected{x}{a,c}$ (since otherwise
we are done), then $x \in R_b$,
and one of the rays of $\wedge{x}$ intersects $\segment{ab}$ and $\segment{bc}$. Notice that this ray 
can be either $\leftray{x}$ or $\rightray{x}$. 

If it is $\leftray{x}$ (see Figure~\ref{fig:case2d}), then
the orientations associated with $\wedge{x}$ are:
$\orientation{\leftray{x}} \in \rc$, $\orientation{\rightray{x}} \in \ra$, and $\orientation{\thirdray{x}} \in \re$.
Notice that $\bisector{x}$ does not intersect $\leftray{a}$.
Let $y$ be the point of $S_2$ such that $\orientation{\leftray{y}}= \orientation{\thirdray{x}}$ and $\orientation{\rightray{y}} =
\orientation{\leftray{x}}$. Since $\connected{x}{y}$, we have that $y \in \wedge{x}$ and $y$ lies to the right of $\bisector{x}$. 
Notice that $\wedge{y}$ contains the (imaginary) wedge of orientation $\orientation{\wedge{y}}$ and apex $x$. If $y \in \wedge{b}$, then $\connected{y}{b}$, since $y$ covers $b$. 
Otherwise, $y \in R_c$. Since $\leftray{y}$ passes to the right of $x$ (or through $x$) and $\orientation{\rightray{y}} \in \rc$, we have that $c \in \wedge{y}$ and therefore $\connected{y}{c}$.

If the ray intersecting $\segment{ab}$ and $\segment{bc}$ is $\rightray{x}$ (see Figure~\ref{fig:case2e}), then
the orientations associated with $\wedge{x}$ are:
$\orientation{\leftray{x}} \in \rf$, $\orientation{\rightray{x}} \in \rd$, and $\orientation{\thirdray{x}} \in \rb$.
Notice that $\bisector{x}$ does not intersect $\rightray{c}$.
Let $y$ be the point of $S_2$ such that $\orientation{\rightray{y}}=
\orientation{\thirdray{x}}$ and $\orientation{\leftray{y}} = \orientation{\rightray{x}}$. 
Since $\connected{x}{y}$, we have that $y \in \wedge{x}$ and $y$ lies to the left of $\bisector{x}$. 
Notice that $\wedge{y}$ contains the (imaginary) wedge of orientation $\orientation{\wedge{y}}$ and apex $x$. If $y \in \wedge{b}$, then $\connected{y}{b}$, since $y$ covers $b$. 
Otherwise, $y \in R_a$. Since $\rightray{y}$ passes to the right of $x$ (or through $x$) and $\orientation{\leftray{y}} \in \rd$, we have that $a \in \wedge{y}$ and therefore $\connected{y}{a}$.

\begin{figure}[htb]
 \centering 
    \includegraphics[width=0.5\textwidth]{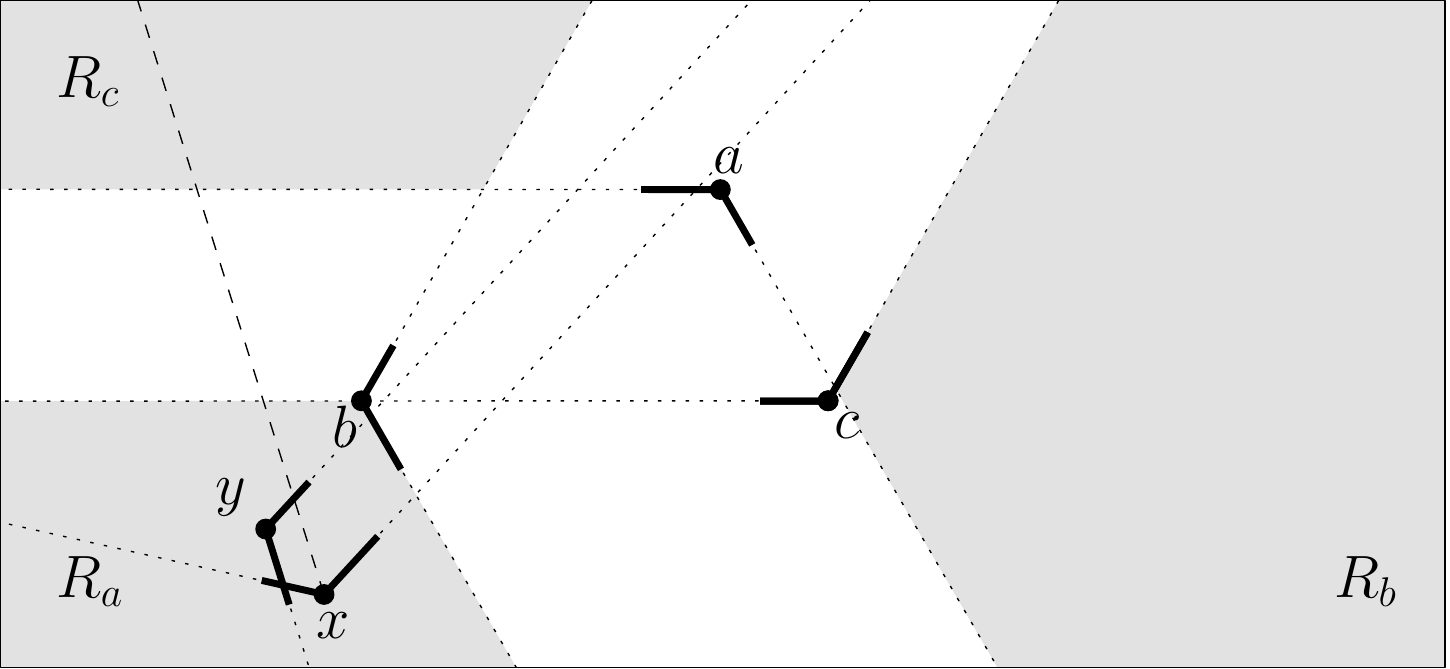}
 \label{fig:case3a}
 \caption{Proof of Lemma~\ref{lem:two_cliques}, Case (iii).}	\label{fig:case3}
\end{figure}

{\bf Case (iii):} The wedge of each point in $S_2$ covers exactly one point of $S_1$. We may assume that this condition also holds for the wedges of $S_1$; that is, the wedge of each point in $S_1$ covers exactly one point of $S_2$. Since, otherwise, we can simply interchange the set names. 
It follows that each point of $S_2$ lies in its own private region among the regions $R_a$, $R_b$, and $R_c$. 

Let $x$ be the point that lies in $R_a$. We claim that $\connected{x}{a}$.
Assume that $\notconnected{x}{a}$. We show that there exists a point $y \in S_2$ that covers two points of $S_1$.
If $\wedge{x}$ covers $b$ (see Figure~\ref{fig:case3}), then $\orientation{\rightray{x}} \in (0,120)$, which implies that $\orientation{\thirdray{x}} \in (240,360)$. Let $y$ be the point of $S_2$ such that $\orientation{\rightray{y}}=\orientation{\thirdray{x}}$ and $\orientation{\leftray{y}}=\orientation{\rightray{x}}$. Since $\connected{y}{x}$, we have that $y \in \wedge{x}$ and $y$ lies to the left of $\bisector{x}$, but then $\wedge{y}$ must cover $a$ and $c$ -- contradiction.
If $\wedge{x}$ covers $c$, then $\orientation{\leftray{x}} \in (0,120)$, which implies that $\orientation{\thirdray{x}} \in (120,240)$. Let $y$ be the point of $S_2$ such that $\orientation{\leftray{y}}=\orientation{\thirdray{x}}$ and $\orientation{\rightray{y}}=\orientation{\leftray{x}}$. Since $\connected{y}{x}$, we have that $y \in \wedge{x}$ and $y$ lies to the right of $\bisector{x}$, but then $\wedge{y}$ must cover $a$ and $b$ -- contradiction.

\end{proof}
}

\begin{proof}
The wedges of $S_2$ cover the plane, in particular they cover all points of $S_1$. Therefore, we distinguish between three (not necessarily disjoint) cases: 
(i) there exists a point $x \in S_2$ such that $\wedge{x}$ covers all points of $S_1$, 
(ii) there exists a point $x \in S_2$ such that $\wedge{x}$ covers exactly two points of $S_1$, and 
(iii) the wedge of each point in $S_2$ covers exactly one point of $S_1$.

{\bf Case (i):} There exists a point $x \in S_2$ such that $\wedge{x}$ covers all points of $S_1$. Since
the wedges of $S_1$ cover the plane, at least one of them must cover $x$, and therefore
$\connected{x}{S_1}$.

{\bf Case (ii):} There exists a point $x \in S_2$ such that $\wedge{x}$ covers exactly two points of $S_1$. We divide this
case into three sub-cases, according to which two points of $S_1$ are covered by $\wedge{x}$.

\begin{figure}[htb]
 \centering
 \subfigure[$y\in R_a$]{
   \centering
       \includegraphics[width=0.45\textwidth]{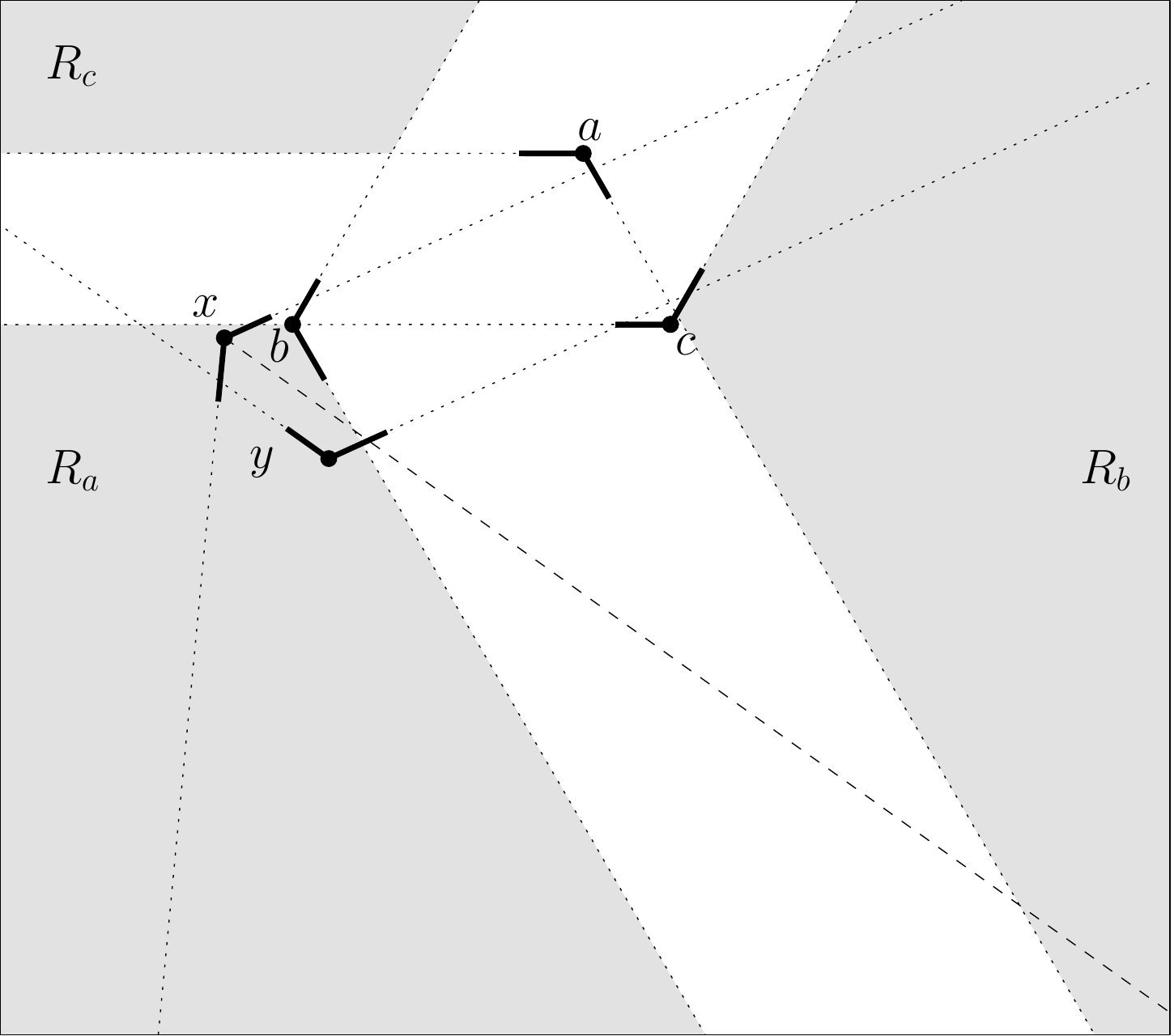}
 }
 \subfigure[$y \in \wedge{b}$]{
  \centering
       \includegraphics[width=0.45\textwidth]{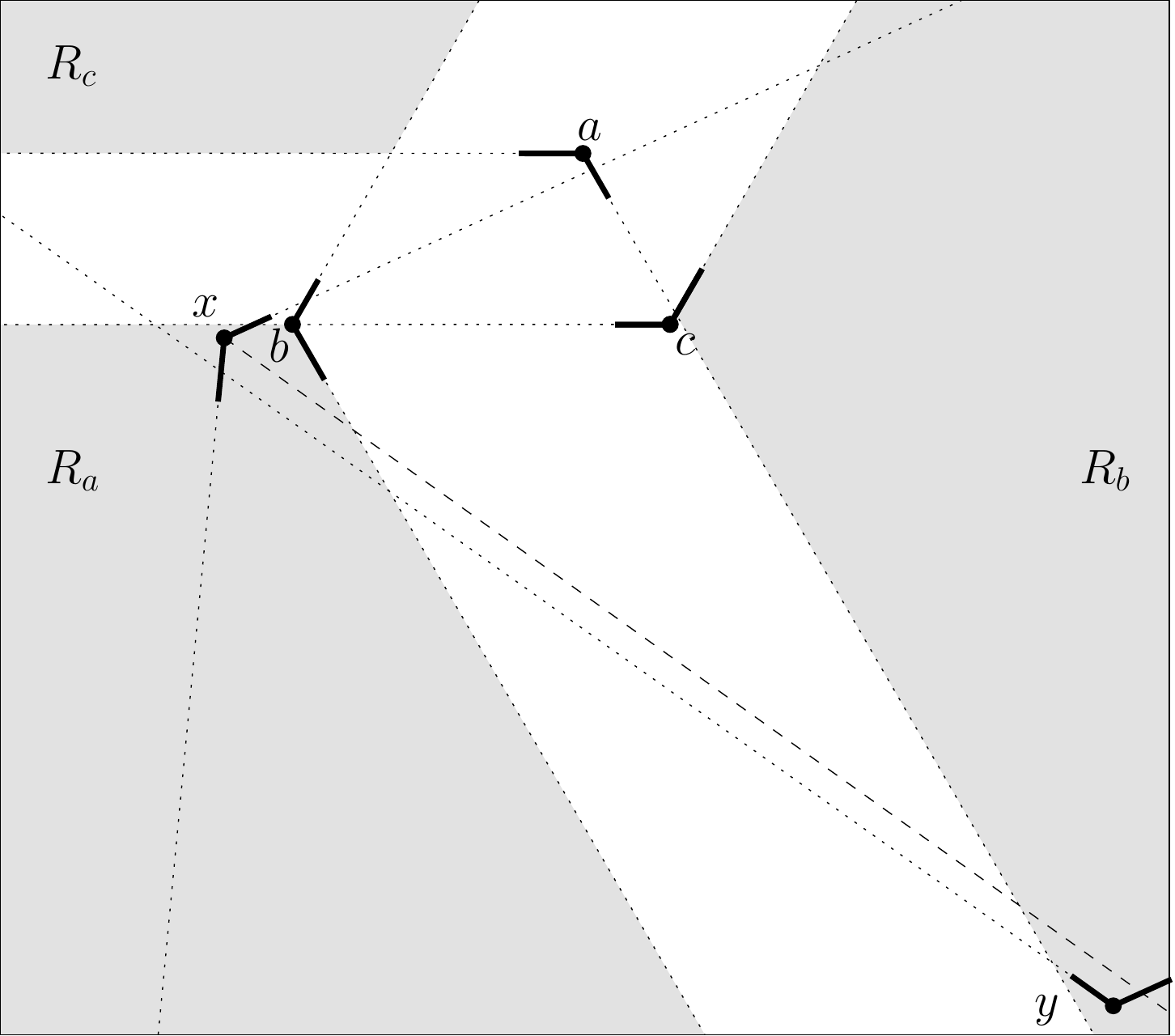}
 }
 \caption{Proof of Lemma~\ref{lem:two_cliques}, Case (ii)(1).}	\label{fig:case2b}	
\end{figure} 

(1) $\wedge{x}$ covers $b$ and $c$ and does not cover $a$. Assume $\notconnected{x}{b,c}$ (since otherwise
we are done), then $x \in R_a$ and one of the rays of $\wedge{x}$ intersects $\segment{ab}$ and
$\segment{ac}$. Notice that this ray must be $\leftray{x}$ and that $\leftray{x}$ also intersects $\leftray{b}$ (see~Figure~\ref{fig:case2b}).
Since $x$ lies below $l$, $\leftray{x}$ intersects $\leftray{b}$, and $\orientation{\leftray{b}}=60$, we have that $\orientation{\leftray{x}} \in \ra$. 
It follows that $\orientation{\rightray{x}} \in \re$, $\orientation{\wedge{x}} \in \rf$, and $\orientation{\thirdray{x}} \in \rc$.
Therefore, $\bisector{x}$ (whose orientation is $\orientation{\wedge{x}}$) does not intersect $l$.
Let $y$ be the point of $S_2$ such that $\orientation{\leftray{y}}= \orientation{\thirdray{x}} \in \rc$ and 
$\orientation{\rightray{y}} = \orientation{\leftray{x}} \in \ra$. Since $\connected{x}{y}$, we have that
$y \in \wedge{x}$ and $y$ lies to the right of $\bisector{x}$. Notice that $\wedge{y}$ contains the (imaginary) wedge of orientation $\orientation{\wedge{y}}$ and apex $x$. If $y \in R_a$ (see Figure~\ref{fig:case2b}(a)), then $\connected{y}{a}$, since $\wedge{y}$ covers $a$. 
Otherwise, $y \in \wedge{b}$ and in particular $y$ lies to the right of $b$ (see Figure~\ref{fig:case2b}(b)). In this case we show that $\connected{y}{b}$. Indeed, $\rightray{y}$ intersects $l$ to the right of $b$, since $\orientation{\rightray{y}}\in\ra$, and, since $l(\leftray{y})$ is parallel to $l(\bisector{x})$ and below it, we have that $\leftray{y}$ intersects $l$ to the left of $b$. We conclude that $b \in \wedge{y}$ and $\connected{y}{b}$.

\begin{figure}[htb]
 \centering
       \includegraphics[width=0.35\textwidth]{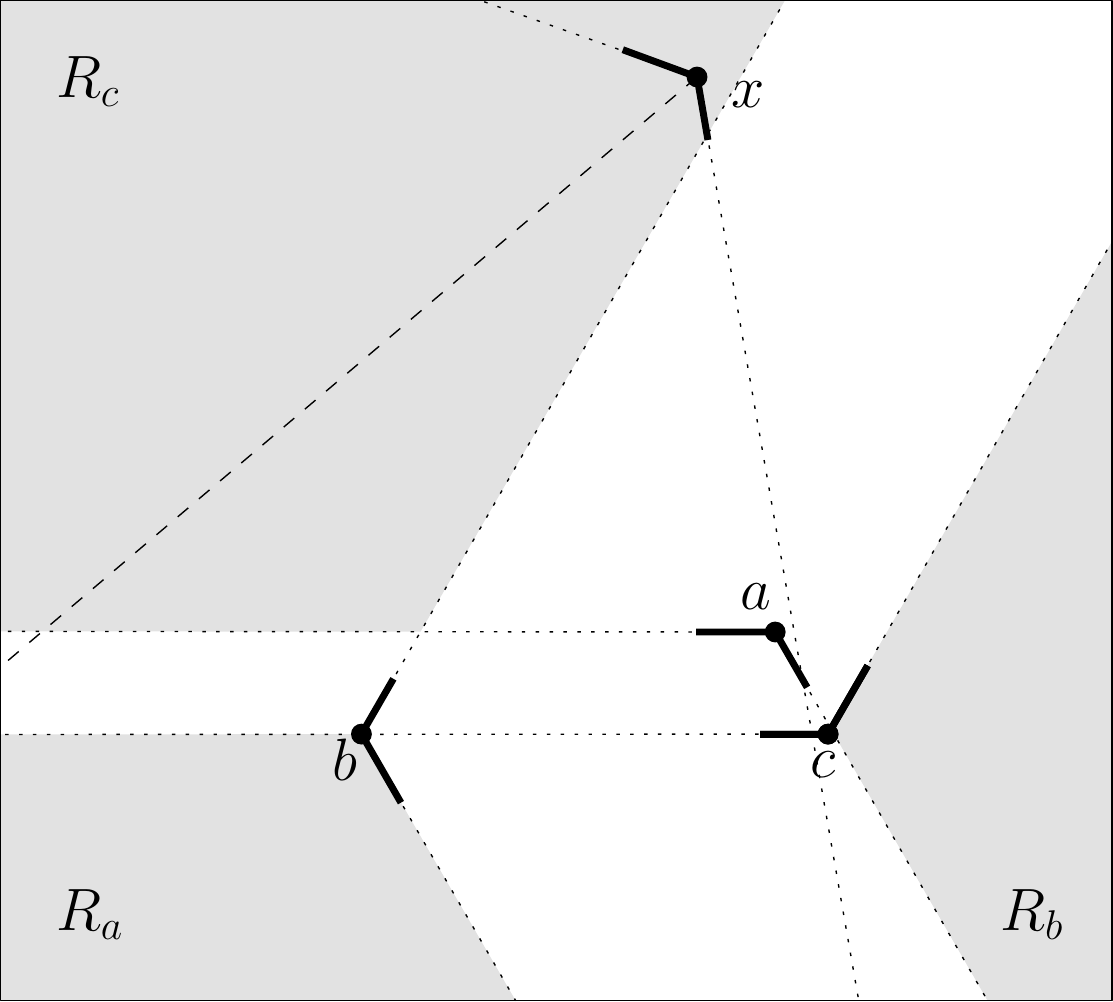}
 \caption{Proof of Lemma~\ref{lem:two_cliques}, Case (ii)(2).}	\label{fig:case2c}	
\end{figure} 
  
(2) $\wedge{x}$ covers $a$ and $b$ and does not cover $c$. Assume $\notconnected{x}{a,b}$ (since otherwise
we are done), then $x \in R_c$
and one of the rays of $\wedge{x}$ intersects $\segment{ac}$ and $\segment{bc}$. Notice that this ray must be $\leftray{x}$ and that $\leftray{x}$ also intersects $\leftray{a}$ (see Figure~\ref{fig:case2c}).
Since $x$ lies above $l$, $\leftray{x}$ intersects $\leftray{a}$, and $\orientation{\leftray{a}}=300$, we have that $\orientation{\leftray{x}} \in \re$. 
It follows that $\orientation{\rightray{x}} \in \rc$, $\orientation{\wedge{x}} \in \rd$, and $\orientation{\thirdray{x}} \in \ra$. The rest of the proof for this case is very similar to the proof of Case~(ii)(1), thus we omit further details.

\begin{figure}[htb]
 \centering 
  \subfigure[$\leftray{x}$ intersects $\segment{ab}$ and $\segment{bc}$]{
    \centering
        \includegraphics[width=0.45\textwidth]{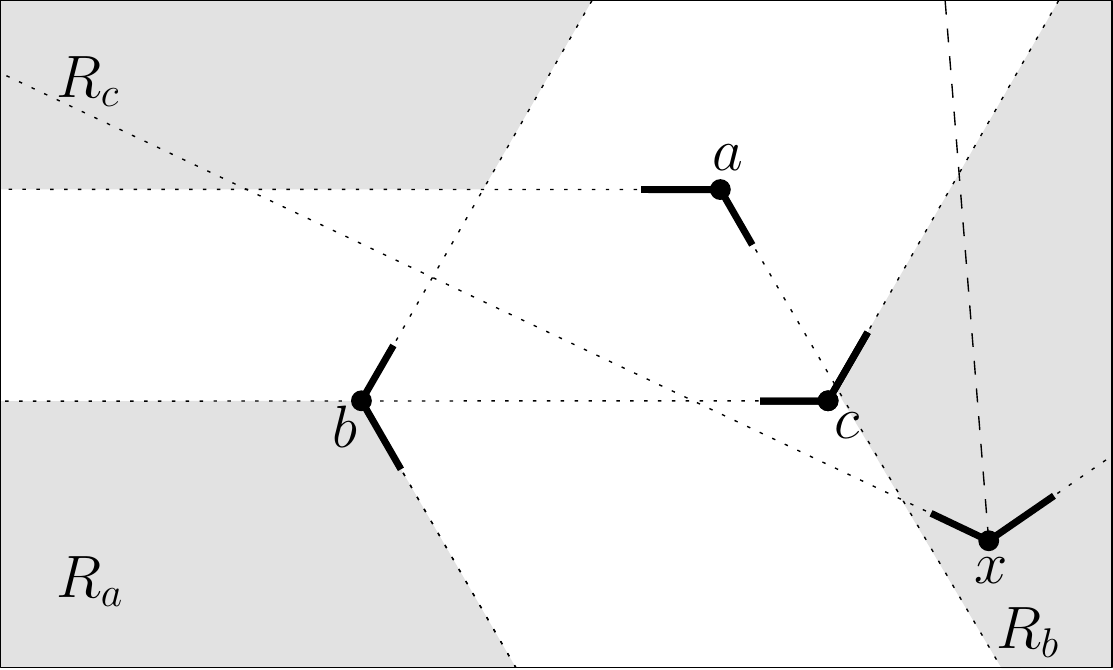}
	\label{fig:case2d}
 }
  \subfigure[$\rightray{x}$ intersects $\segment{ab}$ and $\segment{bc}$]{
    \centering
        \includegraphics[width=0.45\textwidth]{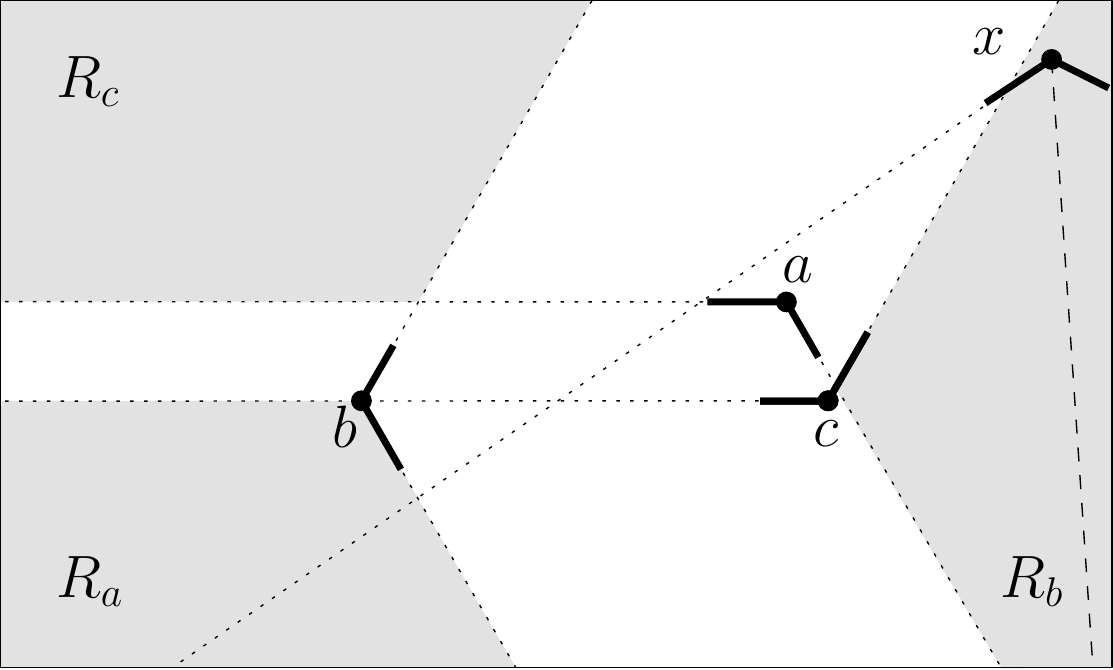}
	\label{fig:case2e}
 }
 \caption{Proof of Lemma~\ref{lem:two_cliques}, Case (ii)(3).}	\label{fig:caseb}
\end{figure}

(3) $\wedge{x}$ covers $a$ and $c$ and does not cover $b$. Assume $\notconnected{x}{a,c}$ (since otherwise
we are done), then $x \in R_b$,
and one of the rays of $\wedge{x}$ intersects $\segment{ab}$ and $\segment{bc}$. Notice that this ray 
can be either $\leftray{x}$ or $\rightray{x}$. 

If it is $\leftray{x}$ (see Figure~\ref{fig:case2d}), then
the orientations associated with $\wedge{x}$ are:
$\orientation{\leftray{x}} \in \rc$, $\orientation{\rightray{x}} \in \ra$, and $\orientation{\thirdray{x}} \in \re$.
The rest of the proof for this branch is very similar to the proof of Case (ii)(1), thus we omit further details.

If the ray intersecting $\segment{ab}$ and $\segment{bc}$ is $\rightray{x}$ (see Figure~\ref{fig:case2e}), then
the orientations associated with $\wedge{x}$ are:
$\orientation{\leftray{x}} \in \rf$, $\orientation{\rightray{x}} \in \rd$, and $\orientation{\thirdray{x}} \in \rb$.
Again, the rest of the proof for this branch is very similar to the proof of Case (ii)(1), thus we omit further details.

\begin{figure}[htb]
 \centering 
    \includegraphics[width=0.5\textwidth]{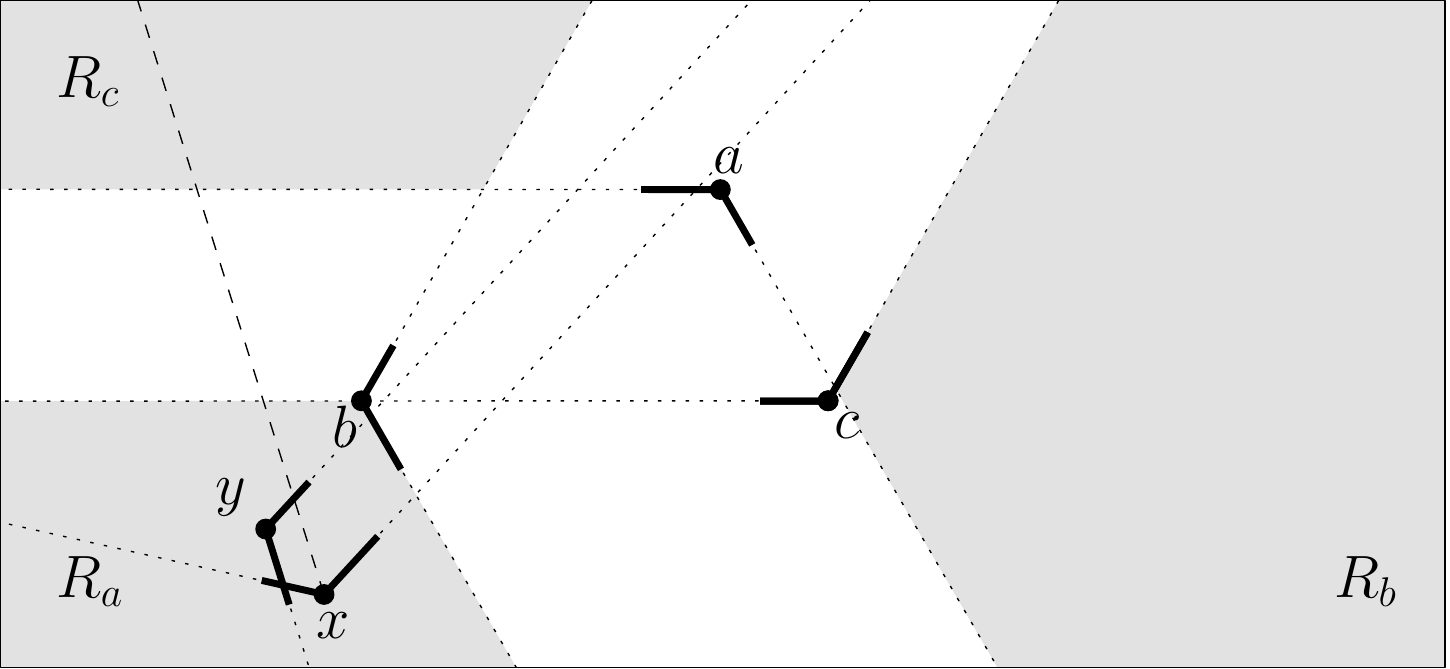}
 \label{fig:case3a}
 \caption{Proof of Lemma~\ref{lem:two_cliques}, Case (iii).}	\label{fig:case3}
\end{figure}

{\bf Case (iii):} The wedge of each point in $S_2$ covers exactly one point of $S_1$. We may assume that this condition also holds for the wedges of $S_1$; that is, the wedge of each point in $S_1$ covers exactly one point of $S_2$. Since, otherwise, we can simply interchange the set names. 
It follows that each point of $S_2$ lies in its own private region among the regions $R_a$, $R_b$, and $R_c$. 

Let $x$ be the point that lies in $R_a$. We claim that $\connected{x}{a}$.
Assume that $\notconnected{x}{a}$. We show that there exists a point $y \in S_2$ that covers two points of $S_1$.
If $\wedge{x}$ covers $b$ (see Figure~\ref{fig:case3}), then $\orientation{\rightray{x}} \in (0,120)$, which implies that $\orientation{\thirdray{x}} \in (240,360)$. Let $y$ be the point of $S_2$ such that $\orientation{\rightray{y}}=\orientation{\thirdray{x}}$ and $\orientation{\leftray{y}}=\orientation{\rightray{x}}$. Since $\connected{y}{x}$, we have that $y \in \wedge{x}$ and $y$ lies to the left of $\bisector{x}$, but then $\wedge{y}$ must cover $a$ and $c$ -- contradiction.
If $\wedge{x}$ covers $c$, then $\orientation{\leftray{x}} \in (0,120)$, which implies that $\orientation{\thirdray{x}} \in (120,240)$. Let $y$ be the point of $S_2$ such that $\orientation{\leftray{y}}=\orientation{\thirdray{x}}$ and $\orientation{\rightray{y}}=\orientation{\leftray{x}}$. Since $\connected{y}{x}$, we have that $y \in \wedge{x}$ and $y$ lies to the right of $\bisector{x}$, but then $\wedge{y}$ must cover $a$ and $b$ -- contradiction.

\end{proof}

\begin{lemma}[One clique]\label{lem:one_clique}
Let $S_1=\{a,b,c\}$ and $S_2$ be two triplets of points in the plane and let $\alpha=2\pi/3$.
Assume that the wedges of $S_1$ and, independently, of $S_2$ are oriented according to the proof of 
Claim~\ref{lem:three_pts}, and that the induced graph $G_{S_2}$ is a clique.
Then, the induced graph $G_{S_1\cup S_2}$ is connected.
\end{lemma}

\begin{proof}
If the induced graph $G_{S_1}$ is also a clique, then, by Lemma~\ref{lem:two_cliques}, we are done.
Assume therefore that $G_{S_1}$ is not a clique. 
Let $c'$ be the intersection point of $\leftray{a}$ and $\leftray{c}$ (see Figure~\ref{fig:oneclique}), and consider the wedge $\wedge{c'}$ of orientation $\orientation{\wedge{c'}}=\orientation{\wedge{c}}$ and apex $c'$.
The graph induced by $\{a,b,c'\}$ is a clique, and therefore, by Lemma~\ref{lem:two_cliques}, $\connected{a,b,c'}{S_2}$.
If $\connected{a,b}{S_2}$, then we are done, so assume that $\connected{c'}{S_2}$.
Let $x$ be a point of $S_2$ such that $\connected{x}{c'}$, and assume that $\wedge{x}$ does not cover $c$ (if it does, then $\connected{x}{c}$, since $\wedge{c'} \subseteq \wedge{c}$). Then, $x$ lies above $l$ and $\leftray{x}$ intersects $\segment{cc'}$.
Below we consider the three cases: (i) $\rightray{x}$ intersects $\segment{bc'}$, (ii) $\rightray{x}$ intersects $l$ to the left of $b$, and (iii) $\rightray{x}$ does not intersect $l$. However, in the first case (i.e., Case~(i)) and in sub-cases (1) and (2) of the second case (i.e., Case~(ii)(1) and Case~(ii)(2)) we refrain from using the assumption that $G_{S_2}$ is a clique. This is because these cases appear again later in the proof of Theorem\mbox{~\ref{thm:no_cliques}}, where we may not assume that $G_{S_2}$ is a clique.

\begin{figure}[htb]
 \centering 
 \subfigure[Case (i)]{
    \centering
        \includegraphics[width=0.30\textwidth]{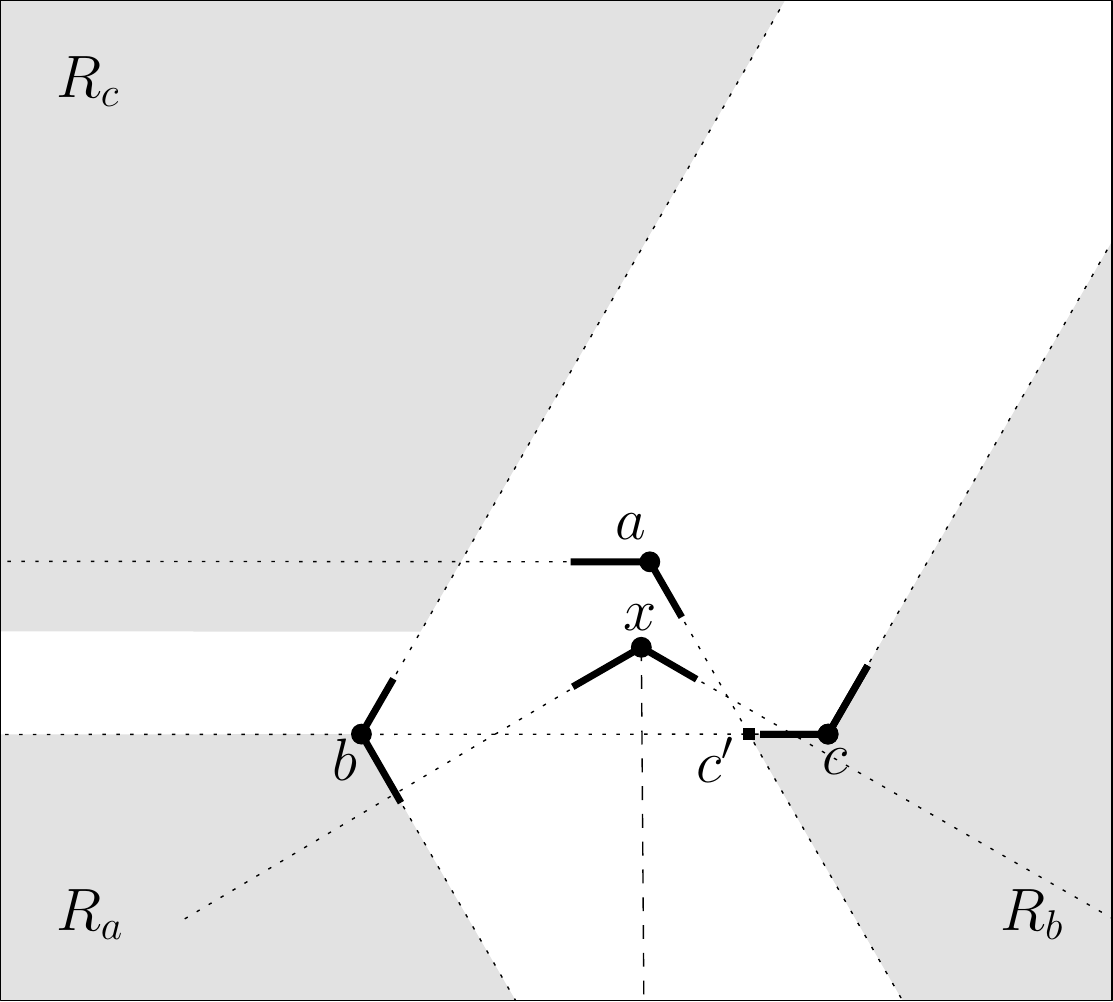}
	     \label{fig:oneclique_case_1_a}}
 \subfigure[Case (ii)]{
    \centering
        \includegraphics[width=0.30\textwidth]{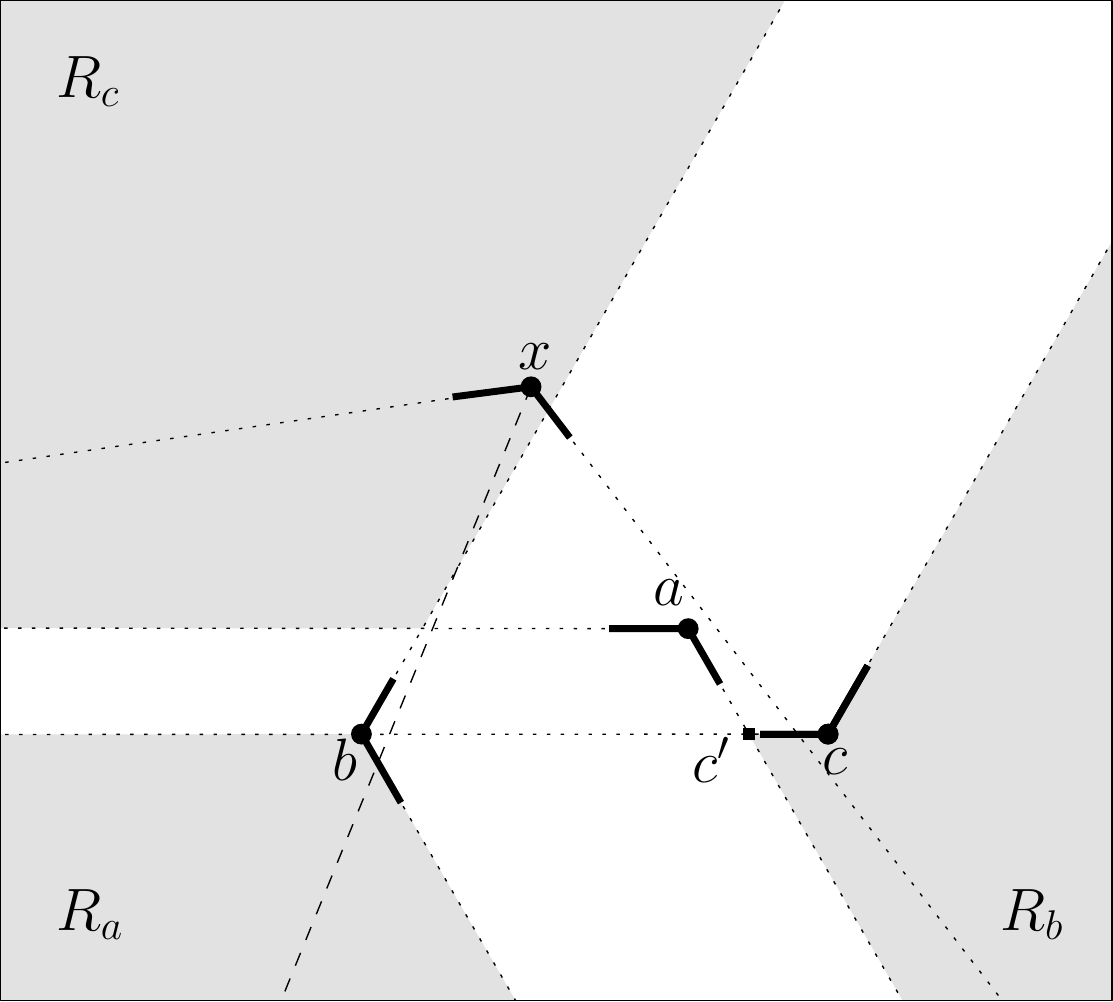}
	     \label{fig:oneclique_case_1_b}}
 \subfigure[Case (iii)]{
    \centering
        \includegraphics[width=0.30\textwidth]{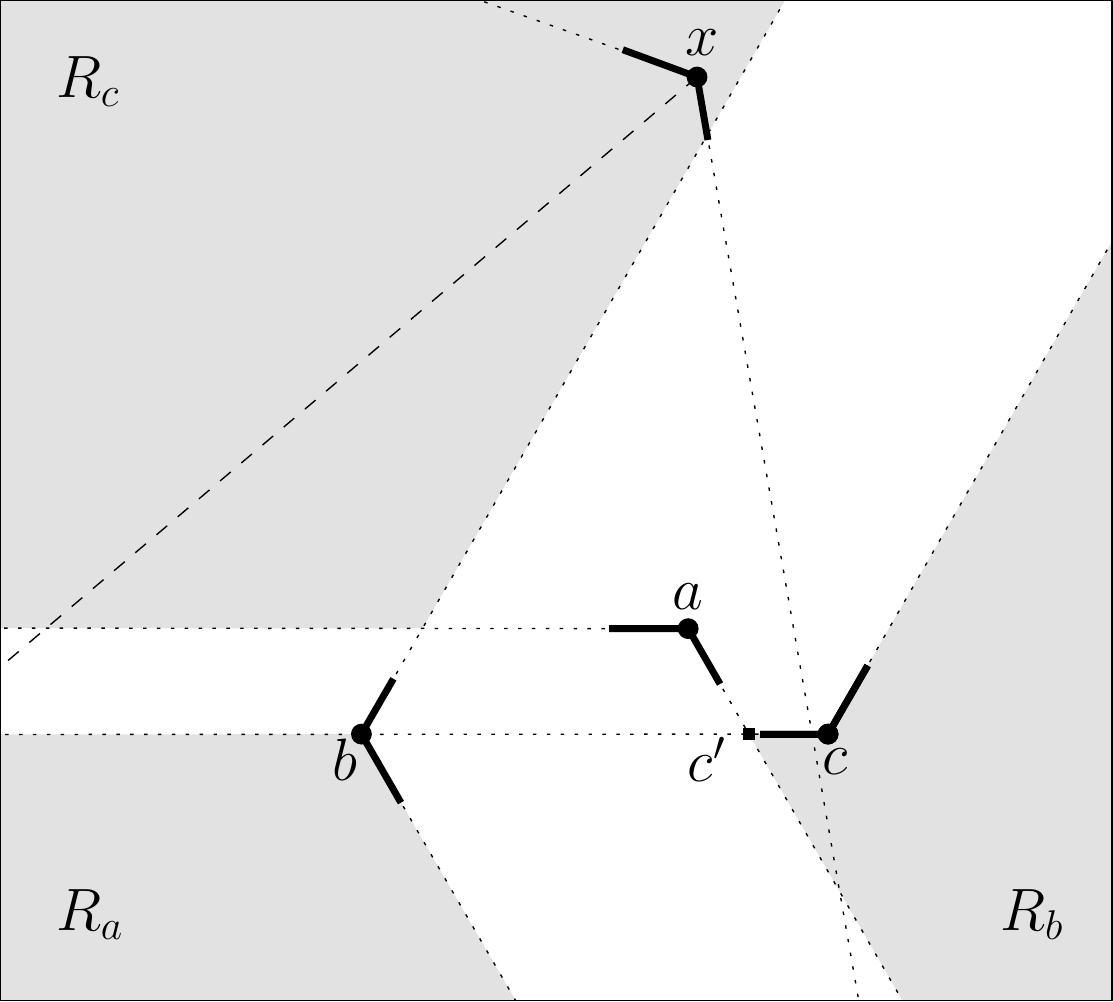}
	     \label{fig:oneclique_case_1_c}}
	\caption{Proof of Lemma~\ref{lem:one_clique}.}	\label{fig:oneclique}
\end{figure}

{\bf Case (i):} $\rightray{x}$ intersects $\segment{bc'}$ (see Figure~\ref{fig:oneclique_case_1_a}). Notice that in this case $\wedge{x}$ does not cover points $b$ and $c$. Since $\orientation{\leftray{x}} < 360$ and $\orientation{\rightray{x}} > 180$, we get that $\orientation{\leftray{x}} \in \rf$, $\orientation{\rightray{x}} \in \rd$, and
$\orientation{\thirdray{x}} \in \rb$.
Between the two points in $S_2 \setminus \{x\}$, let $y$ be the one whose wedge covers more points of $S_1$; in case of tie, let $y$ be any one of them. We know that one of $\wedge{y}$'s rays has orientation in $\rb$.
There are five sub-cases:

(1) $\wedge{y}$ covers all points of $S_1$. There must exist a point in $S_1$ that covers $y$, so we are done. 

(2) $\wedge{y}$ covers $b$ and $c$ and does not cover $a$. If $\connected{y}{b,c}$, then we are done. Otherwise, $y \in R_a$. Now, since $\orientation{\leftray{b}} = 60$ and $\wedge{y}$ must cover $b$ and $c$ and avoid $a$, we get that $\orientation{\leftray{y}} \in \ra$. But, this is impossible, since $\ra$ is not among the three relevant ranges mentioned above. 

(3) $\wedge{y}$ covers $a$ and $b$ and does not cover $c$. If $\connected{y}{a,b}$, then we are done. Otherwise, $y \in R_c$. We show that this is impossible.
If $\orientation{\rightray{y}} \in \rb$, then $\orientation{\leftray{y}} \in \rd$, and $a,b \notin \wedge{y}$. And, if $\orientation{\leftray{y}} \in \rb$, then $\orientation{\rightray{y}} \in \rf$, and $\wedge{y}$ must also cover $c$.

(4) $\wedge{y}$ covers $a$ and $c$ and does not cover $b$. This case is analogous to the previous one.

\old{
If $\connected{y}{a,c}$, then we are done. Otherwise, $y \in R_b$.
We show that this is impossible.
If $\orientation{\leftray{y}} \in \rb$ , then $\orientation{\rightray{y}} \in \rf$, and $a \notin \wedge{y}$. And, if $\orientation{\rightray{y}} \in \rb$, then $\orientation{\leftray{y}} \in \rd$, and $\wedge{y}$ must also cover $b$.
}

(5) $\wedge{y}$ covers exactly one point of $S_1$. Therefore, the wedge of each point in $S_2$ covers exactly one point of $S_1$. Since $\wedge{x}$ does not cover points $b$ and $c$, it must cover $a$.
Assume, w.l.o.g., that $\wedge{y}$ covers $c$ and $\wedge{z}$, the wedge of the remaining point, covers $b$. Next, we show that this is impossible.
Indeed, if $\orientation{\rightray{y}} \in \rb$ and $\orientation{\leftray{y}} \in \rd$, then $\wedge{y}$ must also cover $a$ and $b$. 
And, if $\orientation{\rightray{z}} \in \rb$ and $\orientation{\leftray{z}} \in \rd$, then both $y$ and $z$ must lie below $l$. (Since, if $y$ is above $l$, then $\connected{y}{c}$, and, if $z$ is above $l$, then $\connected{z}{b}$). Therefore, $\wedge{y} \cup \wedge{z}$ covers the halfplane above $l$ (see Property~3), and, in particular, at least one of the two wedges covers $a$.

{\bf Case (ii):} $\rightray{x}$ intersects $l$ to the left of $b$ (see Figure~\ref{fig:oneclique_case_1_b}). In this case, as in Case (i), $\orientation{\leftray{x}} \in \rf$, $\orientation{\rightray{x}} \in \rd$, and
$\orientation{\thirdray{x}} \in \rb$. Notice that in this case $b \in \wedge{x}$, so we assume that
$x \notin \wedge{b}$, since otherwise $\connected{x}{b}$. Let $y$ be a point of $S_2$ whose wedge covers
$c$. We distinguish between three sub-cases:

(1) $\wedge{y}$ covers all points of $S_1$. There must exist a point in $S_1$ that covers $y$, so we are done.

(2) $\wedge{y}$ covers exactly two points of $S_1$. If $\wedge{y}$ covers $b$ and $c$ and $\notconnected{y}{b,c}$, then $y \in R_a$ and either $\orientation{\wedge{y}} \in \ra$ or $\orientation{\wedge{y}} \in \rc$.
However, in both cases, $\wedge{y}$ must also cover $a$ -- contradiction. (Since, in the former case, $\leftray{y}$ does not intersect $\leftray{b}$, and in the latter case, $\rightray{y}$ does not intersect $\leftray{a}$.) If $\wedge{y}$ covers $a$ and $c$ and $\notconnected{y}{a,c}$, then $y \in R_b$ and $\orientation{\wedge{y}} \in \rc$. However, in the case, $\wedge{y}$ must also cover $b$ -- contradiction. (Since $\rightray{y}$ passes above $a$ and is directed upwards, and $\leftray{y}$ passes below $c$ and is directed downward.)

(3) $\wedge{y}$ covers exactly one point of $S_2$, namely, $c$. We know that either $\orientation{\wedge{y}} \in \ra$ or $\orientation{\wedge{y}} \in \rc$. In the latter case, $\wedge{y}$ must also cover $b$, which is impossible. In the former case, if $y$ is above $l$, then $\connected{y}{c}$, so $y$ is necessarily below $l$.
Let $z$ be the remaining point. Then, $\orientation{\wedge{z}} \in \rc$. We show below that $\connected{z}{a,b}$. 
Notice first that $\leftray{y}$ separates between $a$ and $c$ and between $b$ and $c$, since $\orientation{\wedge{y}} \in \ra$ and $\wedge{y}$ covers only $c$. Since $G_{S_2}$ is a clique, we know that $\connected{y}{z}$, and therefore $z$ lies to the right of $\bisector{y}$. Clearly, $a$ and $b$ lie to the left of $\rightray{z}$ (whose orientation is in $\rb$), and to the right of $\leftray{z}$ (whose orientation is in $\rd$). In other words, $\wedge{z}$ covers both $a$ and $b$. Notice also that $z \not \in R_c$, since $\bisector{y}$ (whose orientation is in $\ra$) intersects $l$ to the right of $b$, and $z$ lies to the right of $\bisector{y}$. Therefore, either $\wedge{a}$ or $\wedge{b}$ (or both) covers $z$. We conclude that $\connected{z}{a,b}$.

{\bf Case (iii):} $\rightray{x}$ does not intersect $l$, i.e., $\orientation{\rightray{x}} < 180$ (see Figure~\ref{fig:oneclique_case_1_c}). Since $\wedge{x}$ covers $b$, we may assume that $x \not \in \wedge{b}$. Therefore, $\orientation{\leftray{x}} > 240$. We thus have that $\orientation{\rightray{x}} \in \rc$ and $\orientation{\leftray{x}} \in \re$. Notice that $\bisector{x}$ (whose orientation is in $\rd$)
intersects $l$ to the right of $b$. Moreover, $\wedge{x}$ necessarily covers $a$, since $\orientation{\leftray{x}} \in \re$ and $\leftray{x}$ intersects $l$ between $c'$ and $c$. Let $y$ be the point of $S_2$ such that $\orientation{\wedge{y}} \in \rf$.
Since $G_{S_2}$ is a clique, we know that $\connected{x}{y}$, and therefore $y$ lies to the right of $\bisector{x}$. If $y$ is above $l$, then $\connected{y}{c}$. Otherwise, $y$ is below $l$ and in $\wedge{a}$ (since it is to the left of $b$). But then $\connected{y}{a}$, since $\leftray{y}$ passes above $a$ and $\rightray{y}$ is directed downwards.
\end{proof}

\begin{theorem}\label{thm:no_cliques}
Let $S_1=\{a,b,c\}$ and $S_2$ be two triplets of points in the plane and let $\alpha=2\pi/3$.
Assume that the wedges of $S_1$ and, independently, of $S_2$ are oriented according to the proof of 
Claim~\ref{lem:three_pts}.
Then, the induced graph $G_{S_1\cup S_2}$ is connected.
\end{theorem}

\begin{proof}
If one (or both) of the induced graphs $G_{S_1}$, $G_{S_2}$ is a clique, then, by Lemma~\ref{lem:one_clique}, we are done.
Assume therefore that none of them is a clique. 
Let $c'$ be the intersection point of $\leftray{a}$ and $\leftray{c}$, and consider the wedge $\wedge{c'}$ of orientation $\orientation{\wedge{c'}}=\orientation{\wedge{c}}$ and apex $c'$.
The graph induced by $\{a,b,c'\}$ is a clique, and therefore, by Lemma~\ref{lem:one_clique}, $\connected{a,b,c'}{S_2}$.
If $\connected{a,b}{S_2}$, then we are done, so assume that $\connected{c'}{S_2}$.
Let $x$ be a point of $S_2$ such that $\connected{x}{c'}$, and assume that $\notconnected{x}{c}$ (otherwise
we are done). Then, $x$ lies above $l$ and $\leftray{x}$ intersects $\segment{cc'}$.
We distinguish between three cases, as in the proof of Lemma~\ref{lem:one_clique}:
(i) $\rightray{x}$ intersects $\segment{bc'}$, (ii) $\rightray{x}$ intersects $l$ to the left of $b$, and (iii) $\rightray{x}$ does not intersect $l$. 
As mentioned in the proof of Lemma~\ref{lem:one_clique}, our arguments there for Case~(i) and Cases~(ii)(1) and (ii)(2) do not use the extra assumption that $G_{S_2}$ is a clique. Therefore, we can reuse them here. It remains to show that $\connected{S_1}{S_2}$ in Cases~(ii)(3) and~(iii).  

{\bf Case~(ii)(3):} $\wedge{y}$ covers exactly one point of $S_2$, namely, $c$. We know that either $\orientation{\wedge{y}} \in \ra$ or $\orientation{\wedge{y}} \in \rc$. In the latter case, $\wedge{y}$ must also cover $b$, which is impossible. In the former case, if $y$ is above $l$, then $\connected{y}{c}$, so $y$ is necessarily below $l$.
Let $z$ be the remaining point. Then, $\orientation{\wedge{z}} \in \rc$. 
At this point, we would like to show, as in the proof of Lemma~\ref{lem:one_clique}, that $\connected{z}{a,b}$. However, we cannot assume now that $\connected{y}{z}$. So, we first prove that $\connected{y}{z}$, by proving that $\notconnected{x}{y}$, and then we proceed as in the proof of Lemma~\ref{lem:one_clique}.


Thus, our goal now is to prove that $\notconnected{x}{y}$. Let $p$ be the midpoint of $\segment{bc}$, and let $a'$ be the projection of $a$ onto $l$. According to the construction in the proof of Claim~\ref{lem:three_pts}, $a'$ lies somewhere between $p$ and $c$ (not including $c$). Let $o$ be the intersection point of $\leftray{x}$ and $l$. We know that $o$ is somewhere between $c'$ and $c$ (not including $c$). Finally, let $t$ be the intersection point of $\bisector{x}$ and $l$ (see Figure~\ref{fig:no_clique}). We show that $t$ lies to the left of $p$ and therefore also to the left of $a'$. If $t$ is to the left of $b$ (or $t=b$), then this is clear. Assume therefore that $t$ is to the right of $b$, and consider the two triangles $\triangle xto$ and $\triangle xbt$. Recall first that $x$ is above $\leftray{b}$ and notice that it is below $\bisector{c}$ (since, if $x$ were above $\bisector{c}$, then $\connected{x}{c}$). Therefore $\angle xbt > 60$ and the projection of $x$ onto $l$ lies to the left of $p$. Now, in $\triangle xto$, $\angle xot \le 60$ and $\angle txo = 60$, and therefore $|xt| \le |to|$. And, in $\triangle xbt$, $\angle bxt < 60$ and $\angle xbt > 60$, and therefore $|bt| < |xt|$. Together, we get that $|bt| < |to| < |tc|$, so $t$ lies to the left of $p$ and therefore to the left of $a'$.

Since the projection of $x$ onto $l$ lies to the left of $p$ and so does $t$, we have that $a$ lies to the left of $\bisector{x}$. Now, if $\connected{x}{y}$, then $y$ must lie to the right of $\bisector{x}$ and therefore cover $a$, which is impossible. We conclude that $\notconnected{x}{y}$, and therefore $\connected{y}{z}$ (and $\connected{x}{z}$). 

From this point, we continue as in the proof of Lemma~\ref{lem:one_clique}.  
Notice that $\leftray{y}$ separates between $a$ and $c$ and between $b$ and $c$, since $\orientation{\wedge{y}} \in \ra$ and $\wedge{y}$ covers only $c$. Since $\connected{y}{z}$, we know that $z$ lies to the right of $\bisector{y}$. Clearly, $a$ and $b$ lie to the left of $\rightray{z}$ (whose orientation is in $\rb$), and to the right of $\leftray{z}$ (whose orientation is in $\rd$). In other words, $\wedge{z}$ covers both $a$ and $b$. Notice also that $z \not \in R_c$, since $\bisector{y}$ (whose orientation is in $\ra$) intersects $l$ to the right of $b$, and $z$ lies to the right of $\bisector{y}$. Therefore, either $\wedge{a}$ or $\wedge{b}$ (or both) covers $z$. We conclude that $\connected{z}{a,b}$.

{\bf Case~(iii):} $\rightray{x}$ does not intersect $l$, implying that $\orientation{\rightray{x}} < 180$.
Notice that in this case $b \in \wedge{x}$, so we assume that
$x \notin \wedge{b}$, implying that $\orientation{\leftray{x}} > 240$. 
It follows that $\orientation{\wedge{x}} \in \rd$, $\orientation{\leftray{x}} \in \re$, $\orientation{\rightray{x}} \in \rc$, and $\orientation{\thirdray{x}} \in \ra$. 
Notice also that $\bisector{x}$, whose orientation is in $\rd$, intersects $l$ to the left of $b$.

Let $y$ be the point of $S_2$ such that $\orientation{\leftray{y}} \in \rc$ and $\orientation{\rightray{y}} \in \ra$, and let $z$ be the point of $S_2$ such that $\orientation{\leftray{z}} \in \ra$ and $\orientation{\rightray{z}} \in \re$. 
Notice that for $\leftray{x}$ to intersect $l$ to the right of $c'$, $x$ must lie above $l(\leftray{a})$, and, therefore, $\wedge{x}$ covers $a$.

We first show that if $\connected{x}{z}$, then $\connected{S_1}{S_2}$.
Indeed, if $\connected{x}{z}$, then $z$ must lie to the right of $\bisector{x}$. 
If $z$ is above $l$, then $\connected{z}{c}$. Assume, therefore, that $z$ is below $l$. Notice that $\leftray{z}$ intersects $\leftray{b}$ at a point above $x$, implying that $\leftray{z}$ passes above $a$. Moreover, $\rightray{z}$ passes below $a$, since it is directed downwards. It follows that $\wedge{z}$ covers $a$. But, $z \in \wedge{a}$, since $z$ lies to the right of $\bisector{x}$, which intersects $l$ to the left of $b$. We conclude that $\connected{z}{a}$.

Next, we address the most difficult case, in which $\notconnected{x}{z}$.
If $\notconnected{x}{z}$, then necessarily $y$ is connected to both $x$ and $z$. 
Notice that $z$ must lie below $\leftray{x}$. Also, if it is above $l$, then $\connected{z}{c}$. Assume, therefore, that $z$ is below $l$.
Since $\wedge{y}$'s rays are directed upwards and $\connected{y}{z}$, we know that $y$ is below $z$ and therefore also below $l$. 
According to the construction in the proof of Claim~\ref{lem:three_pts}, either $x$ or $z$ lies on $\bisector{y}$, and the angle at this point in $\triangle xyz$ does not exceed the angle at the other point. It follows that the point that lies on $\bisector{y}$ is necessarily $x$. Since, if it were $z$, then $\angle{yzx} \ge 120$, as it contains $\wedge{z}$.


\begin{figure}[htp]
\centering

  \subfigure[]{
   \centering
       \includegraphics[width=0.31\textwidth]{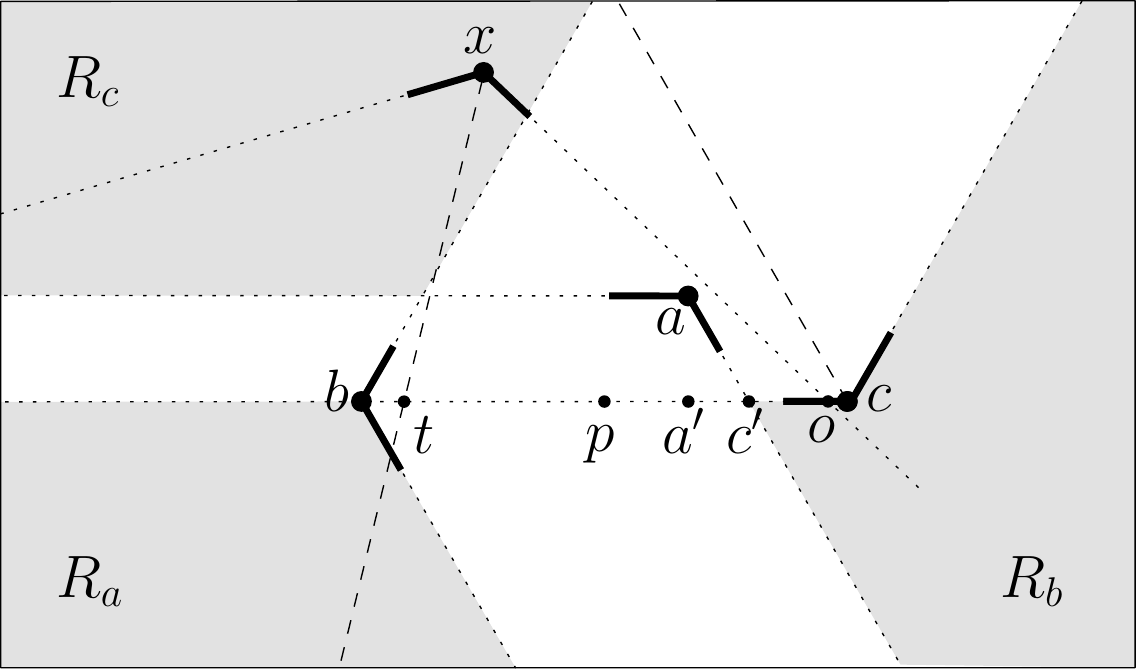}   \label{fig:no_clique}
	}
  \subfigure[]{
   \centering
       \includegraphics[width=0.31\textwidth]{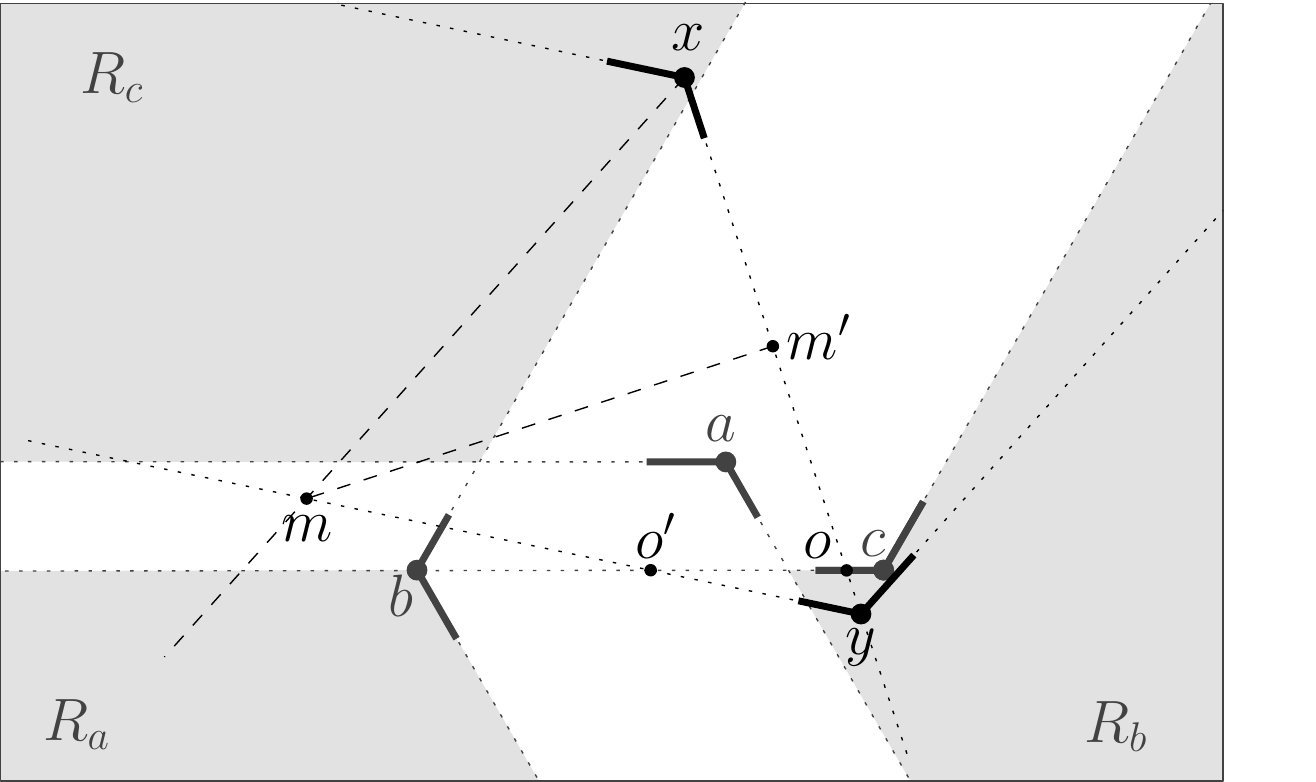}  \label{fig:case3abc}
	}
	\subfigure[]{
		 \centering
     \includegraphics[width=0.31\textwidth]{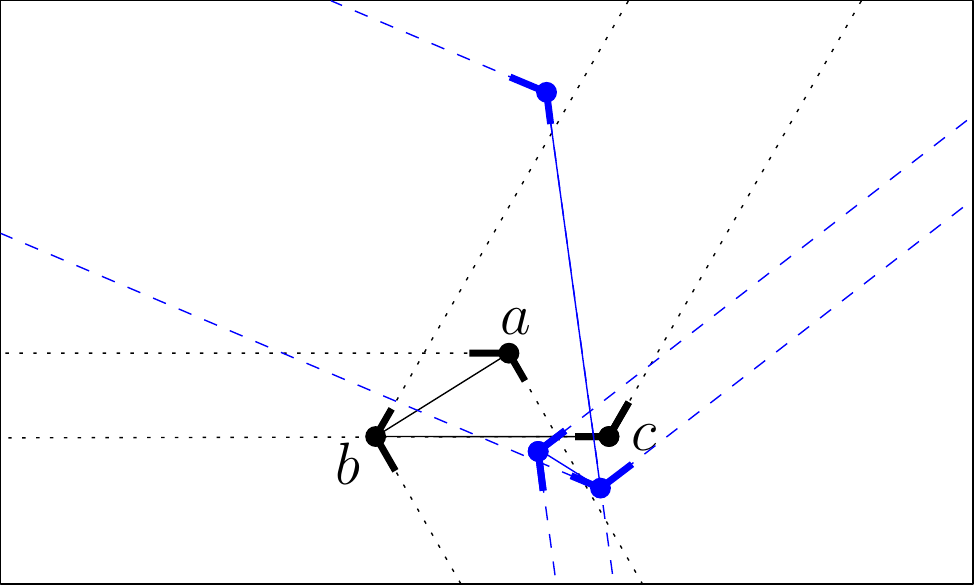}  \label{fig:no_edge}
   }
   \caption{(a) $t$ lies to the left of $a'$. (b) $x$ lies on $\bisector{y}$. (c)
	Each of the triplets induces a connected graph and covers the plane, but the graph of their union is not connected.}
\end{figure}

The case where $x$ lies on $\bisector{y}$ is also impossible, as we show below (see Figure~\ref{fig:case3abc}).
If $y \notin \wedge{b}$, then $\connected{y}{a}$, since $\leftray{y}$ is below $l(\rightray{b})$ and $\rightray{y}$ is above $l(\leftray{x})$. Assume, therefore, that $y \in \wedge{b}$ but $\notconnected{y}{b}$.
Let $m$ be the intersection point of $\leftray{y}$ and $\bisector{x}$. Then, $m$ is above $l$ (since otherwise $\connected{y}{b}$). 
Notice that $\triangle{xmy}$ is equilateral, and consider the bisector of $\angle xmy$.
Let $m'$ be the intersection point of this bisector and side $\segment{xy}$. Then, $\segment{mm'}$ is the perpendicular bisector of $\segment{xy}$.

Next, we show that $m'$ lies above $l$.
Let $o$ be the intersection point of $\segment{xy}$ and $l$, and let $o'$ the intersection point of
$\segment{my}$ and $l$.
We show that $|\segment{yo}| < |\segment{xo}|$, implying that $m'$ is somewhere between $o$ and $x$ and thus above $l$.
Consider $\triangle yoo'$. Since $\orientation{\leftray{y}} \in \rc$, we know that $\angle{yo'o}<60$. But $\angle{oyo'}=60$, so we get that $|\segment{oy}| < |\segment{oo'}|$. Now, consider $\triangle xbo$.
$\angle{xbo}>60$ and $\angle{bxo}<60$, and therefore $|\segment{ox}| > |\segment{ob}|$.
It follows that
$|\segment{oy}| < |\segment{oo'}| < |\segment{ob}| < |\segment{ox}|$.

Since all its corners lie above $l$, $\triangle mm'x$ is above $l$. Since $\connected{y}{z}$ and $z$ is below $l$, we have that 
$z \in \triangle yo'o \subseteq \triangle ymm'$, and therefore $z$ is closer to $y$ than to $x$ -- contradiction the construction of Claim~\ref{lem:three_pts}. 
\end{proof}



\paragraph{Remark.} 
Theorem~\ref{thm:no_cliques} above proves that when the wedges of each of the triplets are oriented, independently, according to the construction of Claim~\ref{lem:three_pts}, then there is always an edge between the two triplets. This is not necessarily true for other constructions with similar properties. For example, the wedges of each of the triplets in Figure~\ref{fig:no_edge} form a connected graph and cover the plane, but there is no edge between the triplets. 

\section{Approximating the $\boldsymbol{\alpha}$-MST} \label{sec:tsp_apx}
Let $P$ be a set of $n$ points in the plane.
In this section we consider the problem of computing an $\alpha$-MST of $P$, for $\alpha = \pi, 2\pi/3, \pi/2$. For each of these angles, we devise a constant-factor approximation algorithm. The approximation ratios are actually with respect to the weight of a Euclidean MST, which is a lower bound for the weight of an $\alpha$-MST, for any $\alpha$.
 
Consider the TSP tour $\Pi=e_0,e_1,\ldots,e_{n-1}$ obtained by applying the standard 2-approximation algorithm for metric TSP. This algorithm first duplicates the edges of a MST to obtain an Eulerian tour, and then transforms the Eulerian tour into a TSP tour by introducing shortcuts. Thus, $wt(\Pi) \le 2wt(\mbox{MST})$.  
Each of our approximation algorithms below begins by constructing $\Pi$. It then constructs, using $\Pi$, a connected $\alpha$-graph, i.e., a graph in which, for each node $p$, the angle spanned by the edges adjacent to $p$ is at most $\alpha$. Finally, it construct an $\alpha$-ST from the $\alpha$-graph, whose weight is bounded by $c \cdot wt(\Pi)$, for some constant $c=c(\alpha)$, and thus is a $2c$-approximation of an $\alpha$-MST.

\paragraph{$\boldsymbol{\alpha=\pi}$.}
Observe that any graph of maximum degree two is a $\pi$-graph. In particular, $\Pi$ is a $\pi$-graph, and, by removing an arbitrary edge, we obtain a $\pi$-ST of weight at most $2 wt(\pi\!\,\mbox{-MST})$.

\paragraph{$\boldsymbol{\alpha=2\pi/3}$.}
Assume, for convenience, that $n = 3m$, for some integer $m$.
We partition $P$ into $m$ triplets, by traversing $\Pi$ from an arbitrary point $p \in P$.
That is, each of the triplets consists of three consecutive points along $\Pi$.
Orient the wedges of each triplet, independently, according to Claim~\ref{lem:three_pts}.
By Theorem~\ref{thm:no_cliques}, the graph induced by $P$, denoted here $G_{\alpha}$ (instead of $G_P$), is connected. In particular, for any two consecutive triplets $t, t'$ along $\Pi$, there exists an edge of the graph between a point of $t$ and a point of $t'$.

Next, we construct a $2\pi/3$-ST, $T$, and show that $wt(T) \le 6 \cdot wt(2\pi/3\!\,\mbox{-MST})$.
Initially, $T$ has no edges. For each of the $m$ triplets $t$, add to $T$ any two edges (of the at least two edges) of $G_{\alpha}$ connecting between pairs of points of $t$. We call these edges {\em inner-edges}. Next, for each of the $m$ pairs of consecutive triplets $t, t'$ along $\Pi$ (except for the pair consisting of the `last' triplet and the `first' triplet), add to $T$ any edge (of the at least one edge) of $G_{\alpha}$ connecting between a point of $t$ and a point of $t'$. We call these edges {\em connecting-edges}.
$T$ is connected and has $2n/3$ inner-edges and $n/3 - 1$ connecting-edges, thus the total number of edges is $n-1$, and $T$ is a $2\pi/3$-ST. 

We now bound the weight of $T$. By the triangle inequality, the weight of an edge $(u,v)$ of $T$ does not exceed the weight of the shorter path (in terms of number of edges) in $\Pi$ between $u$ and $v$. We charge the weight of this path for the edge $(u,v)$.
Each edge of $\Pi$ between two points of the same triplet $t$ is charged at most four times. Twice for the two inner-edges chosen for $t$, and twice for the two connecting-edges that connect $t$ to its two adjacent triplets along $\Pi$.
Each edge of $\Pi$ between two consecutive triplets $t, t'$ (except for the edge between the last and first) is charged only once for the corresponding connecting-edge of $T$. 
Thus, each edge of $\Pi$ is charged at most four times, and $wt(T) = \Sigma_{e \in T} |e| \le
4 \Sigma_{e \in \Pi}{|e|}=4 \cdot wt(\Pi) \le 8 \cdot wt(\mbox{MST}) \le 8 \cdot wt(2\pi/3\!\,\mbox{-MST})$.

Next, we improve the approximation ratio. Observe, that there are three possible ways to partition $\Pi$ into $m$ triplets.
In other words, the set of edges of $\Pi$ connecting between the triplets can be either $E_0, E_1$, or $E_2$, where
$E_j = \{ e_i \in E : i = (j \mod 3)\}$, for $0 \le j \le 2$. By the pigeon hole principle, the weight of one of these sets, say $E_2$, is at least $\frac{1}{3} \cdot wt(\Pi)$. We partition $\Pi$ into triplets, such that the set of edges connecting between the triples is $E_2$. Now, each of the edges of $E_2$ (except $e_{n-1}$) is charged exactly once, and each of the edges of $E_0 \cup E_1$ is charged at most four times. Thus, 
$wt(T) \leq wt(E_2) + 4(wt(E_0)+wt(E_1)) = wt(\Pi) + 3(wt(E_0)+wt(E_1)) \leq wt(\Pi) + 3 \cdot \frac{2}{3} wt(\Pi) = 3 \cdot wt(\Pi) \le 6 \cdot wt(\mbox{MST}) \le 6 \cdot wt(2\pi/3\!\,\mbox{-MST})$.

\paragraph{$\boldsymbol{\alpha=\pi/2}$.}
Assume, for convenience, that $n = 8m$, for some integer $m$.
Our construction for $\alpha=\pi/2$ is similar to the one for $\alpha=2\pi/3$, but slightly more complicated. It is based on a basic gadget described by Aschner et al.~\cite{AKM13} for a set $S$ of four points, indicating the locations of four $\pi/2$-wedges. This gadget is presented as the proof for the claim that one can orient the wedges of $S$, such that the induced graph is connected, and the wedges of $S$ cover the plane.
Unfortunately, we cannot claim that two quadruplets, whose wedges are oriented independently, are connected. However, if they are separable by a line, then they are connected, see~\cite{AKM13}. 

We use this latter claim in our construction.
We partition the tour $\Pi$ into $m$ sections, each consisting of 8 consecutive points along $\Pi$. Then, we partition each of the sections into two quadruplets, a left quadruplet consisting of the 4 leftmost points of the section and a right quadruplet consisting of the 4 rightmost points. (Notice that the points of a quadruplet are not necessarily consecutive along $\Pi$.) Thus, in each section, the two quadruplets are separable by a (vertical) line. Now, orient the wedges of each quadruplet, independently, such that their induced graph is connected and the wedges cover the plane. Let $G_{\alpha}$ be the graph induced by $P$. Observe that $G_{\alpha}$ is connected, since, for any two consecutive sections, there exists two quadruplets, one from each section, that are separable by a (vertical) line and thus connected.

Next, we construct the tree $T$ from $G_{\alpha}$. We distinguish between three types of edges. The first type are the {\em inner-edges}, which connect between points of the same quadruplet. For each quadruplet, we pick three such edges that make the quadruplet connected. The second type are the {\em q-connecting-edges}, which connect between quadruplets of the same section. For each section, we pick one such edge. The third type are the {\em s-connecting-edges}, which connect between consecutive sections along $\Pi$. For each pair of consecutive sections along $\Pi$ (except for the pair consisting of the last and first sections), we pick one such edge. Notice that $T$ is a $\pi/2$-ST, since it is connected and it has $n-1$ edges, i.e., $3n/4$ inner-edges, $n/8$ q-connecting-edges, and $n/8 - 1$ s-connecting-edges.

We compute the approximation ratio by charging the edges of $\Pi$. Each edge of $\Pi$ either connects between points of the same section, or between points of consecutive sections. An edge of the former kind is charged at most nine times. Since for a section, we have six inner-edges, one q-connecting-edge, and two s-connecting-edges.
An edge of the latter kind is charged only once. 

As for $\alpha=2\pi/3$, we can choose the subset of edges of $\Pi$ that connect between consecutive sections, so that its weight is at least $\frac{1}{8} \cdot wt(\Pi)$. Let $E_7$ denote this subset. Then, 
$wt(T) \leq wt(E_7) + 9\cdot wt(E \setminus E_7) \leq wt(\Pi) + 8\cdot wt(E \setminus E_7) \leq wt(\Pi) + 8 \cdot \frac{7}{8} wt(\Pi) = 8 \cdot wt(\Pi) \leq 16\cdot wt(\mbox{MST}) \le 16 \cdot wt(\pi/2\!\,\mbox{-MST})$.
 
The following theorem summarizes the results of this section.
\begin{theorem}
Let $P$ be a set of points in the plane. Then, one can construct (i) a $\pi$-ST
of weight at most $2 \cdot wt(\pi\!\,\mbox{-MST})$, (ii) a $2\pi/3$-ST of weight at most $6 \cdot wt(2\pi/3\!\,\mbox{-MST})$, and
(iii) a $\pi/2$-ST of weight at most $16 \cdot wt(\pi/2\!\,\mbox{-MST})$.
\end{theorem}

\paragraph{Remark.} As mentioned, the approximation ratios above are with respect to $wt(\mbox{MST})$, which is a lower bound for $wt(\alpha\mbox{-MST})$. It is possible that by comparing the weight of the constructed $\alpha$-ST with that of an $\alpha$-MST, one can get better ratios, but it is not clear how to do so.
Moreover, it is easy to see that, for $\alpha \in [60,180)$, 2 is a lower bound on the ratio with respect to a MST, e.g., consider $n$ points on a line. And, for $\alpha \in [180,240)$, $\frac{2+\sqrt{3}}{3}\approx1.244$ is a lower bound on the ratio, e.g., consider 3 points at the corners of an equilateral triangle and a fourth point at the center of the circle passing through them. 
Finally, for $\alpha \in [60,90)$, it is easy to give an example where $wt(\alpha\mbox{-MST})/wt(\mbox{MST}) \rightarrow n-1$. Therefore, any construction algorithm for an angle $\alpha$ in this range, should be analyzed with respect to $wt(\alpha\mbox{-MST})$.   
\old{
\paragraph{Remark.} As mentioned, the approximation ratios above are with respect to the weight of a MST, which is a lower bound for the weight of an $\alpha$-MST. It is possible that by comparing the weight of the constructed $\alpha$-ST with that of an $\alpha$-MST, one can get better ratios, but it is not clear how to do so.
Moreover, it is easy to see that, for $\alpha \in [60,180)$, 2 is a lower bound on the ratio with respect to a MST, e.g., consider $n$ points on a line. And, for $\alpha \in [180,240)$, $\frac{2+\sqrt{3}}{3}\approx1.244$ is a lower bound on the ratio, e.g., consider 3 points at the corners of an equilateral triangle and a fourth point at the center of the circle passing through them. 
}

\section{Constant range hop-spanner for $\boldsymbol{\alpha=2\pi/3}$} \label{sec:boundedrange}
In this section we apply Theorem~\ref{thm:no_cliques} to obtain a solution to a problem that arises in wireless communication networks.
Let $P$ be a set of $n$ points in the plane, where each point in $P$ represents a transceiver equipped with an omni-directional antenna. The coverage region of $p$'s antenna is modeled by a disk centered at $p$, and assume that all disks are of radius 1. Then, the resulting communication graph is the {\em unit disk graph} of $P$, denoted $\UDG$. (I.e, there is an edge between points $p$ and $q$ if the distance between them is at most 1.)
As mentioned in the introduction, directional antennas have some advantages over omni-directional antennas and are gaining popularity. The coverage region of a directional antenna of angle $\alpha$ is modeled by a circular sector of angle $\alpha$. 

Assume that $\UDG$ is connected.
Before stating our problem, we need the following definition. 
A graph $G=(P,E)$ is a {\em $c$-hop-spanner} of $\UDG$, for some constant $c$, if for any two points $p,q \in P$, the minimum number of hops between $p$ and $q$ in $G$ is at most $c$ times this number in $\UDG$. That is, for each edge $e=(p,q)$ in $\UDG$, there exists a path in $G$ between $p$ and $q$ consisting of at most $c$ edges. Assume now that one replaces each of the omni-directional antennas by a directional antenna of angle $2\pi/3$. We address the following {\em Antenna Conversion} problem:
Orient the directional antennas and fix a range $\delta = O(1)$, such that the resulting (symmetric) communication graph is a $c$-hop-spanner of $\UDG$, for some constant $c$. I.e., construct a $2\pi/3$-graph, such that the length of its edges is bounded by $\delta$ and it is a $c$-hop-spanner of~$\UDG$.

We show how to construct such a graph with $\delta=7$ and $c=6$, in $O(n \log n)$ time.
We first partition the points of $P$ into connected components (of $\UDG$) of size at most three. This is done greedily. Set $Q=P$. As long as $Q \neq \emptyset$, perform the following step, which finds the next component $C$. Pick any point $a \in Q$, add it to $C$ (which is initially empty), and remove it from $Q$. Now, if $Q \neq \emptyset$ and there exists a point in $Q$ whose distance from $a$ is at most 1, then pick any such point $b \in Q$, add it to $C$, and remove it from $Q$. Finally, if $Q \neq \emptyset$ and there exists a point in $Q$ whose distance to either $a$ or $b$ (or both) is at most 1, then pick any such point $c \in Q$, add it to $C$, and remove it from $Q$. 

\begin{claim}
\label{clm:comp}
Let $C$ be a connected component of size one or two. Then, each of the neighbors of $C$ in $\UDG$ belongs to a component of size three.
\end{claim}
\begin{proof}
Assume that one of the neighbors of $C$ belongs to a component $C'$ of size one or two, i.e., there exists an edge of $\UDG$ between a point in $C$ and a point in $C'$. Moreover, assume, e.g., that $C$ was found before $C'$. Then, in the iteration in which $C$ was found, we would have found a larger component, i.e., with at least one additional point.
\end{proof}

Now, consider the connected components that were found. We first orient the wedges of each connected component of size exactly three, independently, according to the proof of Claim~\ref{lem:three_pts}. Next, for each connected component $C$ of size one or two, let $C'$ be any connected component of size exactly three, such that $C$ has a neighbor in $C'$. Recall that the wedges of $C'$ cover the plane. We orient each of the wedges of $C$ (alternatively, the single wedge of $C$) towards the wedge of $C'$ that covers it. 
Observe that if the length of the edges is not limited, then the $2\pi/3$-graph, $G_\alpha$, that is induced by the wedges of $P$ is connected. Moreover, it is easy to verify that $G_\alpha$ is a $c$-hop-spanner, for $c=5$. However, our goal is to limit the length of the edges without increasing $c$ by much. 

Let $C$ be a component of size one or two.
Then, the edge of $G_\alpha$ connecting between $C$ and $C'$, where $C'$ is the component of size three to which $C$ was connected, is of length at most $4$. Moreover, consider any two components of size three $C'$ and $C''$, such that $C$ has a neighbor both in $C'$ and in $C''$. Then, the edge of $G_\alpha$ connecting between $C'$ and $C''$ is of length at most 7. Finally, the edge of $G_\alpha$ connecting between two neighboring components of size three is of length at most 5. 
Therefore, one can drop all edges of length greater than 7 from $G_\alpha$, without disconnecting it. 

Next, we show that the resulting graph $G$ is a $6$-hop spanner.
Let $(p,q)$ be an edge of $\UDG$. We show that the number of hops between $p$ and $q$ in $G$ is at most $6$. 
If $p,q$ are in the same component of size three, then, clearly, the path between them consists of at most $2$ edges. 
If $p,q$ are in the same component of size two, then, the path between them passes through a single component of size three, and thus consists of at most $4$ edges. 
If $p,q$ are in different components, then, by Claim~\ref{clm:comp}, at least one of them, say $p$, is in a component of size three. If $q$ is also in a component of size three, then the path from $p$ to $q$ consists of at most $5$ edges. Otherwise, the path between them goes from $q$ to some component $C'$ of size three (which is not necessarily $p$'s component) and from there to $p$'s component, and thus consists of at most 6 edges.

\old{
The following claim shows that the final graph $G$ is a $c$-hop-spanner, for $c=6$. 
\begin{claim}\label{clm:hop-spanner}
$G$ is a $6$-hop spanner.
\end{claim}
\begin{proof}
Let $(p,q)$ be an edge of $\UDG$. We show that the number of hops between $p$ and $q$ in $G$ is at most $6$. 
If $p,q$ are in the same component of size three, then, clearly, the path between them consists of at most $2$ edges. 
If $p,q$ are in the same component of size two, then, the path between them passes through a single component of size three, and thus consists of at most $4$ edges. 
If $p,q$ are in different components, then, by Claim~\ref{clm:comp}, at least one of them, say $p$, is in a component of size three. If $q$ is also in a component of size three, then the path from $p$ to $q$ consists of at most $5$ edges. Otherwise, the path between them goes from $q$ to some component $C'$ of size three (which is not necessarily $p$'s component) and from there to $p$'s component, and thus consists of at most 6 edges.
\end{proof}
}

The following theorem summarizes the result of this section. 
\begin{theorem}
\label{thm:hop-spanner}
Let $P$ be a set of points in the plane and assume that $\UDG$ is connected. Let $\alpha=2\pi/3$. Then, one can construct, in $O(n\log n)$ time, a 6-hop-spanner of $\UDG$, in which each edge is of length at most 7.
\end{theorem}

\paragraph{Running time.}
It is possible to implement the algorithm described above in $O(n \log n)$ time, using a data structure presented by Efrat et al.~\cite{EIK01}. This data structure is designed for a given set $Q$ of $n$ points in the plane. It supports queries of the following form: Given a query point $a$, return a point $q \in Q$ whose distance from $a$ is at most 1, and also delete it from $Q$ (if requested to do so). The data structure can be constructed in $O(n \log n)$ time, and a query (including deletion if needed) can be answered in amortized $O(\log n)$ time. We use it in both phases of the algorithm. For the first phase, in which $P$ is partitioned into connected components of $\UDG$ of size at most three, we construct the data structure for the set $P$. Now, finding a single component requires at most three queries (plus deletions), and, since there are $O(n)$ components, the total running time of this phase is $O(n \log n)$. In the second phase, we orient the wedges of each component. Orienting the wedges of the components of size three can be done in $O(1)$ time per component. For the components of size one or two, we construct the data structure for the subset of $P$ consisting of all points belonging to components of size three. Now, by Claim~\ref{clm:comp}, orienting the wedge of a component of size one (alternatively, the wedges of a component of size two) can be done in amortized $O(\log n)$ time. For a component of size one, we perform a single query (without deletion) in the data structure, and, for a component of size two, we perform one or two queries (without deletion), depending on whether the query with the first point is successful or not.  
We conclude that the overall running time of the algorithm for constructing $G$ is $O(n \log n)$.

\old{
\paragraph{Remark.} If our goal is solely to limit one of the measures (i.e., either range or hop distance), then better constants can be easily obtained. E.g., a $3$-hop spanner or an $\alpha$-graph with maximum length $5$.
}

\section{NP-hardness}\label{sec:np_hardness}

We prove that the problem of computing an $\alpha$-MST, for $\alpha=\pi$ and $\alpha=2\pi/3$, is NP-hard.

\subsection{$\boldsymbol{\alpha=\pi}$}

\begin{figure}[htb]
\centering
  \subfigure[]{
   \centering
       \includegraphics[width=0.2\textwidth,page=1]{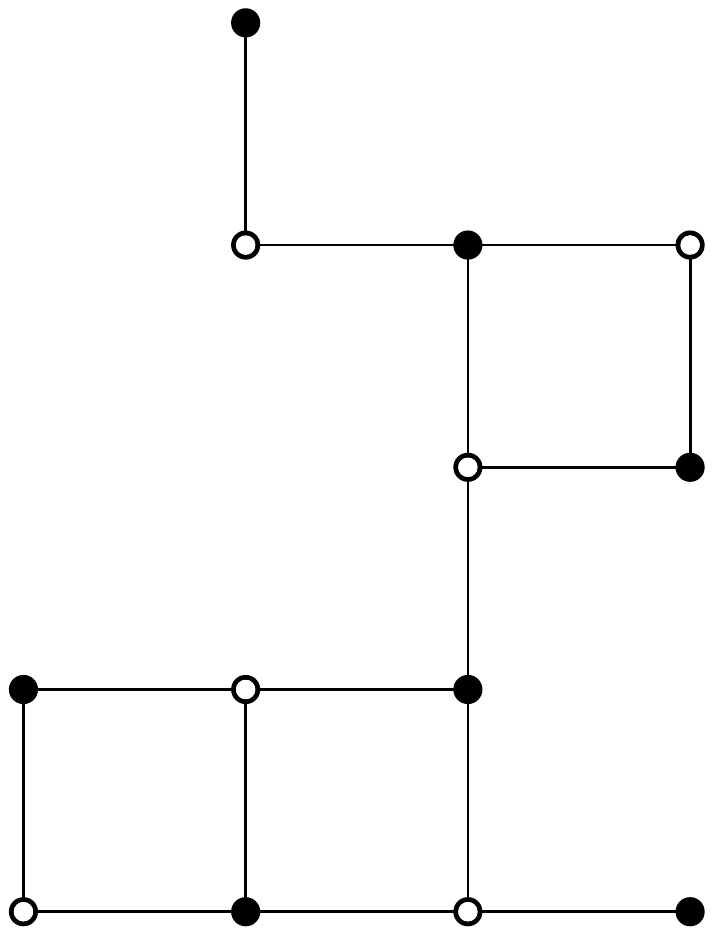}\label{fig:grid_red_a}
  }
  \hspace{.75cm}
  \subfigure[]{
   \centering
       \includegraphics[width=0.2\textwidth,page=2]{fig/grid_graph_reduction}\label{fig:grid_red_b}
  }
  \hspace{.75cm}
  \subfigure[]{
   \centering
       \includegraphics[width=0.2\textwidth,page=3]{fig/grid_graph_reduction}\label{fig:grid_red_c}
  }
 	\caption{(a) A square grid graph $G=(V,E)$ consisting of $13$ vertices and a 2-coloring of the graph. (b) For each point $v \in V$, we add a corresponding point $q_v$ (denoted by a square). (c) A Hamiltonian path of $G$. The semi-circles indicate the orientation of the wedges at the points of $V$.}
 	\label{fig:reduction180}	
\end{figure} 

We describe a reduction from the problem of finding a Hamiltonian path in square grid graphs of degree at most 3. This problem was shown to be NP-hard by Itai et al.~\cite{IPS82}.
A {\em square grid graph} is a graph whose vertices correspond to points in the plane with integer coordinates, and there is an edge between two vertices if the distance between their corresponding points is 1. 
Let $G=(V,E)$, $|V|=n$, be a square grid graph of degree at most 3. Our reduction is very similar to the one in~\cite{PV84}. Since every square grid graph is bipartite, one can color $G$ with two colors, say, black and white (see Figure~\ref{fig:grid_red_a}). Let $c(v)$ denote $v$'s color in some 2-coloring of $G$, and let $n_1$ and $n_2$, $n_1+n_2=n$, be the number of black and white points (i.e., vertices), respectively. 
We add a set $Q(V)$ of $n$ points as follows. For each point $v \in V$, consider any edge $e$ of the {\em complete} grid graph that is adjacent to $v$ and is missing in $G$. We place a point $q_v$ on $e$, such that $|vq_v|=1/4$, if $c(v)$ is black, and $|vq_v|=1/5$, if $c(v)$ is white (see Figure~\ref{fig:grid_red_b}). Notice that the distance from $q_v$ to any point in $V \cup Q(V) \setminus \{v,q_v\}$ is greater than $1$. 
Therefore, any MST of the point set $V \cup Q(V)$ contains the $n$ edges $(v,q(v))$, $v \in V$, and $n-1$ edges from $E$. Its weight is $L = n-1 + n_1/4 + n_2/5$.

\begin{lemma}
\label{lem:hardness1}
G has a Hamiltonian path if and only if $V \cup Q(V)$ has a $\pi$-MST of weight $L$. 
\end{lemma}
\begin{proof}
Assume $G$ has a Hamiltonian path, then its weight is $n-1$, since it consists of $n-1$ edges and the weight of each edge in $G$ is 1.
By adding an edge between each point $v \in V$ and its corresponding point $q_v$, we obtain a tree, $T$, of weight $L$. Notice that the degree of each point in $T$ is at most 3. Moreover, $T$ is a $\pi$-ST of $V \cup Q(V)$, since, at each vertex of $T$, one can place a $\pi$-wedge that covers all its neighbors (see Figure~\ref{fig:grid_red_c}). Finally, the weight of $T$ is $L$, and, therefore, $T$ is a $\pi$-MST of $V \cup Q(V)$.

Assume now that $V \cup Q(V)$ has a $\pi$-MST, $T$, of weight $L$. Then, $T$ must contain the $n$ edges $(v,q_v)$, $v \in V$, plus $n-1$ edges from $E$. Moreover, the maximum degree in $T$ is at most 3 (since, a $\pi$-wedge can cover at most 3 orthogonal directions). We conclude that $G$ has a spanning tree of maximum degree at most 2 of weight $n-1$. That is, $G$ has a Hamiltonian path.
\end{proof}

\subsection{$\boldsymbol{\alpha=2\pi/3}$}

We describe a reduction from the problem of finding a Hamiltonian {\em path} in hexagonal grid graphs. 
Consider a tiling of the plane with regular hexagons of side length 1. The vertex set of an {\em hexagonal grid graph} is a subset of the vertices of the tiling, and there is an edge between two vertices of the graph if the distance between them is 1 (see Figure~\ref{fig:reduction120}). 
The problem of finding a Hamiltonian {\em cycle} in such graphs was shown to be NP-hard by Arkin et al.~\cite{AFIMMRPRX09}.
We first show that the path version is also NP-hard.

\begin{figure}[htb]
\centering
  \subfigure[]{
   \centering
       \includegraphics[width=0.3\textwidth,page=2]{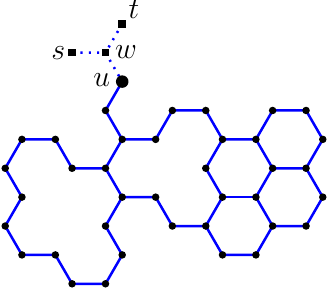}\label{fig:reduction_a}	
  }
	\hspace{.75cm}
  \subfigure[]{
   \centering
       \includegraphics[width=0.3\textwidth,page=2]{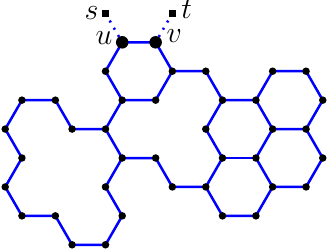}\label{fig:reduction_b}	
  }
 	\caption{(a) If $deg(u)=1$, then we add three points. (b) If $deg(u)=2$, then we add two points.} \label{fig:reduction120}	
\end{figure} 
 
Let $G=(P,E)$ be an hexagonal grid graph, and let $u$ be the highest point in $P$. (If there are several highest points, then pick the leftmost among them.). Notice that since $u$ is the highest point, its degree cannot be 3, i.e., $deg(u) \le 2$. 
Moreover, if $deg(u)=2$, then one of the edges adjacent to $u$ must be horizontal.
We construct another hexagonal grid graph, $G'$, by adding at most three points to $P$, depending on $u$'s degree in $G$.
If $deg(u)=0$, then $G' = G$.
If $deg(u)=1$, then we add the points, $s$, $t$, and $w$ to $P$, as in Figure~\ref{fig:reduction_a}. The only edges that are formed due to this addition are $(s,w)$, $(t,w)$, and $(u,w)$.
Finally, if $deg(u)=2$, then  we add the points $s$ and $t$ to $P$, as in Figure~\ref{fig:reduction_b}. The only edges that are formed due to this addition are $(s,u)$ and $(t,v)$, where $v$ is the horizontal neighbor of $u$.

\begin{lemma}
$G$ contains a Hamiltonian cycle if and only if $G'$ contains a Hamiltonian path.
\end{lemma}
\begin{proof}
Assume that $G$ contains a Hamiltonian cycle. Then, $deg(u)= 2$, and $(u,v)$ is an edge of this cycle. By dropping the edge $(u,v)$, we obtain a Hamiltonian path between $u$ and $v$ in $G$. And, by adding the edges $(s,u)$ and $(t,v)$ to this path, we obtain a Hamiltonian path between $s$ and $t$ in $G'$. 

Assume now that $G'$ contains a Hamiltonian path. We claim that this is possible only if $G'$ was obtained by adding two points to $P$. Indeed, if $G'=G$, then $deg(u)=0$, and $G'$ cannot contain a Hamiltonian path. And, if $G'$ was obtained by adding three points to $P$, then, since $deg(s)=deg(t)=1$ and $s$ and $t$ have a common neighbor, namely, $w$, $G'$ cannot contain a Hamiltonian path. 
So, consider the graph $G'$ that is obtained by adding the points $s$ and $t$ to $P$ (see Figure~\ref{fig:reduction_b}), and recall that we are assuming that $G'$ contains a Hamiltonian path. Since $deg(s)=deg(t)=1$, the endpoints of this path are necessarily $s$ and $t$. By dropping the edges $(s,u)$ and $(t,v)$ from this path, we get a Hamiltonian path $\Pi$ in $G$, between $u$ and $v$. Notice that the edge $(u,v)$ is not in $\Pi$, and by adding it to $\Pi$, we get a Hamiltonian cycle in $G$. 
\end{proof}

We are now ready to show that the problem of computing an $\alpha$-MST, for $\alpha=2\pi/3$, is NP-hard.
Let $G=(P,E)$ be an hexagonal grid graph, where $|P|=n$. Since the distance between any two points in $P$ is at least 1, the weight of a MST of $P$ is at least $n-1$.
\begin{lemma}
$G$ has a Hamiltonian path if and only if $P$ has a $2\pi/3$-MST of weight $n-1$. 
\end{lemma}
\begin{proof}
Assume first that $G$ has a Hamiltonian path. Then, its weight is $n-1$. Moreover, the angle between any two consecutive edges along the path is $2\pi/3$. Therefore, the path is also a $2\pi/3$-ST, and, since its weight is $n-1$, it is a $2\pi/3$-MST.

Assume now that $P$ has a $2\pi/3$-MST, $T$, of weight $n-1$.
Then, all the edges of $T$ are of length exactly 1, and therefore belong also to $E$. It follows that the degree of any point in $T$ is at most 2. Therefore, $T$ is a Hamiltonian path of $G$. 
\end{proof}



\old{
\newpage

\appendix

\section{Proofs of Lemma~\ref{lem:two_cliques} and Lemma~\ref{lem:one_clique}}

\setcounter{lemma}{0}
\begin{lemma}[Two cliques]
Let $S_1=\{a,b,c\}$ and $S_2$ be two triplets of points and let $\alpha=2\pi/3$.
Assume that the wedges (associated with the points) of $S_1$ and, independently, of $S_2$ are oriented according to the proof of Claim~\ref{lem:three_pts}, and that both induced graphs, $G_{S_1}$ and $G_{S_2}$, are cliques.
Then, the induced graph $G_{S_1\cup S_2}$ is connected.
\end{lemma}

\begin{proof}
The wedges of $S_2$ cover the plane, in particular they cover all points of $S_1$. Therefore, we distinguish between three (not necessarily disjoint) cases: 
(i) there exists a point $x \in S_2$ such that $\wedge{x}$ covers all points of $S_1$, 
(ii) there exists a point $x \in S_2$ such that $\wedge{x}$ covers exactly two points of $S_1$, and 
(iii) the wedge of each point in $S_2$ covers exactly one point of $S_1$.

{\bf Case (i):} There exists a point $x \in S_2$ such that $\wedge{x}$ covers all points of $S_1$. Since
the wedges of $S_1$ cover the plane, at least one of them must cover $x$, and therefore
$\connected{x}{S_1}$.

{\bf Case (ii):} There exists a point $x \in S_2$ such that $\wedge{x}$ covers exactly two points of $S_1$. We divide this
case into three sub-cases, according to which two points of $S_1$ are covered by $\wedge{x}$.

\begin{figure}[htb]
 \centering
 \subfigure[$y\in R_a$]{
   \centering
       \includegraphics[width=0.45\textwidth]{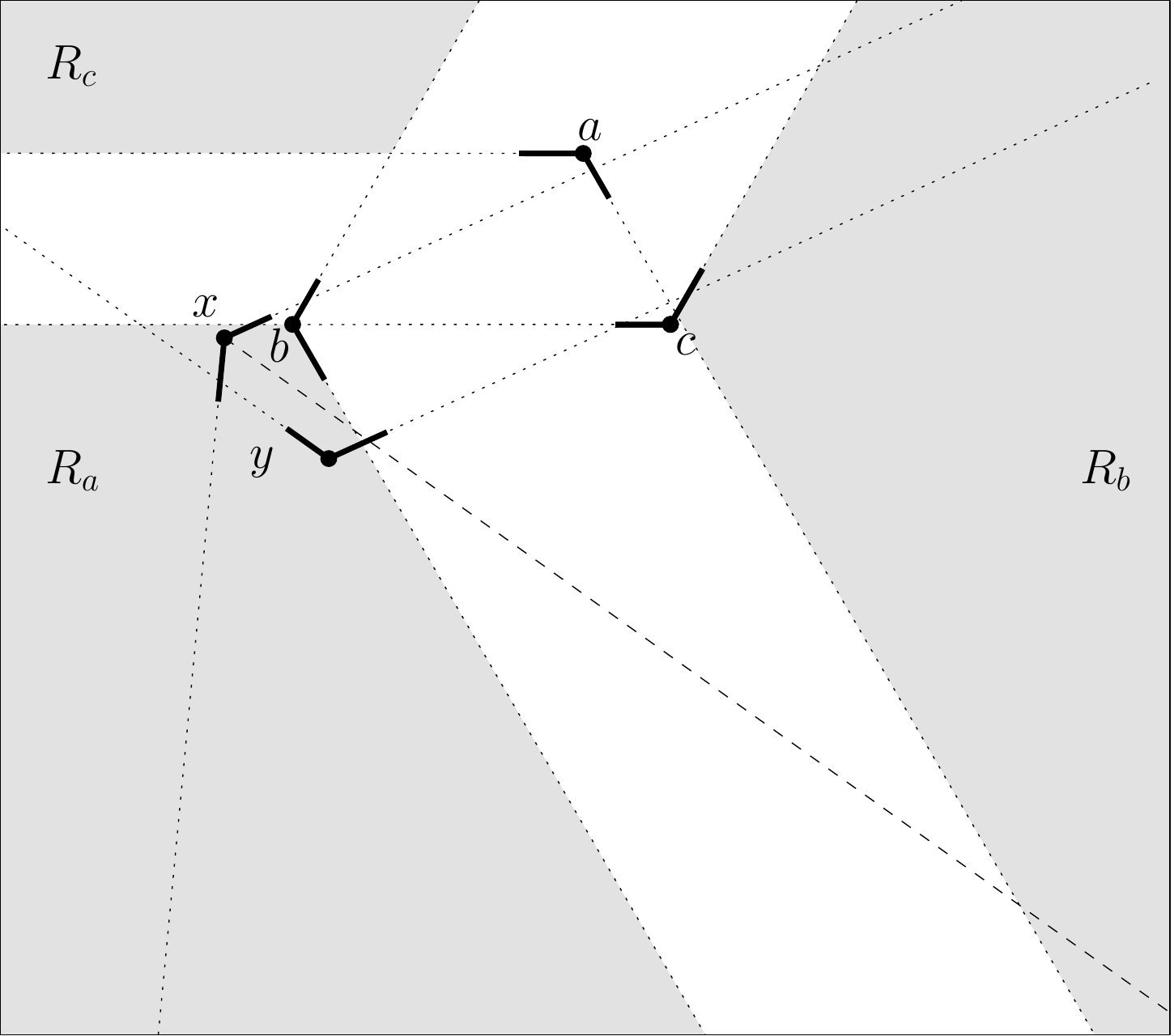}
 }
 \subfigure[$y \in \wedge{b}$]{
  \centering
       \includegraphics[width=0.45\textwidth]{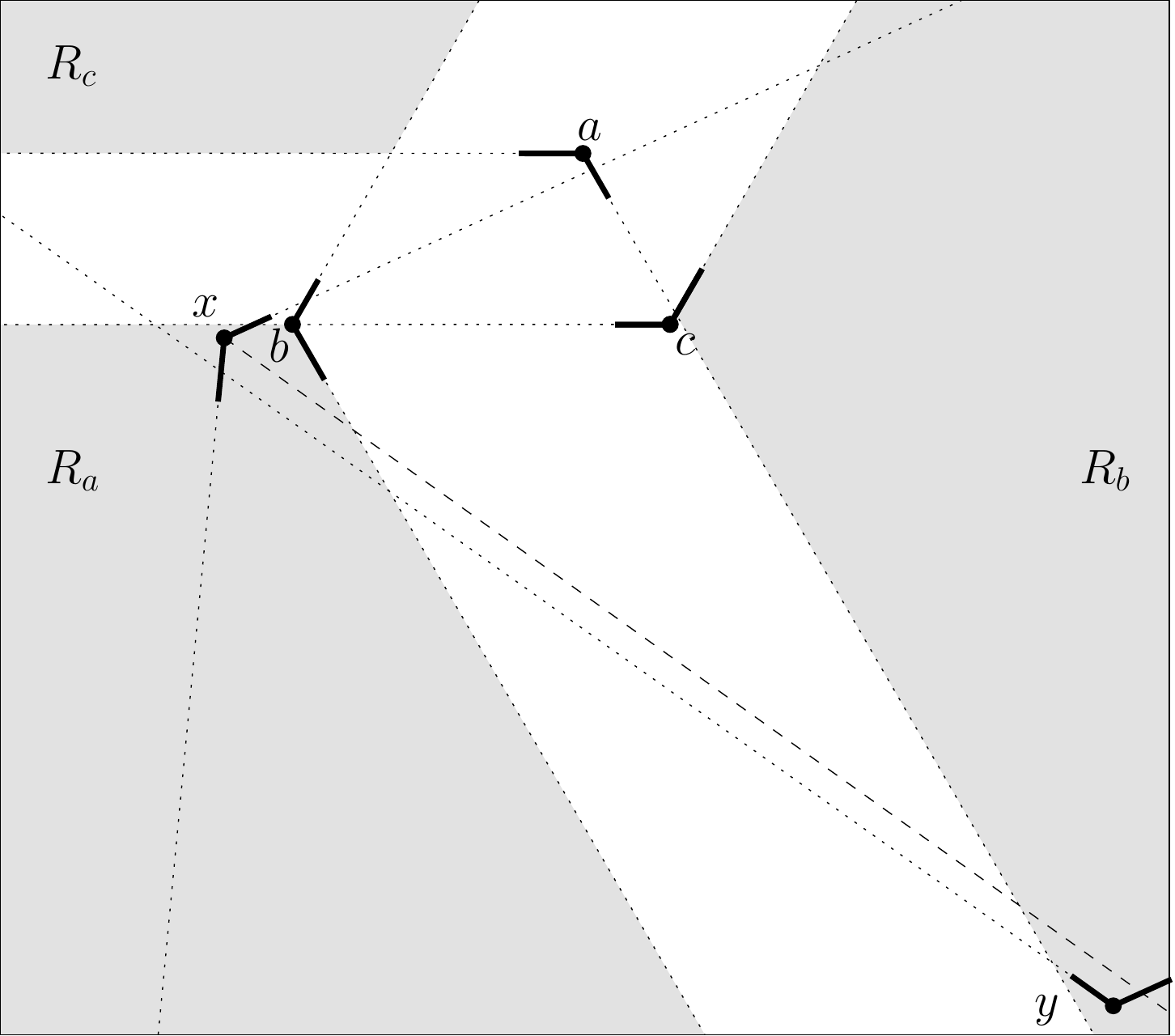}
 }
 \caption{Proof of Lemma~\ref{lem:two_cliques}, Case (ii)(1).}	\label{fig:case2b}	
\end{figure} 

(1) $\wedge{x}$ covers $b$ and $c$ and does not cover $a$. Assume $\notconnected{x}{b,c}$ (since otherwise
we are done), then $x \in R_a$ and one of the rays of $\wedge{x}$ intersects $\segment{ab}$ and
$\segment{ac}$. Notice that this ray must be $\leftray{x}$ and that $\leftray{x}$ also intersects $\leftray{b}$ (see~Figure~\ref{fig:case2b}).
Since $x$ lies below $l$, $\leftray{x}$ intersects $\leftray{b}$, and $\orientation{\leftray{b}}=60$, we have that $\orientation{\leftray{x}} \in \ra$. 
It follows that $\orientation{\rightray{x}} \in \re$, $\orientation{\wedge{x}} \in \rf$, and $\orientation{\thirdray{x}} \in \rc$.
Therefore, $\bisector{x}$ (whose orientation is $\orientation{\wedge{x}}$) does not intersect $l$.
Let $y$ be the point of $S_2$ such that $\orientation{\leftray{y}}= \orientation{\thirdray{x}} \in \rc$ and 
$\orientation{\rightray{y}} = \orientation{\leftray{x}} \in \ra$. Since $\connected{x}{y}$, we have that
$y \in \wedge{x}$ and $y$ lies to the right of $\bisector{x}$. Notice that $\wedge{y}$ contains the (imaginary) wedge of orientation $\orientation{\wedge{y}}$ and apex $x$. If $y \in R_a$ (see Figure~\ref{fig:case2b}(a)), then $\connected{y}{a}$, since $\wedge{y}$ covers $a$. 
Otherwise, $y \in \wedge{b}$ and in particular $y$ lies to the right of $b$ (see Figure~\ref{fig:case2b}(b)). In this case we show that $\connected{y}{b}$. Indeed, $\rightray{y}$ intersects $l$ to the right of $b$, since $\orientation{\rightray{y}}\in\ra$, and, since $l(\leftray{y})$ is parallel to $l(\bisector{x})$ and below it, we have that $\leftray{y}$ intersects $l$ to the left of $b$. We conclude that $b \in \wedge{y}$ and $\connected{y}{b}$.

\begin{figure}[htb]
 \centering
       \includegraphics[width=0.35\textwidth]{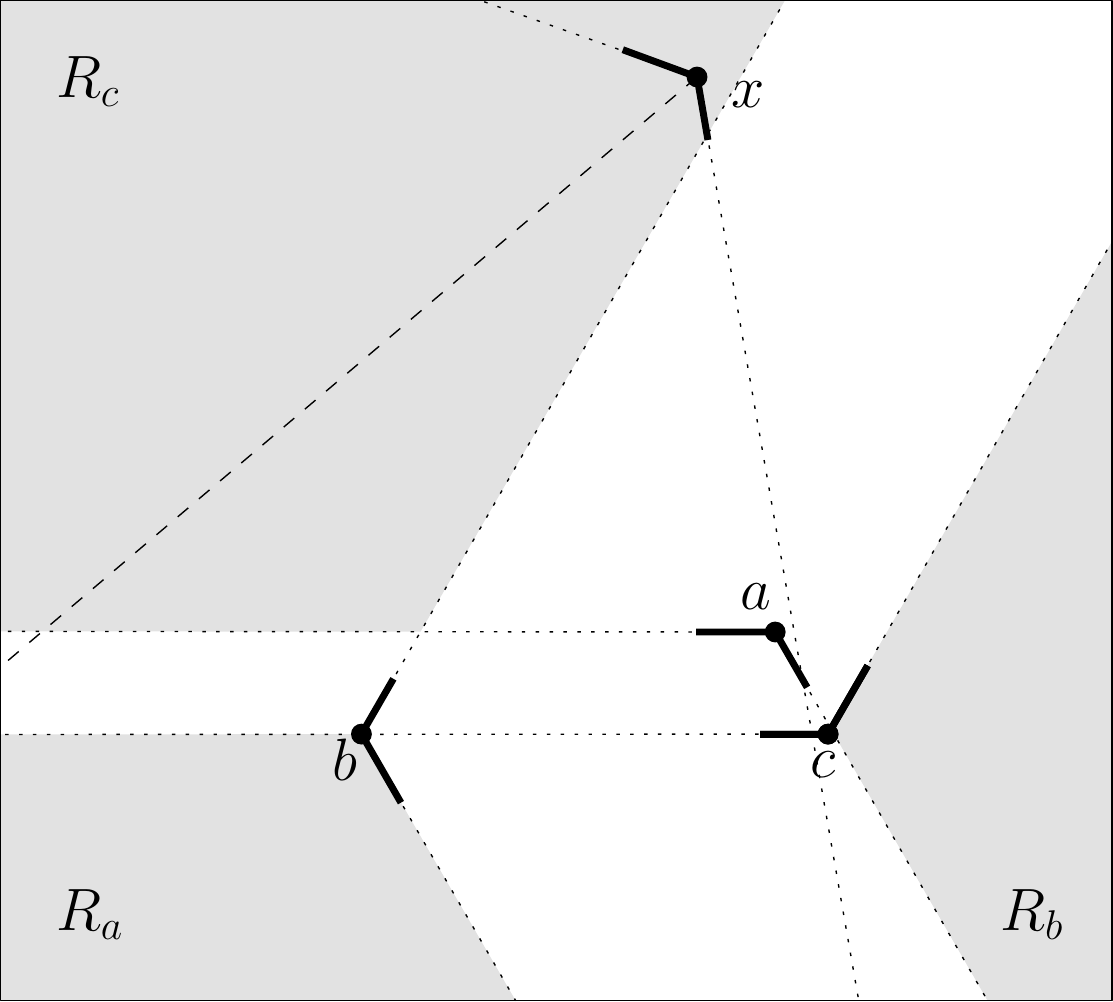}
 \caption{Proof of Lemma~\ref{lem:two_cliques}, Case (ii)(2).}	\label{fig:case2c}	
\end{figure} 
  
(2) $\wedge{x}$ covers $a$ and $b$ and does not cover $c$. Assume $\notconnected{x}{a,b}$ (since otherwise
we are done), then $x \in R_c$
and one of the rays of $\wedge{x}$ intersects $\segment{ac}$ and $\segment{bc}$. Notice that this ray must be $\leftray{x}$ and that $\leftray{x}$ also intersects $\leftray{a}$ (see Figure~\ref{fig:case2c}).
Since $x$ lies above $l$, $\leftray{x}$ intersects $\leftray{a}$, and $\orientation{\leftray{a}}=300$, we have that $\orientation{\leftray{x}} \in \re$. 
It follows that $\orientation{\rightray{x}} \in \rc$, $\orientation{\wedge{x}} \in \rd$, and $\orientation{\thirdray{x}} \in \ra$. The rest of the proof for this case is very similar to the proof of Case~(ii)(1), thus we omit further details.

\begin{figure}[htb]
 \centering 
  \subfigure[$\leftray{x}$ intersects $\segment{ab}$ and $\segment{bc}$]{
    \centering
        \includegraphics[width=0.45\textwidth]{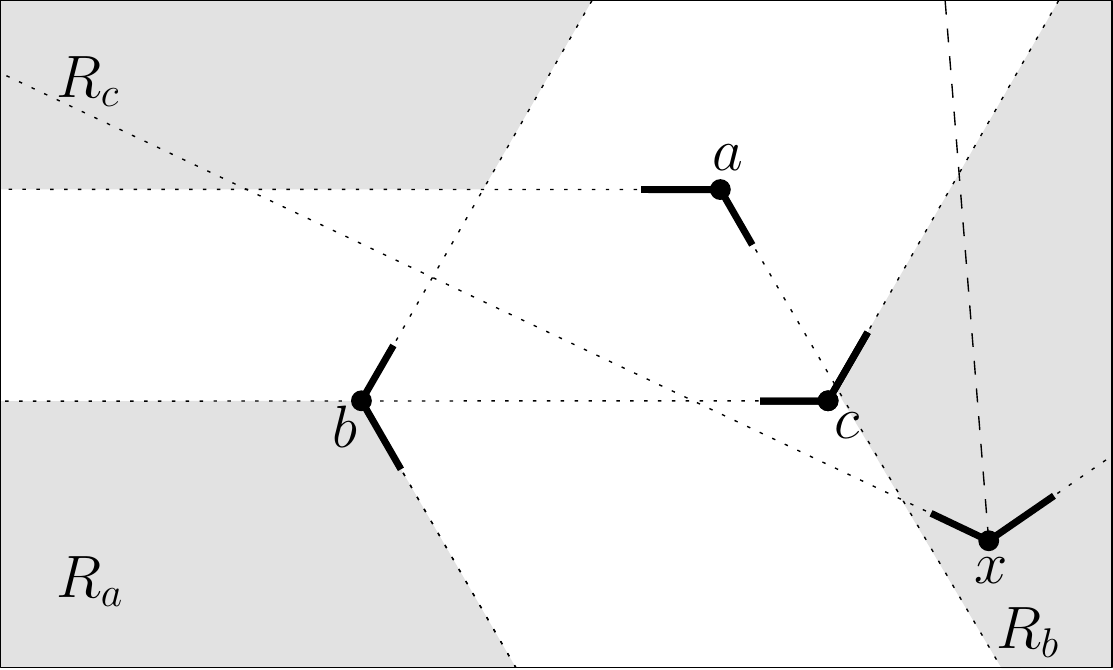}
	\label{fig:case2d}
 }
  \subfigure[$\rightray{x}$ intersects $\segment{ab}$ and $\segment{bc}$]{
    \centering
        \includegraphics[width=0.45\textwidth]{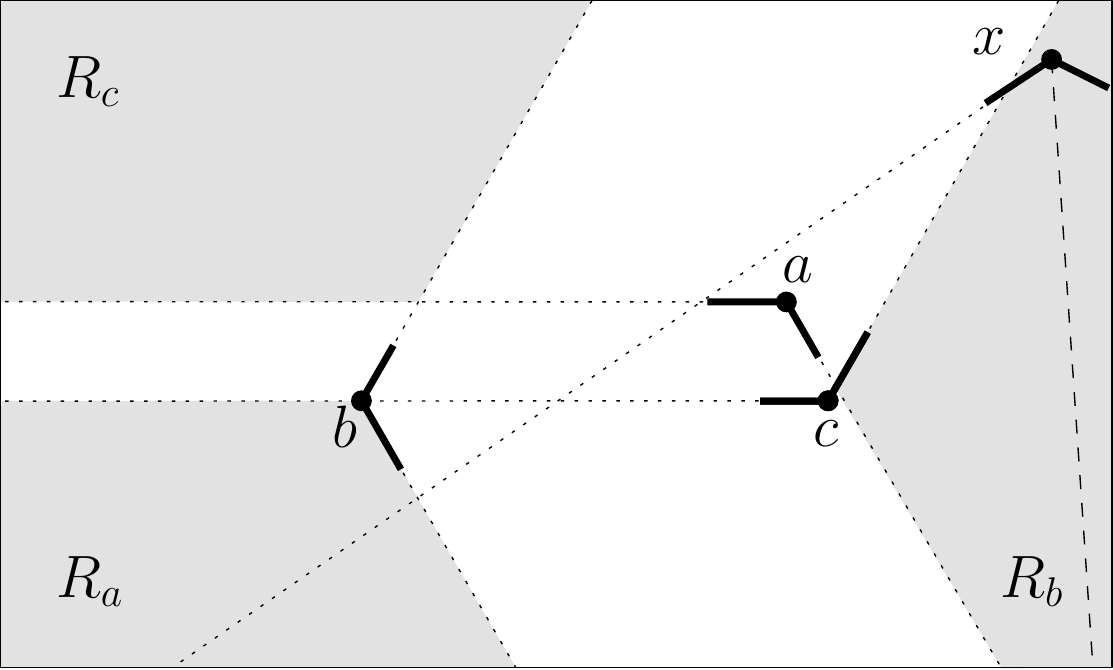}
	\label{fig:case2e}
 }
 \caption{Proof of Lemma~\ref{lem:two_cliques}, Case (ii)(3).}	\label{fig:caseb}
\end{figure}

(3) $\wedge{x}$ covers $a$ and $c$ and does not cover $b$. Assume $\notconnected{x}{a,c}$ (since otherwise
we are done), then $x \in R_b$,
and one of the rays of $\wedge{x}$ intersects $\segment{ab}$ and $\segment{bc}$. Notice that this ray 
can be either $\leftray{x}$ or $\rightray{x}$. 

If it is $\leftray{x}$ (see Figure~\ref{fig:case2d}), then
the orientations associated with $\wedge{x}$ are:
$\orientation{\leftray{x}} \in \rc$, $\orientation{\rightray{x}} \in \ra$, and $\orientation{\thirdray{x}} \in \re$.
The rest of the proof for this branch is very similar to the proof of Case (ii)(1), thus we omit further details.

If the ray intersecting $\segment{ab}$ and $\segment{bc}$ is $\rightray{x}$ (see Figure~\ref{fig:case2e}), then
the orientations associated with $\wedge{x}$ are:
$\orientation{\leftray{x}} \in \rf$, $\orientation{\rightray{x}} \in \rd$, and $\orientation{\thirdray{x}} \in \rb$.
Again, the rest of the proof for this branch is very similar to the proof of Case (ii)(1), thus we omit further details.

\begin{figure}[htb]
 \centering 
    \includegraphics[width=0.5\textwidth]{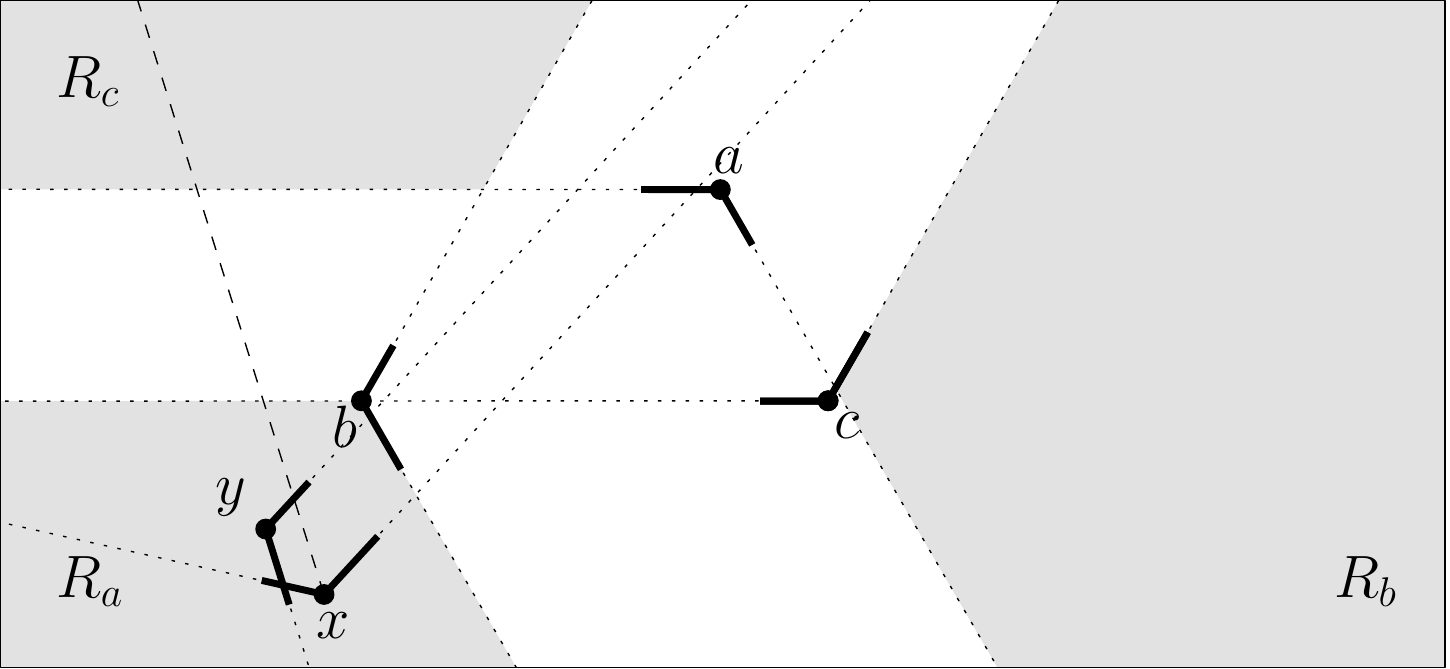}
 \label{fig:case3a}
 \caption{Proof of Lemma~\ref{lem:two_cliques}, Case (iii).}	\label{fig:case3}
\end{figure}

{\bf Case (iii):} The wedge of each point in $S_2$ covers exactly one point of $S_1$. We may assume that this condition also holds for the wedges of $S_1$; that is, the wedge of each point in $S_1$ covers exactly one point of $S_2$. Since, otherwise, we can simply interchange the set names. 
It follows that each point of $S_2$ lies in its own private region among the regions $R_a$, $R_b$, and $R_c$. 

Let $x$ be the point that lies in $R_a$. We claim that $\connected{x}{a}$.
Assume that $\notconnected{x}{a}$. We show that there exists a point $y \in S_2$ that covers two points of $S_1$.
If $\wedge{x}$ covers $b$ (see Figure~\ref{fig:case3}), then $\orientation{\rightray{x}} \in (0,120)$, which implies that $\orientation{\thirdray{x}} \in (240,360)$. Let $y$ be the point of $S_2$ such that $\orientation{\rightray{y}}=\orientation{\thirdray{x}}$ and $\orientation{\leftray{y}}=\orientation{\rightray{x}}$. Since $\connected{y}{x}$, we have that $y \in \wedge{x}$ and $y$ lies to the left of $\bisector{x}$, but then $\wedge{y}$ must cover $a$ and $c$ -- contradiction.
If $\wedge{x}$ covers $c$, then $\orientation{\leftray{x}} \in (0,120)$, which implies that $\orientation{\thirdray{x}} \in (120,240)$. Let $y$ be the point of $S_2$ such that $\orientation{\leftray{y}}=\orientation{\thirdray{x}}$ and $\orientation{\rightray{y}}=\orientation{\leftray{x}}$. Since $\connected{y}{x}$, we have that $y \in \wedge{x}$ and $y$ lies to the right of $\bisector{x}$, but then $\wedge{y}$ must cover $a$ and $b$ -- contradiction.

\end{proof}

\begin{lemma}[One clique]
Let $S_1=\{a,b,c\}$ and $S_2$ be two triplets of points and let $\alpha=2\pi/3$.
Assume that the wedges of $S_1$ and, independently, of $S_2$ are oriented according to the proof of 
Claim~\ref{lem:three_pts}, and that the induced graph $G_{S_2}$ is a clique.
Then, the induced graph $G_{S_1\cup S_2}$ is connected.
\end{lemma}

\begin{proof}
If the induced graph $G_{S_1}$ is also a clique, then, by Lemma~\ref{lem:two_cliques}, we are done.
Assume therefore that $G_{S_1}$ is not a clique. 
Let $c'$ be the intersection point of $\leftray{a}$ and $\leftray{c}$ (see Figure~\ref{fig:oneclique}), and consider the wedge $\wedge{c'}$ of orientation $\orientation{\wedge{c'}}=\orientation{\wedge{c}}$ and apex $c'$.
The graph induced by $\{a,b,c'\}$ is a clique, and therefore, by Lemma~\ref{lem:two_cliques}, $\connected{a,b,c'}{S_2}$.
If $\connected{a,b}{S_2}$, then we are done, so assume that $\connected{c'}{S_2}$.
Let $x$ be a point of $S_2$ such that $\connected{x}{c'}$, and assume that $\wedge{x}$ does not cover $c$ (if it does, then $\connected{x}{c}$, since $\wedge{c'} \subseteq \wedge{c}$). Then, $x$ lies above $l$ and $\leftray{x}$ intersects $\segment{cc'}$.
Below we consider the three cases: (i) $\rightray{x}$ intersects $\segment{bc'}$, (ii) $\rightray{x}$ intersects $l$ to the left of $b$, and (iii) $\rightray{x}$ does not intersect $l$. However, in the first case (i.e., Case~(i)) and in sub-cases (1) and (2) of the second case (i.e., Case~(ii)(1) and Case~(ii)(2)) we refrain from using the assumption that $G_{S_2}$ is a clique. This is because these cases appear again later in the proof of Theorem\mbox{~\ref{thm:no_cliques}}, where we may not assume that $G_{S_2}$ is a clique.

\begin{figure}[htb]
 \centering 
 \subfigure[Case (i)]{
    \centering
        \includegraphics[width=0.30\textwidth]{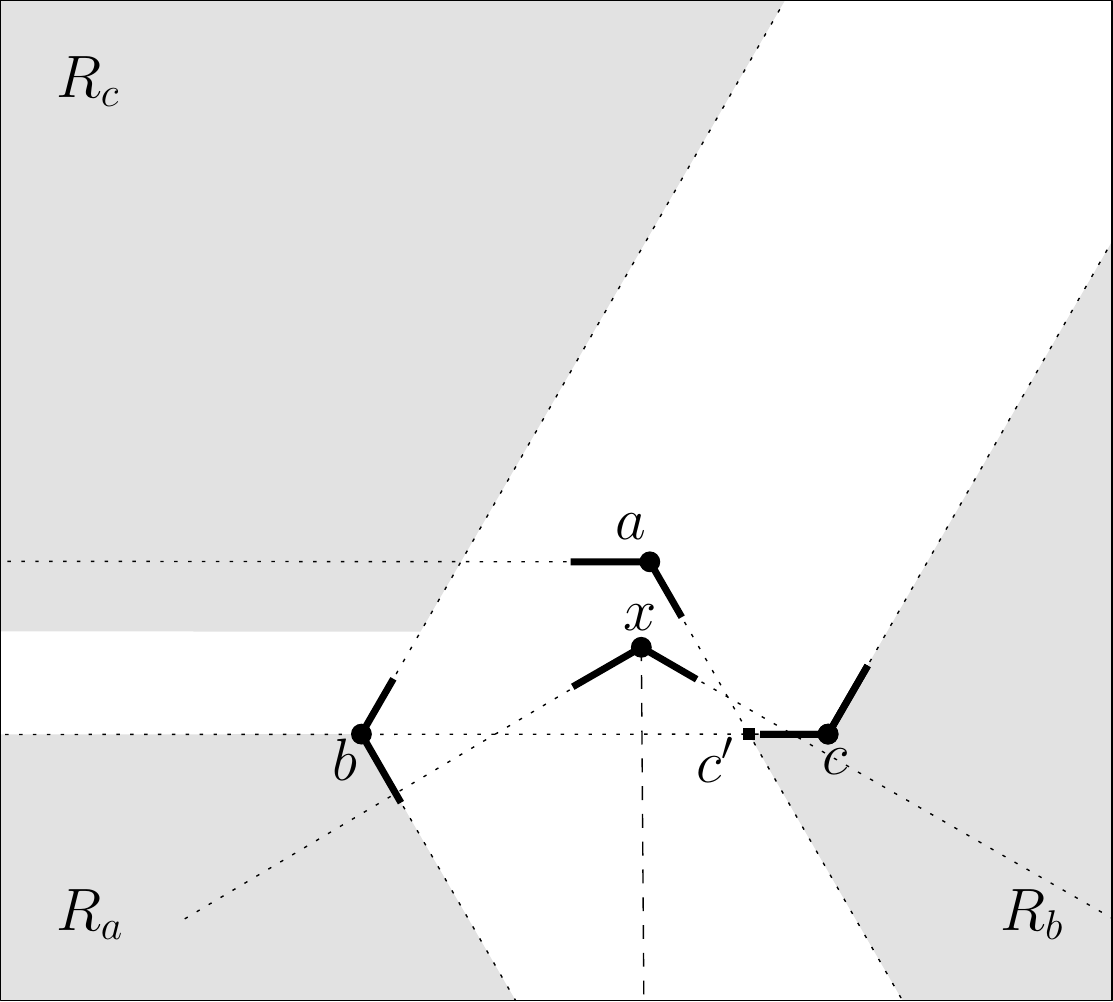}
	     \label{fig:oneclique_case_1_a}}
 \subfigure[Case (ii)]{
    \centering
        \includegraphics[width=0.30\textwidth]{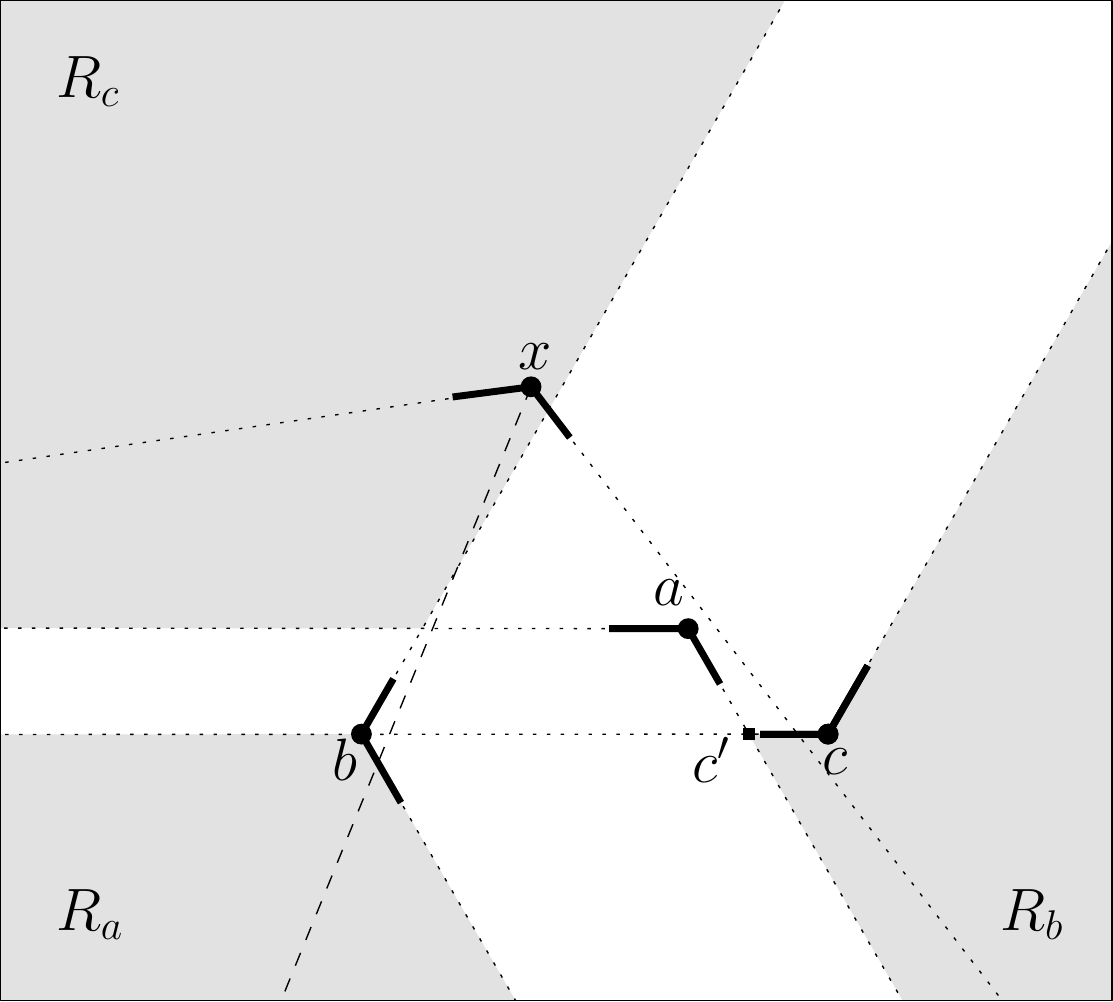}
	     \label{fig:oneclique_case_1_b}}
 \subfigure[Case (iii)]{
    \centering
        \includegraphics[width=0.30\textwidth]{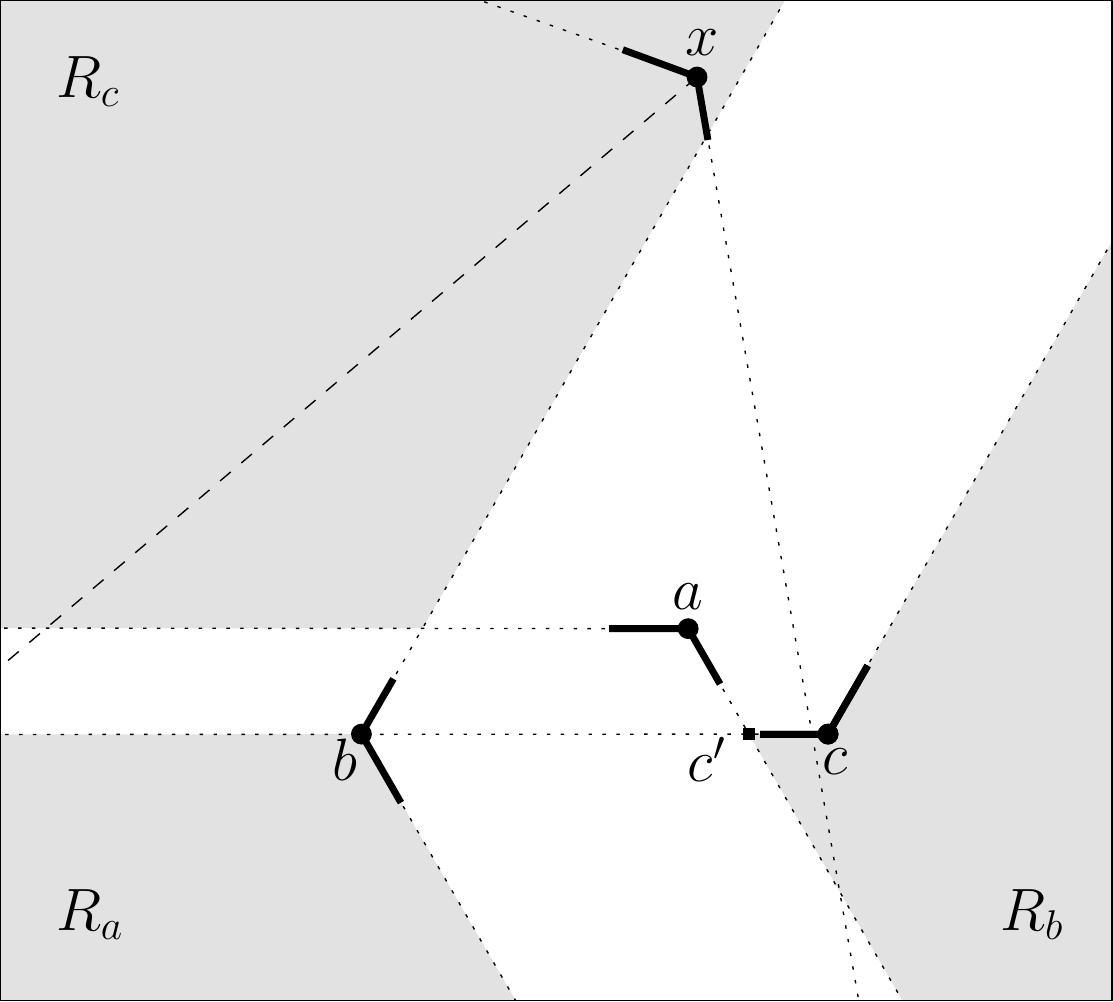}
	     \label{fig:oneclique_case_1_c}}
	\caption{Proof of Lemma~\ref{lem:one_clique}.}	\label{fig:oneclique}
\end{figure}

{\bf Case (i):} $\rightray{x}$ intersects $\segment{bc'}$ (see Figure~\ref{fig:oneclique_case_1_a}). Notice that in this case $\wedge{x}$ does not cover points $b$ and $c$. Since $\orientation{\leftray{x}} < 360$ and $\orientation{\rightray{x}} > 180$, we get that $\orientation{\leftray{x}} \in \rf$, $\orientation{\rightray{x}} \in \rd$, and
$\orientation{\thirdray{x}} \in \rb$.
Between the two points in $S_2 \setminus \{x\}$, let $y$ be the one whose wedge covers more points of $S_1$; in case of tie, let $y$ be any one of them. We know that one of $\wedge{y}$'s rays has orientation in $\rb$.
There are five sub-cases:

(1) $\wedge{y}$ covers all points of $S_1$. There must exist a point in $S_1$ that covers $y$, so we are done. 

(2) $\wedge{y}$ covers $b$ and $c$ and does not cover $a$. If $\connected{y}{b,c}$, then we are done. Otherwise, $y \in R_a$. Now, since $\orientation{\leftray{b}} = 60$ and $\wedge{y}$ must cover $b$ and $c$ and avoid $a$, we get that $\orientation{\leftray{y}} \in \ra$. But, this is impossible, since $\ra$ is not among the three relevant ranges mentioned above. 

(3) $\wedge{y}$ covers $a$ and $b$ and does not cover $c$. If $\connected{y}{a,b}$, then we are done. Otherwise, $y \in R_c$. We show that this is impossible.
If $\orientation{\rightray{y}} \in \rb$, then $\orientation{\leftray{y}} \in \rd$, and $a,b \notin \wedge{y}$. And, if $\orientation{\leftray{y}} \in \rb$, then $\orientation{\rightray{y}} \in \rf$, and $\wedge{y}$ must also cover $c$.

(4) $\wedge{y}$ covers $a$ and $c$ and does not cover $b$. This case is analogous to the previous one.
\old{FULL VERSION
If $\connected{y}{a,c}$, then we are done. Otherwise, $y \in R_b$.
We show that this is impossible.
If $\orientation{\leftray{y}} \in \rb$ , then $\orientation{\rightray{y}} \in \rf$, and $a \notin \wedge{y}$. And, if $\orientation{\rightray{y}} \in \rb$, then $\orientation{\leftray{y}} \in \rd$, and $\wedge{y}$ must also cover $b$.
}

(5) $\wedge{y}$ covers exactly one point of $S_1$. Therefore, the wedge of each point in $S_2$ covers exactly one point of $S_1$. Since $\wedge{x}$ does not cover points $b$ and $c$, it must cover $a$.
Assume, w.l.o.g., that $\wedge{y}$ covers $c$ and $\wedge{z}$, the wedge of the remaining point, covers $b$. Next, we show that this is impossible.
Indeed, if $\orientation{\rightray{y}} \in \rb$ and $\orientation{\leftray{y}} \in \rd$, then $\wedge{y}$ must also cover $a$ and $b$. 
And, if $\orientation{\rightray{z}} \in \rb$ and $\orientation{\leftray{z}} \in \rd$, then both $y$ and $z$ must lie below $l$. (Since, if $y$ is above $l$, then $\connected{y}{c}$, and, if $z$ is above $l$, then $\connected{z}{b}$). Therefore, $\wedge{y} \cup \wedge{z}$ covers the halfplane above $l$ (see Property~3), and, in particular, at least one of the two wedges covers $a$.

{\bf Case (ii):} $\rightray{x}$ intersects $l$ to the left of $b$ (see Figure~\ref{fig:oneclique_case_1_b}). In this case, as in Case (i), $\orientation{\leftray{x}} \in \rf$, $\orientation{\rightray{x}} \in \rd$, and
$\orientation{\thirdray{x}} \in \rb$. Notice that in this case $b \in \wedge{x}$, so we assume that
$x \notin \wedge{b}$, since otherwise $\connected{x}{b}$. Let $y$ be a point of $S_2$ whose wedge covers
$c$. We distinguish between three sub-cases:

(1) $\wedge{y}$ covers all points of $S_1$. There must exist a point in $S_1$ that covers $y$, so we are done.

(2) $\wedge{y}$ covers exactly two points of $S_1$. If $\wedge{y}$ covers $b$ and $c$ and $\notconnected{y}{b,c}$, then $y \in R_a$ and either $\orientation{\wedge{y}} \in \ra$ or $\orientation{\wedge{y}} \in \rc$.
However, in both cases, $\wedge{y}$ must also cover $a$ -- contradiction. (Since, in the former case, $\leftray{y}$ does not intersect $\leftray{b}$, and in the latter case, $\rightray{y}$ does not intersect $\leftray{a}$.) If $\wedge{y}$ covers $a$ and $c$ and $\notconnected{y}{a,c}$, then $y \in R_b$ and $\orientation{\wedge{y}} \in \rc$. However, in the case, $\wedge{y}$ must also cover $b$ -- contradiction. (Since $\rightray{y}$ passes above $a$ and is directed upwards, and $\leftray{y}$ passes below $c$ and is directed downward.)

(3) $\wedge{y}$ covers exactly one point of $S_2$, namely, $c$. We know that either $\orientation{\wedge{y}} \in \ra$ or $\orientation{\wedge{y}} \in \rc$. In the latter case, $\wedge{y}$ must also cover $b$, which is impossible. In the former case, if $y$ is above $l$, then $\connected{y}{c}$, so $y$ is necessarily below $l$.
Let $z$ be the remaining point. Then, $\orientation{\wedge{z}} \in \rc$. We show below that $\connected{z}{a,b}$. 
Notice first that $\leftray{y}$ separates between $a$ and $c$ and between $b$ and $c$, since $\orientation{\wedge{y}} \in \ra$ and $\wedge{y}$ covers only $c$. Since $G_{S_2}$ is a clique, we know that $\connected{y}{z}$, and therefore $z$ lies to the right of $\bisector{y}$. Clearly, $a$ and $b$ lie to the left of $\rightray{z}$ (whose orientation is in $\rb$), and to the right of $\leftray{z}$ (whose orientation is in $\rd$). In other words, $\wedge{z}$ covers both $a$ and $b$. Notice also that $z \not \in R_c$, since $\bisector{y}$ (whose orientation is in $\ra$) intersects $l$ to the right of $b$, and $z$ lies to the right of $\bisector{y}$. Therefore, either $\wedge{a}$ or $\wedge{b}$ (or both) covers $z$. We conclude that $\connected{z}{a,b}$.

{\bf Case (iii):} $\rightray{x}$ does not intersect $l$, i.e., $\orientation{\rightray{x}} < 180$ (see Figure~\ref{fig:oneclique_case_1_c}). Since $\wedge{x}$ covers $b$, we may assume that $x \not \in \wedge{b}$. Therefore, $\orientation{\leftray{x}} > 240$. We thus have that $\orientation{\rightray{x}} \in \rc$ and $\orientation{\leftray{x}} \in \re$. Notice that $\bisector{x}$ (whose orientation is in $\rd$)
intersects $l$ to the right of $b$. Moreover, $\wedge{x}$ necessarily covers $a$, since $\orientation{\leftray{x}} \in \re$ and $\leftray{x}$ intersects $l$ between $c'$ and $c$. Let $y$ be the point of $S_2$ such that $\orientation{\wedge{y}} \in \rf$.
Since $G_{S_2}$ is a clique, we know that $\connected{x}{y}$, and therefore $y$ lies to the right of $\bisector{x}$. If $y$ is above $l$, then $\connected{y}{c}$. Otherwise, $y$ is below $l$ and in $\wedge{a}$ (since it is to the left of $b$). But then $\connected{y}{a}$, since $\leftray{y}$ passes above $a$ and $\rightray{y}$ is directed downwards.
\end{proof}

\vspace{-7mm}
\section{Proof of Theorem~\ref{thm:hop-spanner} --- Run-time analysis}

It is possible to implement the algorithm described in Section~\ref{sec:boundedrange} in $O(n \log n)$ time, using a data structure presented by Efrat et al.~\cite{EIK01}. This data structure is designed for a given set $Q$ of $n$ points in the plane. It supports queries of the following form: Given a query point $a$, return a point $q \in Q$ whose distance from $a$ is at most 1, and also delete it from $Q$ (if requested to do so). The data structure can be constructed in $O(n \log n)$ time, and a query (including deletion if needed) can be answered in amortized $O(\log n)$ time. We use it in both phases of the algorithm.

For the first phase, in which $P$ is partitioned into connected components of $\UDG$ of size at most three, we construct the data structure for the set $P$. Now, finding a single component requires at most three queries (plus deletions), and, since there are $O(n)$ components, the total running time of this phase is $O(n \log n)$.

In the second phase, we orient the wedges of each component. Orienting the wedges of the components of size three can be done in $O(1)$ time per component. For the components of size one or two, we construct the data structure for the subset of $P$ consisting of all points belonging to components of size three. Now, by Claim~\ref{clm:comp}, orienting the wedge of a component of size one (alternatively, the wedges of a component of size two) can be done in amortized $O(\log n)$ time. For a component of size one, we perform a single query (without deletion) in the data structure, and, for a component of size two, we perform one or two queries (without deletion), depending on whether the query with the first point is successful or not.  
We conclude that the overall running time of the algorithm for constructing $G$ is $O(n \log n)$.

\vspace{-3mm}
\section{NP-hardness}

We prove that the problem of computing an $\alpha$-MST, for $\alpha=\pi$ and $\alpha=2\pi/3$, is NP-hard.

\vspace{-3mm}
\subsection{$\boldsymbol{\alpha=\pi}$}

\begin{figure}[htb]
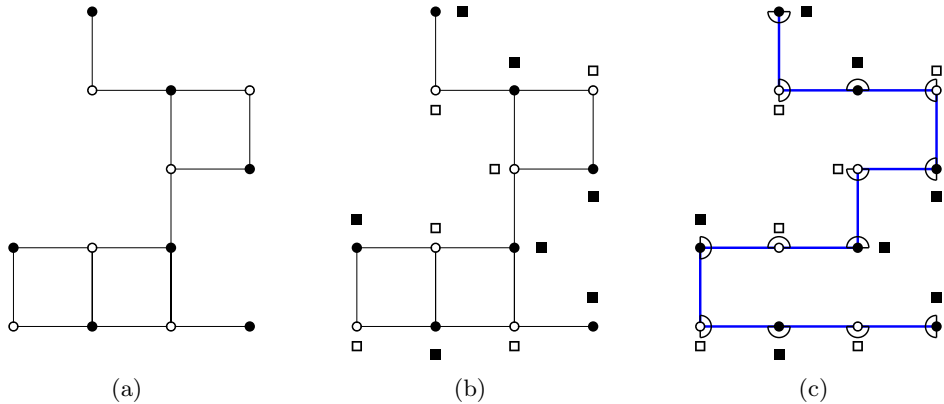

\centering
  \subfigure[]{
   \centering
       \includegraphics[width=0.2\textwidth,page=1]{fig/grid_graph_reduction}\label{fig:grid_red_a}
  }
  \hspace{.75cm}
  \subfigure[]{
   \centering
       \includegraphics[width=0.2\textwidth,page=2]{fig/grid_graph_reduction}\label{fig:grid_red_b}
  }
  \hspace{.75cm}
  \subfigure[]{
   \centering
       \includegraphics[width=0.2\textwidth,page=3]{fig/grid_graph_reduction}\label{fig:grid_red_c}
  }
 	\caption{(a) A square grid graph $G=(V,E)$ consisting of $13$ vertices and a 2-coloring of the graph. (b) For each point $v \in V$, we add a corresponding point $q_v$ (denoted by a square). (c) A Hamiltonian path of $G$. The semi-circles indicate the orientation of the wedges at the points of $V$.}
 	\label{fig:reduction180}	
\end{figure} 

We describe a reduction from the problem of finding a Hamiltonian path in square grid graphs of degree at most 3. This problem was shown to be NP-hard by Itai et al.~\cite{IPS82}.
A {\em square grid graph} is a graph whose vertices correspond to points in the plane with integer coordinates, and there is an edge between two vertices if the distance between their corresponding points is 1. 
Let $G=(V,E)$, $|V|=n$, be a square grid graph of degree at most 3. Our reduction is very similar to the one in~\cite{PV84}. Since every square grid graph is bipartite, one can color $G$ with two colors, say, black and white (see Figure~\ref{fig:grid_red_a}). Let $c(v)$ denote $v$'s color in some 2-coloring of $G$, and let $n_1$ and $n_2$, $n_1+n_2=n$, be the number of black and white points (i.e., vertices), respectively. 
We add a set $Q(V)$ of $n$ points as follows. For each point $v \in V$, consider any edge $e$ of the {\em complete} grid graph that is adjacent to $v$ and is missing in $G$. We place a point $q_v$ on $e$, such that $|vq_v|=1/4$, if $c(v)$ is black, and $|vq_v|=1/5$, if $c(v)$ is white (see Figure~\ref{fig:grid_red_b}). Notice that the distance from $q_v$ to any point in $V \cup Q(V) \setminus \{v,q_v\}$ is greater than $1$. 
Therefore, any MST of the point set $V \cup Q(V)$ contains the $n$ edges $(v,q(v))$, $v \in V$, and $n-1$ edges from $E$. Its weight is $L = n-1 + n_1/4 + n_2/5$.

\begin{lemma}
\label{lem:hardness1}
G has a Hamiltonian path if and only if $V \cup Q(V)$ has a $\pi$-MST of weight $L$. 
\end{lemma}
\begin{proof}
Assume $G$ has a Hamiltonian path, then its weight is $n-1$, since it consists of $n-1$ edges and the weight of each edge in $G$ is 1.
By adding an edge between each point $v \in V$ and its corresponding point $q_v$, we obtain a tree, $T$, of weight $L$. Notice that the degree of each point in $T$ is at most 3. Moreover, $T$ is a $\pi$-ST of $V \cup Q(V)$, since, at each vertex of $T$, one can place a $\pi$-wedge that covers all its neighbors (see Figure~\ref{fig:grid_red_c}). Finally, the weight of $T$ is $L$, and, therefore, $T$ is a $\pi$-MST of $V \cup Q(V)$.

Assume now that $V \cup Q(V)$ has a $\pi$-MST, $T$, of weight $L$. Then, $T$ must contain the $n$ edges $(v,q_v)$, $v \in V$, plus $n-1$ edges from $E$. Moreover, the maximum degree in $T$ is at most 3 (since, a $\pi$-wedge can cover at most 3 orthogonal directions). We conclude that $G$ has a spanning tree of maximum degree at most 2 of weight $n-1$. That is, $G$ has a Hamiltonian path.
\end{proof}

\vspace{-3mm}
\subsection{$\boldsymbol{\alpha=2\pi/3}$}

We describe a reduction from the problem of finding a Hamiltonian {\em path} in hexagonal grid graphs. 
Consider a tiling of the plane with regular hexagons of side length 1. The vertex set of an {\em hexagonal grid graph} is a subset of the vertices of the tiling, and there is an edge between two vertices of the graph if the distance between them is 1 (see Figure~\ref{fig:reduction120}). 
The problem of finding a Hamiltonian {\em cycle} in such graphs was shown to be NP-hard by Arkin et al.~\cite{AFIMMRPRX09}.
We first show that the path version is also NP-hard.

\begin{figure}[htb]
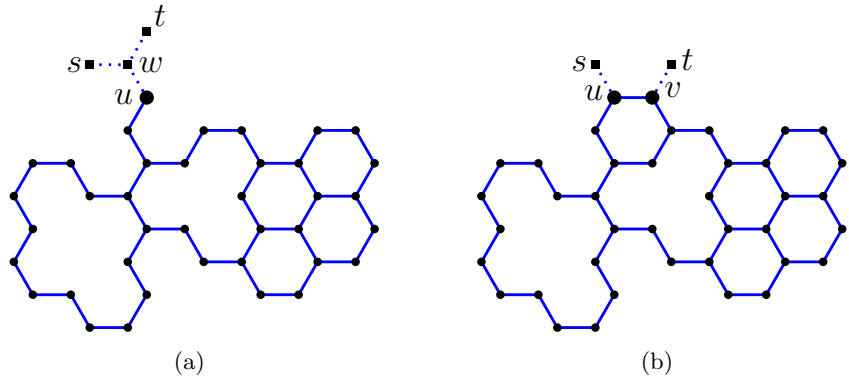

\centering
  \subfigure[]{
   \centering
       \includegraphics[width=0.3\textwidth,page=2]{fig/hex_grid_reduction3}\label{fig:reduction_a}	
  }
	\hspace{.75cm}
  \subfigure[]{
   \centering
       \includegraphics[width=0.3\textwidth,page=2]{fig/hex_grid_reduction2}\label{fig:reduction_b}	
  }
 	\caption{(a) If $deg(u)=1$, then we add three points. (b) If $deg(u)=2$, then we add two points.} \label{fig:reduction120}	
\end{figure} 
 
Let $G=(P,E)$ be an hexagonal grid graph, and let $u$ be the highest point in $P$. (If there are several highest points, then pick the leftmost among them.). Notice that since $u$ is the highest point, its degree cannot be 3, i.e., $deg(u) \le 2$. 
Moreover, if $deg(u)=2$, then one of the edges adjacent to $u$ must be horizontal.
We construct another hexagonal grid graph, $G'$, by adding at most three points to $P$, depending on $u$'s degree in $G$.
If $deg(u)=0$, then $G' = G$.
If $deg(u)=1$, then we add the points, $s$, $t$, and $w$ to $P$, as in Figure~\ref{fig:reduction_a}. The only edges that are formed due to this addition are $(s,w)$, $(t,w)$, and $(u,w)$.
Finally, if $deg(u)=2$, then  we add the points $s$ and $t$ to $P$, as in Figure~\ref{fig:reduction_b}. The only edges that are formed due to this addition are $(s,u)$ and $(t,v)$, where $v$ is the horizontal neighbor of $u$.

\begin{lemma}
$G$ contains a Hamiltonian cycle if and only if $G'$ contains a Hamiltonian path.
\end{lemma}
\begin{proof}
Assume that $G$ contains a Hamiltonian cycle. Then, $deg(u)= 2$, and $(u,v)$ is an edge of this cycle. By dropping the edge $(u,v)$, we obtain a Hamiltonian path between $u$ and $v$ in $G$. And, by adding the edges $(s,u)$ and $(t,v)$ to this path, we obtain a Hamiltonian path between $s$ and $t$ in $G'$. 

Assume now that $G'$ contains a Hamiltonian path. We claim that this is possible only if $G'$ was obtained by adding two points to $P$. Indeed, if $G'=G$, then $deg(u)=0$, and $G'$ cannot contain a Hamiltonian path. And, if $G'$ was obtained by adding three points to $P$, then, since $deg(s)=deg(t)=1$ and $s$ and $t$ have a common neighbor, namely, $w$, $G'$ cannot contain a Hamiltonian path. 
So, consider the graph $G'$ that is obtained by adding the points $s$ and $t$ to $P$ (see Figure~\ref{fig:reduction_b}), and recall that we are assuming that $G'$ contains a Hamiltonian path. Since $deg(s)=deg(t)=1$, the endpoints of this path are necessarily $s$ and $t$. By dropping the edges $(s,u)$ and $(t,v)$ from this path, we get a Hamiltonian path $\Pi$ in $G$, between $u$ and $v$. Notice that the edge $(u,v)$ is not in $\Pi$, and by adding it to $\Pi$, we get a Hamiltonian cycle in $G$. 
\end{proof}

We are now ready to show that the problem of computing an $\alpha$-MST, for $\alpha=2\pi/3$, is NP-hard.
Let $G=(P,E)$ be an hexagonal grid graph, where $|P|=n$. Since the distance between any two points in $P$ is at least 1, the weight of a MST of $P$ is at least $n-1$.
\begin{lemma}
$G$ has a Hamiltonian path if and only if $P$ has a $2\pi/3$-MST of weight $n-1$. 
\end{lemma}
\begin{proof}
Assume first that $G$ has a Hamiltonian path. Then, its weight is $n-1$. Moreover, the angle between any two consecutive edges along the path is $2\pi/3$. Therefore, the path is also a $2\pi/3$-ST, and, since its weight is $n-1$, it is a $2\pi/3$-MST.

Assume now that $P$ has a $2\pi/3$-MST, $T$, of weight $n-1$.
Then, all the edges of $T$ are of length exactly 1, and therefore belong also to $E$. It follows that the degree of any point in $T$ is at most 2. Therefore, $T$ is a Hamiltonian path of $G$. 
\end{proof}
}


\begin{thebibliography}{1}

\bibitem{AGP13}
E. Ackerman, T. Gelander, and R. Pinchasi.
\newblock Ice-creams and wedge graphs.
\newblock {\em Comput. Geom.: Theory \& Applications}, 46(3):213--218, 2013. 

\bibitem{AHHHPSSV13}
O. Aichholzer, T. Hackl, M. Hoffmann, C. Huemer, A. P{\'o}r,
F. Santos, B. Speckmann, and B. Vogtenhuber.
\newblock Maximizing maximal angles for plane straight-line graphs.
\newblock {\em Comput. Geom.: Theory \& Applications}, 46(1):17--28, 2013. 

\bibitem{AFIMMRPRX09}
E. M. Arkin, S. P. Fekete, K. Islam, H. Meijer, J. S. B. Mitchell, Y. N. Rodr\'{\i}guez, V. Polishchuk,  D. Rappaport, and H. Xiao.
\newblock Not being (super) thin or solid is hard: A study of grid Hamiltonicity.
\newblock {\em Comput. Geom.: Theory \& Applications}, 42(6-7):582--605, 2009.

\bibitem{A98}
S. Arora.
\newblock Polynomial time approximation schemes for Euclidean traveling salesman and other geometric problems.
\newblock {\em J. ACM}, 45(5):753--782, 1998.
  
\bibitem{AKM13}
R. Aschner, M. J. Katz, and G. Morgenstern.
\newblock Symmetric connectivity with directional antennas.
\newblock {\em Comput. Geom.: Theory \& Applications}, 46(9):1017--1026, 2013.

\bibitem{BPV09}
I. B{\'a}r{\'a}ny, A. P{\'o}r, and P. Valtr.
\newblock Paths with no small angles.
\newblock {\em SIAM Journal Discrete Mathematics}, 23(4):1655--1666, 2009.

\bibitem{BCDFKM11}
P. Bose, P. Carmi, M. Damian, R. Flatland, M. J. Katz, and A. Maheshwari.
\newblock Switching to directional antennas with constant increase in radius and hop distance.
\newblock In {\em  Proc. 12th Algorithms and Data Structures Sympos.}, pages 134--146, 2011.

\bibitem{CKKKW08}
I. Caragiannis, C. Kaklamanis, E. Kranakis, D. Krizanc, and A. Wiese.
\newblock Communication in wireless networks with directional antennas.
\newblock In {\em 20th ACM Sympos. on Parallelism in Algorithms and Architectures}, pages 344--351, 2008.

  
\bibitem{CKLR11}
P. Carmi, M. J. Katz, Z. Lotker, and A. Ros\'{e}n.
\newblock Connectivity guarantees for wireless networks with directional antennas.
\newblock {\em Comput. Geom.: Theory \& Applications}, 44(9):477--485, 2011.

\bibitem{C04}
T. M. Chan.
\newblock Euclidean bounded-degree spanning tree ratios.
\newblock {\em Discrete {\&} Computational Geometry}, 32(2):177--194, 2004. 

\bibitem{DPT12}
A. Dumitrescu, J. Pach, and G. T{\'o}th.
Drawing Hamiltonian cycles with no large angles.
\newblock {\em Electronic Journal of Combinatorics}, 19(2):P31, 2012.

\bibitem{EIK01}
A. Efrat, A. Itai, and M. J. Katz.
\newblock Geometry helps in bottleneck matching and related problems.
\newblock {\em Algorithmica}, 31(1):1--28, 2001.

\bibitem{FW97}
S. P. Fekete and G. J. Woeginger.
\newblock Angle-restricted tours in the plane.
\newblock {\em Comput. Geom.: Theory \& Applications}, 8:195--218, 1997.




\bibitem{IPS82}
A. Itai, C. H. Papadimitriou, and J. L. Szwarcfiter.
\newblock Hamilton paths in grid graphs.
\newblock {\em SIAM Journal on Computing}, 11(4):676--686, 1982.

\bibitem{JR09}
R. Jothi and B. Raghavachari.
\newblock Degree-bounded minimum spanning trees.
\newblock {\em Discrete Applied Mathematics}, 157(5):960--970, 2009.


\bibitem{KRY96}
S. Khuller, B. Raghavachari, and N. E. Young.
\newblock Low-degree spanning trees of small weight.
\newblock {\em SIAM Journal on Computing}, 25(2):355--368, 1996.

\bibitem{KKM}
E. Kranakis, D. Krizanc, and O. Morales.
\newblock Maintaining connectivity in sensor networks using directional antennae.
In Theoretical Aspects of Distributed Computing in Sensor Networks, Chapter 3, 
S. Nikoletseas and J. D. P. Rolim (Eds.), Springer.
   
\bibitem{M99}
J. S. B. Mitchell.
\newblock Guillotine subdivisions approximate polygonal subdivisions: a simple polynomial-time approximation scheme for geometric TSP, k-MST, and related problems.
\newblock {\em SIAM Journal on Computing}, 28(4):1298--1309, 1999.

\bibitem{PV84}
C. H. Papadimitriou and U. V. Vazirani.
\newblock On two geometric problems related to the travelling salesman problem.
\newblock {\em Journal of Algorithms}, 5(2):231--246, 1984.

\end{thebibliography}
\end{document}